\documentclass[11pt,onecolumn]{article}
\usepackage{amsmath,amsthm,amssymb,amsfonts,epsfig,graphicx}
\usepackage[linesnumbered,ruled,vlined]{algorithm2e}
\usepackage{color}
\usepackage{dsfont,xspace}
\usepackage{relsize,accents}
\usepackage{enumitem}

\usepackage[left=1.0in,top=0.8in,right=1.05in,bottom=0.7in,nohead,textheight=10in,footskip=0.3in]{geometry}

\newtheorem{theorem}{Theorem}
\newtheorem{corollary}[theorem]{Corollary}

\newtheorem{lemma}[theorem]{Lemma}
\newtheorem{definition}[theorem]{Definition}
\newtheorem{proposition}[theorem]{Proposition}

\newtheorem{remark}{Remark}

\newtheorem{theoremRESTATED}{Theorem}
\newtheorem{lemmaRESTATED}{Lemma}
\newtheorem{propositionRESTATED}{Proposition}

\def\RR{{\mathbb R}}
\def\ZZ{{\mathbb Z}}
\def\QQ{{\mathbb Q}}

\newcommand{\Conv}{\mathop{\rm Conv} }
\newcommand{\ZeroExt}[1]{\ensuremath{{\tt 0\mbox{-}Ext}[{#1}]}}

\newcommand{\Pol}[2][]{    \ifthenelse{\equal{#1}{}}{  \ensuremath{\operatorname{Pol}(#2)}    }{    \ensuremath{\operatorname{Pol}^{(#1)}(#2)}    }          }
\newcommand{\fPolplus}[2][]{    \ifthenelse{\equal{#1}{}}{  \ensuremath{\operatorname{supp}(#2)}    }{    \ensuremath{\operatorname{supp}^{(#1)}(#2)}    }          }
\newcommand{\fPol}[2][]{    \ifthenelse{\equal{#1}{}}{  \ensuremath{\operatorname{fPol}(#2)}    }{    \ensuremath{\operatorname{fPol}^{(#1)}(#2)}    }          }

\DeclareMathOperator*{\argmin}{arg\,min}

\newcommand{\UP}{{\hspace{2.5pt} \rotatebox{90}{$\mathsmaller\triangleright$} \hspace{2pt}}}
\newcommand{\DOWN}{{\hspace{2.5pt} \rotatebox{90}{$\mathsmaller\triangleleft$} \hspace{2pt}}}
\newcommand{\DIAMOND}{{\hspace{-2.8pt} \raisebox{-3.25pt}{\rotatebox{45}{\setlength{\fboxrule}{0pt} \setlength{\fboxsep}{1.8pt} \fcolorbox{black}{black}{\null}}} \hspace{2pt}}}
\newcommand{\DELTA}[1]{\ensuremath{{#1}^{\hspace{0.4pt}\mathlarger\diamond}}}
\newcommand{\DELTANEIB}{\ensuremath{\mathlarger{\mathlarger\diamond}\xspace}}

\newcommand{\eqdef}{{\stackrel{\mbox{\tiny \tt ~def~}}{=}}}
\newcommand{\BPrel}{{\stackrel{\mbox{\tiny \tt ~Bp}}{\sqsubseteq}}}

\newcommand{\VCSP}[1]{\ensuremath{\operatorname{VCSP}(#1)}}

\DeclareMathOperator{\dom}{\textup{\texttt{dom}}}
\DeclareMathOperator{\supp}{supp}

\newcommand{\myparagraph}[1]{\noindent{\textbf{#1}}\quad} 

\renewcommand{\.}{\hskip 0.5pt}

\def\calE{{\cal E}}

\def\calG{{\cal G}}
\def\calI{{\cal I}}

\def\calL{{\cal L}}

\def\calS{{\cal S}}

\def\calO{{\cal O}}

\begin{document}

\def\myparagraph#1{\vspace{2pt}\noindent{\bf #1~~}}

%\pagestyle{headings}

%%%%%%%%%%%%%%%%%%%%%%%%%%%%%%%%%%%%%%%%%%%%%%%%%%%%%%%%%%%
%%%%%%%%%%%%%%%%%%%%%%%%%%%%%%%%%%%%%%%%%%%%%%%%%%%%%%%%%%%
%%%%%%%%%%%%%%%%%%%%%%%%%%%%%%%%%%%%%%%%%%%%%%%%%%%%%%%%%%%

%%%%%%%%%%%%%%%%%%%%%%%%%%%%%%%%%%%%%%%%%%%%%%%%%%%%%%%%%%%
%%%%%%%%%%%%%%%%%%%%%%%%%%%%%%%%%%%%%%%%%%%%%%%%%%%%%%%%%%%
%%%%%%%%%%%%%%%%%%%%%%%%%%%%%%%%%%%%%%%%%%%%%%%%%%%%%%%%%%%

\title{\Large\bf  \vspace{0pt} Generalized minimum $0$-extension problem and discrete convexity}
\author{Martin Dvorak \hspace{30pt} Vladimir Kolmogorov \\ \normalsize Institute of Science and Technology Austria \\ {\normalsize\tt $\{$martin.dvorak,vnk$\}$@ist.ac.at}}
\date{}
\maketitle

%\vspace{-6pt}
\begin{abstract}
Given a fixed finite metric space $(V,\mu)$, the {\em minimum $0$-extension problem}, denoted as $\ZeroExt{\mu}$, is equivalent to
the following optimization problem:
minimize function of the form $\min\limits_{x\in V^n} \sum_i f_i(x_i) + \sum_{ij} c_{ij}\.\mu(x_i,x_j)$
where $f_i:V\rightarrow \mathbb R$ are functions given by $f_i(x_i)=\sum_{v\in V} c_{vi}\.\mu(x_i,v)$
and $c_{ij},c_{vi}$  are given nonnegative costs.
The computational complexity of $\ZeroExt{\mu}$ has been recently established by Karzanov and by Hirai:
if metric $\mu$ is {\em orientable modular} then $\ZeroExt{\mu}$ can be solved in polynomial time, otherwise
$\ZeroExt{\mu}$ is NP-hard. To prove the tractability part, Hirai developed a theory of discrete convex functions on orientable modular graphs
generalizing several known classes of functions in discrete convex analysis, such as \hbox{$L^\natural$-convex} functions.

We consider a more general version of the problem in which unary functions $f_i(x_i)$ can additionally have terms
of the form $c_{uv;i}\.\mu(x_i,\{u,v\})$ for $\{u,\!\.\.v\}\in F$, where set $F\subseteq\binom{V}{2}$ is~fixed.
We extend the complexity classification above by providing an explicit condition on $(\mu,F)$ for the problem to be tractable.
In order to prove the tractability part, we generalize Hirai's theory and define a larger class of discrete convex functions.
It covers, in particular, another well-known class of functions, namely submodular functions on an integer lattice.

Finally, we improve the complexity of Hirai's algorithm for solving $\ZeroExt{\mu}$ on orientable modular graphs.
\end{abstract}

%\noindent {\bf Keywords}: minimum 0-extension problem, discrete convexity, submodularity, $L$-convexity, orientable modular graphs.

\section{Introduction}
Consider a metric space $(V,\mu)$ where $V$ is a finite set and $\mu$ is a nonnegative function $V\times V\rightarrow \mathbb R$
satisfying the axioms of a metric: $\mu(x,y)=0$ $\Leftrightarrow$ $x=y$, $\.\mu(x,y)=\mu(y,x)$, $\.\mu(x,y)+\mu(y,z)\ge \mu(x,z)$ for all $x,y,z\in V$.
We study optimization problems of the following form:
\begin{equation}\label{eq:multifacility}
\min_{x\in V^n} f(x),\qquad\quad f(x)=\sum_{i\in[n]} f_i(x_i) + \sum_{1\le i<j\le n}c_{ij} \mu(x_i,x_j)
\end{equation}
where weights $c_{ij}$ are nonnegative. If unary terms $f_i:V\rightarrow\mathbb R$ are allowed to be arbitrary nonnegative functions
then this is a well-studied {\em Metric Labeling Problem}~\cite{KleinbergTardos:02}.
Another important special case is when the unary terms are given by
\begin{equation}\label{eq:multifacility:unaries}
f_i(x_i)=\sum_{v\in V}c_{vi}\mu(x_i,v)
\end{equation}
with nonnegative weights $c_{vi}$. 
This is a classical facility location problem, known as {\em multifacility location problem}~\cite{Tansel:83}.
It can be interpreted as follows:
we are going to locate $n$ new facilities
in~$V$, where the facilities communicate with each other and communicate with
existing facilities in $V$. 
The cost of the communication is propositional to the distance.
The goal is to find a location of minimum total communication cost.
The multifacility location problem is also equivalent to the {\em minimum $0$-extension problem} formulated by Karzanov~\cite{Karzanov:98}.
We denote $\ZeroExt{\mu}$ to be class of problems of the form~\eqref{eq:multifacility},\eqref{eq:multifacility:unaries}.

Optimization problems of the form above have applications in computer vision and related clustering problems in machine learning~\cite{KleinbergTardos:02,Felzenszwalb:CVPR10,Blake:MRFbook,Gridchyn:ICCV13}.
$\ZeroExt{\mu}$ also includes a number of basic combinatorial optimization problems.
For example, the multiway cut problem on $k$ vertices can be obtained by setting $(V,\mu)$ to be the uniform metric on $|V|=k$ elements;
it can be solved in polynomial time (via a maximum flow algorithm) if $k=2$, and is NP-hard for $k\ge 3$.

We explore a generalization of $\ZeroExt{\mu}$ in which the unary terms are given by
\begin{equation}\label{eq:multifacility:unaries:F}
f_i(x_i)=\sum_{v\in V}c_{vi}\mu(x_i,v) + \sum_{U\in F}c_{Ui}\. \mu(x_i,U) .
\end{equation}
Here $F$ is a fixed set of subsets of $V$, $c_{vi},c_{Ui}$ are nonnegative weights, and $\mu(x_i,U)=\min_{v\in U}\mu(x_i,v)$.
We refer to this generalization as $\ZeroExt{\mu,F}$. 
In the facility location interpretation, allowing terms of the form $c_{Ui} \mu(x_i,U)$ means 
that the $i$-th facility can be ``served'' by any of the facilities in $U$,
and it can choose to communicate with the closest facility to minimize the communication cost.

Note that $\ZeroExt{\mu}=\ZeroExt{\mu,\varnothing}$. Furthermore, $\ZeroExt{\mu,2^V}$, where $2^V=\{U\:|\:U\subseteq V\}$ is the set of all subsets of $V$,
is the {\em restricted Metric Labeling Problem}~\cite{ChuzhoyNaor:06}, 
which is equivalent to the Metric Labeling Problem~\cite[Section 5.2]{chekuri04:sidma}. %\medskip

\subsection{Complexity classifications}
The computational complexity of $\ZeroExt{\mu}$  has been established in~\cite{Karzanov:98,Hirai:0ext}.
The tractability criterion is based on the properties of graph $H_\mu=(V,E,w)$
defined as the minimal undirected weighted graph whose path metric equals $\mu$.
Clearly, we have 
$$
E=\left\{ xy\in\binom{V}{2} \;\middle|\; \forall z\in V-\{x,y\} :\; \mu(x,y)<\mu(x,z)+\mu(z,y) \right\}
$$
and $w$ is the restriction of $\mu$ to $E$. 
For brevity, we usually denote the elements of $\binom{V}{2}$ as $xy$ instead of $\{x,y\}$. %\bigskip

In order to state the classification of $\ZeroExt{\mu}$, we need to introduce a few definitions.

\myparagraph{Orientable modular graphs}
Let us fix metric $\mu$. For nodes $x,y\in V$ let $I(x,y)=I_\mu(x,y)$ be the {\em metric interval of $x,y$},
i.e.\ the set of points $z\in V$ satisfying $\mu(x,z)+\mu(z,y)=\mu(x,y)$.
Metric $\mu$ is called {\em modular} if for every triplet $x,y,z\in V$
the intersection $I(x,y)\cap I(y,z) \cap I(x,z)$ is non-empty.
(Points in this intersection are called {\em medians} of $x,y,z$.)
We say that graph $H$ is {\em modular} if $H=H_\mu$ for a modular metric $\mu$.

Let $o$ be an edge-orientation of graph $H$ with the relation $\rightarrow_o$ on $V\times V$.
This orientation is called {\em admissible for $H$}
if, for every 4-cycle $(x_1,x_2,x_3,x_4)$, condition $x_1\rightarrow_o x_2$
implies $x_4\rightarrow_o x_3$.
$H$ is called {\em orientable} if it has an admissible orientation. %\smallskip

\begin{theorem}[\cite{Karzanov:98}]\label{th:Karzanov:NP-hard}
If $H_\mu$ is not orientable or not modular then $\ZeroExt{\mu}$ is NP-hard.
\end{theorem}

\begin{theorem}[\cite{Hirai:0ext}]
If $H_\mu$ is orientable modular then $\ZeroExt{\mu}$ can be solved in polynomial time.~\footnote{In this result $\mu$ is implicitly assumed to be rational-valued, since $\mu$ is treated as part of the input.
The same remark applies to later results on VCSPs.}
\end{theorem}

\myparagraph{Our results} We extend the classification above to problems  $\ZeroExt{\mu,F}$
in which all subsets $U\in F$ have cardinality 2, i.e.\ $F\subseteq\binom{V}{2}$. 
To formulate the tractability criterion, we need to introduce some definitions.
Let $o$ be an orientation of $(H,F)$, i.e.\ each edge of $H$ is assigned an orientation, and each element of $F$
is assigned an orientation.
We say that $o$ is  {\em admissible for $(H,F)$} if it is admissible for $H$ and, for every $\{x,y\}\in F$
with $x\rightarrow_o y$, the following holds: if $P$ is a shortest $x$-$y$ path in $H$
then all edges of $P$ are oriented according to~$o$. We say that $H$ is {\em $F$-orientable modular}
if it is orientable modular and $(H,F)$ admits an admissible orientation $o$.
We can now formulate the main result of this paper.

\begin{theorem}\label{th:main}
 If $H_\mu$ is $F$-orientable modular then $\ZeroExt{\mu,F}$ can be solved in polynomial time.
 Otherwise $\ZeroExt{\mu,F}$ is NP-hard.
\end{theorem}

To prove the tractability part, we define {\em $L$-convex functions on extended modular complexes},
and show that they can be minimized in polynomial time. This generalizes {\em $L$-convex functions on modular complexes}
introduced by Hirai in~\cite{Hirai:0ext,Hirai:Lconvexity}.

\subsection{Discrete convex analysis} 
As Hirai remarks, the approach in~\cite{Hirai:0ext,Hirai:Lconvexity} had been inspired by discrete convex analysis
developed in particular in~\cite{FujishigeMurota:00,Murota:98,MurotaShioura:99,MurotaTamura:04,Murota:book} and~\cite[Chapter VII]{Fujishige:book}.
This is a theory of convex functions on integer lattice $\mathbb Z^n$,
with the goal of providing a unified framework for polynomially solvable combinatorial optimization problems
including network flows, matroids, and submodular functions. Hirai's work extends this theory to more general graph structures,
in particular to orientable modular graphs, and provides a unified framework for polynomially solvable minimum 0-extension problems and related multiflow problems.

We develop a yet another generalization. 
To illustrate the relation to previous work, consider two fundamental classes of functions on the integer lattice $V=[k]=\{1,2,\ldots,k\}$:
{\em submodular functions} and {\em $L^\natural$-convex functions}. These are functions $f:[k]^n \rightarrow\overline{\mathbb R}$ satisfying conditions~\eqref{eq:SUBMODULAR}
and~\eqref{eq:LNATURAL}, respectively:
 \footnote{We use a common notation 
 $\overline{\RR}=\RR\cup\{\infty\}$,
 $\overline{\QQ}=\QQ\cup\{\infty\}$,
 $\overline{\ZZ}=\ZZ\cup\{\infty\}$
where $\infty$ is an infinity element  treated as follows: $\infty \cdot 0 = 0$, 
$x < \infty$ $(x \in \RR)$, $\infty + x = \infty$ $(x \in \overline{\RR})$,  
$x \cdot \infty = \infty$ $(a \in \RR: a > 0)$.
We also denote $\RR_+,\QQ_+,\ZZ_+$ to be the sets of nonnegative elements of $\RR,\QQ,\ZZ$, respectively. 
}
\begin{align}
f(x)+f(y) &\;\;\ge\;\; f(x\wedge y)+f(x\vee y) & \forall x,y\in [k]^n \label{eq:SUBMODULAR} \\
 f(x)+f(y) &\;\;\ge\;\;  f(\left\lceil \tfrac 12 (x+y) \right\rceil)+f(\left\lfloor \tfrac 12 (x+y) \right\rfloor)  & \forall x,y\in [k]^n \label{eq:LNATURAL}
\end{align}
where all operations are applied componentwise.

If, for example, $f(x)=\sum_i f_i(x_i) + \sum_{ij} f_{ij}(x_j-x_i)$,
then $f$ is submodular if all functions $f_{ij}$ are convex,
and $f$ is $L^\natural$-convex if all functions $f_i$ and $f_{ij}$ are convex.
The class of submodular functions on $[k]$ is strictly larger than the class of $L^\natural$-functions.
However, $L^\natural$-convex functions possess additional properties that allow more efficient minimization algorithms,
such as the Steepest Descent Algorithm~\cite{Murota:SDA:03,KolmogorovShioura:SDA:09,MurotaShioura:SDA:14}.

The theory developed in~\cite{Hirai:0ext,Hirai:Lconvexity} covers $L^\natural$-convex functions and several other function classes,
such as bisubmodular functions, {\em $k$-submodular functions}~\cite{huber12:ksub}, {\em skew-bisubmodular functions}~\cite{hkp14:sicomp}, and {\em strongly-tree submodular functions}~\cite{kolmogorov11:mfcs}.
However, it excludes submodular functions on $[k]$ for $k\ge 3$, which is a fundamental class of functions in discrete convex analysis.
This paper fills this gap by introducing a unified framework that includes all classes of functions mentioned above.

\myparagraph{Algorithms for solving $\ZeroExt{\mu}$ and $\ZeroExt{\mu,F}$} The tractability of $\ZeroExt{\mu}$ for orientable modular $\mu$ was proven in~\cite{Hirai:Lconvexity} as follows.
Given an instance $f:V^n\rightarrow\overline\RR$, Hirai defines a different instance $f^\ast_\times:(V^\ast)^n\rightarrow\overline\RR$
with the same minimum, where $|V^\ast|=O(|V|^2)$. Function $f^\ast_\times$ is then minimized using the {\em Steepest Descent Algorithm} (SDA).
This is an iterative technique that at each step computes a minimizer of $f^\ast_\times$ in a certain local neighborhood of the current iterate 
(by solving a linear programming relaxation). 
%\cite{Hirai:Lconvexity} shows that SDA terminates after at most $O(|V^\ast|)$ steps;
%this bound can be improved to $O(|V|)$ steps.
We refer to this technique as the {\em SDA$^\ast$ approach}.

%This approach does not seem to be applicable to our generalization $\ZeroExt{\mu,F}$.
We present an alternative algorithm (for tractable classes of  $\ZeroExt{\mu,F}$) that  minimizes function $f$ directly via a version of SDA that we call 
{\em $\DELTANEIB$-SDA}
(``diamond-SDA'').
Both approaches terminate after $O(|V|)$ steps.
However, one step of SDA$^\ast$ is more expensive than one step of SDA:
the LP problem involves up to $O(|V^\ast|)=O(|V|^2)$ labels per node in the former approach
compared to $O(|V|)$ labels in the latter approach.
Thus, we improve the complexity of solving $\ZeroExt{\mu}$.

\myparagraph{Orthogonal generalizations of the minimum $0$-extension problem} 
In~\cite{Hirai:Lconvexity} Hirai 
considered the minimum $0$-extension problem on {\em  swm-graphs} (that generalize orientable modular graphs),
and defined {\em \hbox{$L$-extendable} functions on swm-graphs}.
Minimizing $L$-extendable functions on swm-graphs
is an NP-hard problem (unless the graph is orientable modular); however, these functions admit a {\em discrete relaxation}
on an orientable modular graph (which is an $L$-convex function). The relaxation can be minimized in polynomial time and yields a partial optimal solution for the original function.
Based on this, Chalopin, Chepoi, Hirai and Osajda~\cite{Chalopin} obtained a 2-approximation
algorithm for the minimum 0-extension problem on swm-graphs. They also developed the theory of swm-graphs.

In \cite{HiraiMizutani} Hirai and Mizutani considered minimum $0$-extension problem for {\em directed} metrics,
and provided some partial results (including a dichotomy for directed metrics of a star graph). 

%\vspace{-5pt}

\subsection{Valued Constraint Satisfaction Problems (VCSPs)}
Results of this paper can be naturally stated in the framework of {\em Valued Constraint Satisfaction Problems} (VCSPs).
This framework is defined below.

Let us fix a finite set $D$ called a {\em domain}. A {\em cost function over $D$ of arity $n$} is a function of the form $f:D^n\rightarrow\overline{\mathbb R}$.
It is called  finite-valued if $f(x)<\infty$ for all $x\in D^n$. We denote $\dom f=\{x\in D^n\:|\:f(x)<\infty\}$.
A {\em (VCSP) language over $D$} is a (possibly infinite) set $\Phi$ of cost functions over $D$.
Language $\Phi$ is called finite-valued if all functions $f\in\Phi$ are finite-valued.

A {\em VCSP instance} $\calI$ is a function $D^n\rightarrow \overline\RR$ given
by
\begin{equation}
f_\calI(x)\ =\ \sum_{t\in T} f_t(x_{v(t,1)},\ldots,x_{v(t,n_t)}).      \label{eq:VCSPinst}
\end{equation}
It is specified by a finite set of variables $[n]$, finite set of terms
$T$, cost functions $f_t : D^{n_t}\rightarrow\overline\RR$ of arity $n_t$ and
indices $v(t, k)\in [n]$ for $t\in T , k =1,\ldots, n_t$. 
A solution to $\calI$ is a labeling $x\in [n]^V$ that minimizes $f_\calI(x)$.
Instance $\calI$ is called a {\em $\Phi$-instance} if all terms $f_t$ belong to $\Phi$.
The set of all $\Phi$-instances is denoted as $\VCSP{\Phi}$.
Language $\Phi$ with finite $|\Phi|$ is called {\em tractable} if instances in $\VCSP{\Phi}$
can be solved in polynomial time, and {\em NP-hard} if $\VCSP{\Phi}$ is NP-hard.
If $|\Phi|$ is infinite then $\Phi$ is called tractable if every finite $\Phi'\subseteq \Phi$ is tractable,
and NP-hard if there exists finite $\Phi'\subseteq \Phi$ which is NP-hard.

A key algorithmic tool in the VCSP theory is the {\em Basic Linear Programming} (BLP) relaxation of instance $\calI$.
We refer to~\cite{kolmogorov15:power} for the description of this relaxation. We say that BLP {\em solves instance $\calI$}
if this relaxation is tight, i.e.\ its optimal value equals $\min_{x\in D^n}f_\calI(x)$.

The following results are known; we refer to Section~\ref{sec:VCSP} for the definition of a ``binary symmetric fractional polymorphism''.

\begin{theorem}[{\cite{kolmogorov15:power}}]\label{th:KHDAHLKDAG}
Let $\Phi$ be a finite-valued language. Then BLP solves $\Phi$ if and only if $\Phi$ admits a binary symmetric fractional polymorphism.
If the condition holds, then an optimal solution of an $\Phi$-instance can be computed in polynomial time.
\end{theorem}
\begin{theorem}[{\cite{tz16:jacm}}]\label{th:PASGA}
If a finite-valued language $\Phi$ does not satisfy the condition in Theorem~\ref{th:KHDAHLKDAG} then $\Phi$ is NP-hard.
\end{theorem}

\myparagraph{Application to the minimum $0$-extension problem} %Clearly, the $0$-extension 
Consider again a metric space $(V,\mu)$ and subset $F\subseteq\binom{V}{2}$.
For a set $U\subseteq V$ let $\delta_U:V\rightarrow \{0,\infty\}$ be the indicator function of set $U$,
with $\delta_U(v)=0$ iff $v\in U$.
For brevity, we  write $\delta_{\{u_1,\ldots, u_k\}}$ as $\delta_{u_1\ldots u_k}$. Clearly, the minimum $0$-extension problems introduced earlier can be equivalently defined 
by the  following languages over domain $D=V$: 
\vspace{-10pt}
\begin{eqnarray*}
\ZeroExt{\mu}&=&\{\mu\} \;\cup\; \{\delta_u\::\:u\in V\} \\
\ZeroExt{\mu,F}&=&\{\mu\} \;\cup\; \{\delta_u\::\:u\in V\} \;\cup\; \{\delta_U\::\:U\in F\}
\end{eqnarray*}

Note that the existence of the dichotomy given in Theorem~\ref{th:main} follows from Theorems~\ref{th:KHDAHLKDAG} and~\ref{th:PASGA}
(but not the specific criterion for tractability).\footnote{Theorems~\ref{th:KHDAHLKDAG} 
and~\ref{th:PASGA} are not directly applicable to $\ZeroExt{\mu}$ since $\ZeroExt{\mu}$ is not finite-valued.
However, we can get a finite-valued language by replacing functions $\delta_u:V\rightarrow\{0,\infty\}$
with functions $\mu_u:V\rightarrow\RR$ defined via $\mu_u(x)=\mu(x,u)$. It is not difficult to show that such transformation
does not affect the complexity of $\ZeroExt{\mu}$. A similar remark applies to $\ZeroExt{\mu,F}$.}

\subsection{Summary of contributions}
As described in earlier sections, our first contribution
is complexity classification of $\ZeroExt{\mu,F}$ for subsets $F\subseteq \binom{V}{2}$ with an explicit
criterion for tractability, which is achieved by generalizing Hirai's theory of modular complexes to
{\em extended modular complexes}.
While some of the proofs are relatively straightforward extensions of the corresponding proofs in~\cite{Hirai:0ext,Hirai:Lconvexity},
there are also a number of proofs where we use novel techniques.
We already mentioned a new $\DELTANEIB$-SDA algorithm that improves the complexity of solving the standard 0-minimum extension problem.
In order to analyze this algorithm,
we introduce new binary operations $\UP,\DOWN,\diamond$ for (extended) modular complexes and establish their properties.
Another key technical component that we use is the notion of {\em $f$-extremality} that we introduce.
We believe that these concepts deepen our understanding of orientable modular graphs.

As our last contribution, we prove that the BLP relaxation directly solves $L$-convex functions on extended modular complexes.
Previously, this was shown to hold (for standard modular complexes) only assuming that ${\tt P}\ne {\tt NP}$ (see Section 6 in~\cite{Hirai:0ext}).

\bigskip

The rest of the paper is organized as follows. Section~\ref{sec:background} reviews Hirai's theory
and defines {\em $L$-convex functions on modular complexes} and the SDA algorithm for minimizing them.
Section~\ref{sec:ebp} generalizes this to {\em $L$-convex functions on extended modular complexes}
and presents $\DELTANEIB$-SDA algorithm.
Both sections use the notion of {\em submodular functions on valuated modular semilattices},
which are formally defined in Section~\ref{sec:submodularity-on-semilattice}.
All proofs missing in Section~\ref{sec:ebp} are given in Sections~\ref{sec:proofs} and~\ref{sec:VCSP1}; this completes the proof of the tractability direction of Theorem~\ref{th:main}. The NP-hardness direction of Theorem~\ref{th:main} is proven in Section~\ref{sec:proofs:NPhardness}.
Section~\ref{sec:open} concludes the paper with a list of open problems.

\section{Background on orientable modular graphs}\label{sec:background}

\myparagraph{Notation for graphs}
If $H$ is a simple undirected graph and $o$ its edge orientation, then the pair $(H,o)$ can be viewed as a simple directed graph.
We will usually denote this graph as $\Gamma=(V_\Gamma,E_\Gamma,w_\Gamma)$.
$\Gamma$~is called {\em oriented modular}, or a {\em modular complex}, if it is an admissible orientation of an orientable modular graph~\cite{Karzanov:04,Hirai:0ext}.

We let~$\rightarrow_\Gamma$ be the edge relation of $\Gamma$,
i.e.\ condition $u\rightarrow_\Gamma v$ means that there is an edge from $u$ to $v$ in $\Gamma$.
When $\Gamma$ is clear from the context, we may omit subscript $\Gamma$ and write $V,E,w,\rightarrow$, etc.
A {\em path} $(u_0,u_1,\ldots,u_k)$ in a directed graph $\Gamma$ is defined as a path in the undirected version of $\Gamma$,
i.e.\ for each $i$ we must have either $u_i\rightarrow u_{i+1}$ or $u_{i+1}\rightarrow u_{i}$.
An $x$-$y$ path in $\Gamma$ is a path from $x\in V$ to $y\in V$.
With some abuse of notation we sometimes view $\Gamma$ as a set of its nodes, and write e.g.\ $v\in \Gamma$ to mean $v\in V$.

\myparagraph{Orbits} For an undirected graph $H=(V,E,w)$ edges $e,e'\in E$ 
are called {\em projective} if there is a sequence of edges $(e_0,e_1,\ldots,e_m)$ 
with $(e_0,e_m)=(e,e')$ such that $e_i,e_{i+1}$ are vertex-disjoint and belong to a common 4-cycle of $H$.
Clearly, projectivity is an equivalence relation on $E$.
An equivalence class of this relation in called an {\em orbit}~\cite{Karzanov:04}.
Edge weights $w:E\rightarrow\mathbb R_{>0}$ are called {\em orbit-invariant}
if $w(e)=w(e')$ for any pair of edges $e,e'$ in the same orbit (equivalently, for vertex-disjoint edges $e,e'$ belonging to a common 4-cycle).
\begin{theorem}[\cite{Bandelt:85,Karzanov:04}]\label{th:orbits}
Consider undirected graph $H=H_\mu=(V,E,w)$. 
\begin{itemize}[noitemsep,topsep=0pt]
\item [{\rm (a)}] If $\mu$ is modular then $w$ is orbit-invariant,
and  path $P$  is shortest in $H$ if and only if it is shortest in $(V,E,1)$. 
\item [{\rm (b)}] The following conditions are equivalent: (i) $H$ is (orientable) modular;
(ii) $w$ is orbit-invariant and $(V,E,1)$ is (orientable) modular.
\end{itemize}
\end{theorem}

\myparagraph{Metric spaces} 
For a weighted directed or undirected graph $G=(V,E,w)$ let $\mu_G$ and $d_G$ (or simply $\mu$ and $d$, when $G$ is clear) be its path metrics w.r.t.\ edge lengths $w$ and $1$, respectively.
If graph $G$ is directed then edge orientations are again ignored.

For a metric space $(V,\mu)$, subset $U\subseteq V$ is called {\em convex} if $I(x,y)\subseteq U$
for every $x,y\in U$.  
Note, if $G$ is an orientable (or oriented) modular graph then
the definitions of the metric interval $I(x,y)$ and of convex sets %, gated sets and maps $\Pr_U$ %and $c$-isometric embeddability
coincide for metric spaces $(V_G,\mu_G)$ and $(V_G,d_G)$ (by  Theorem~\ref{th:orbits}).

\myparagraph{Posets} 
A modular complex $\Gamma$ known to be an acyclic graph~\cite[Lemma 2.3]{Hirai:0ext},
and thus induces a partial order $\preceq$ on~$V$.
Partially ordered sets (posets) play a key role in the study of oriented modular graphs.
Below we describe basic facts about posets and terminology that we use, mainly following~\cite{Hirai:0ext,Hirai:Lconvexity}.

Consider poset $\calL$ with relation $\preceq$.
For elements $p,q$ with $p\preceq q$, the {\em interval} $[p,q]$ is the set $\{x\in\calL\:|\:p\preceq x\preceq q\}$.
A {\em chain from $p$ to $q$ of length $k$} is a sequence $u_0\prec u_1\prec\ldots\prec u_k$ with $(u_0,u_k)=(p,q)$,
where notation $a\prec b$ means that $a\preceq b$ and $a\ne b$.
The length $r[p,q]$ of interval $[p,q]$ is defined as the maximum length of a chain from $p$ to $q$.
If $\calL$ has the lowest element (denoted as $0$) then the {\em rank} $r(a)$ of element $a\in\calL$ is defined by $r(a)=r[0,a]$,
and elements of rank 1 are called {\em atoms}.

Element $q$ {\em covers} $p$ if $p\prec q$ and there is no element $u$ with $p\prec u\prec q$.
The {\em Hasse diagram of~$\calL$} is a directed graph on $\calL$ with the set of edges $\{p\rightarrow q\:|\:q\mbox{ covers }p\}$,
and the {\em covering graph of $\calL$} is the corresponding undirected graph.

A pair $x,y\in\calL$ is said to be {\em upper-bounded} (resp. {\em lower-bounded}) if $x,y$ have a common upper bound
(resp. common lower bound). The { lowest} common upper bound, if exists, is denoted by $x\vee y$ (the ``join'' of $x,y$),
and the { greatest} common lower bound, if exists, is denoted by $x\wedge y$ (the ``meet'' of $x,y$).
$\calL$ is called a {\em (meet-)semilattice} if every pair $x,y\in \calL$ has a meet,
and it is called  a {\em lattice} if every pair $x,y\in \calL$ has both a join and a meet.
If $\calL$ is a semilattice and $x,y\in\calL$ are upper-bounded then $x\vee y$ is known to exist.

A (positive) {\em valuation} of a semilattice $\calL$ is a function $v:\calL\rightarrow\mathbb R$ satisfying
\begin{subequations}\label{eq:valuation:def}
\begin{align}
v(q)-v(p) > 0 &\qquad\qquad \forall p,q\in\calL:\,p\prec q \label{eq:valuation:def:a} \\
v(p)+v(q) = v(p\wedge q)+v(p\vee q) &\qquad\qquad \forall p,q\in\calL:\,p,q\mbox{ upper-bounded}  \label{eq:valuation:def:b}
\end{align}
\end{subequations}
In particular, if $\calL$ is a lattice then~\eqref{eq:valuation:def:b} should hold for all $p,q\in\calL$.
A semilattice with valuation~$v$ will be called a {\em valuated semilattice}.
We will view the Hasse diagram of a valuated semilattice $\calL$ (and the corresponding covering graph) as weighted graphs, where the weight of $p\rightarrow q$ is given by $v(q)-v(p)>0$.

A lattice $\calL$ is called {\em modular} if for every $x,y,z \in \calL$ with $x \preceq z$ there holds $x \vee (y \wedge z) = (x \vee y) \wedge z$.
A semilattice $\calL$ is called {\em modular}~\cite{Bandelt:93}  if for every $p\in\calL$ poset $(\{x\in\calL\:|\:x\preceq p\},\preceq)$ is a modular lattice, 
and for every $x,y,z \in \calL$ the join $x \vee y \vee z$ exists provided that $x \vee y$, $y \vee z$, $z \vee x$ exist.
It is known that a lattice $\calL$ is modular if and only if its rank function $r(\cdot)$ is a valuation~\cite[Chapter III, Corollary 1]{Birkhoff:67}.
Furthermore, a (semi)lattice is modular if and only if its covering graph is modular, see~\cite[Proposition 6.2.1]{vanDeVel:book};
\cite[Theorem 5.4]{Bandelt:93}.
The Hasse diagram of a (valuated) modular semilattice is known to be oriented modular~\cite[page 13]{Hirai:0ext}.

\myparagraph{Boolean pairs}
From now on we fix a modular complex $\Gamma=(V,E,w)$.
Graph $B$ is called a {\em cube graph} if it is isomorphic to the Hasse diagram of the Boolean lattice $\{0,1\}^k$ for some $k\ge 0$.
A pair of vertices $(p,q)$ of $\Gamma$
is called a {\em Boolean pair} if~$\Gamma$ contains cube graph $B$ as a subgraph 
so that $p$ and $q$ are respectively the source and the sink of~$B$.
Let $\sqsubseteq$  be the following relation on $V$:
 $p\sqsubseteq q$ iff $(p,q)$ is a Boolean pair in $\Gamma$.
We have $p\sqsubseteq p$ for any $p\in V$, and condition $p\sqsubseteq q$ implies that $p\preceq q$. 
Graph $\Gamma$ is called {\em well-oriented} if the opposite implication holds, i.e.\ if relations $\preceq$ and $\sqsubseteq$ are the same.
Note that relation $\sqsubseteq$ is not necessarily transitive.
We write $p\sqsubset q$ to mean $p\sqsubseteq q$ and $p\ne q$.
Let $\Gamma^{\mathsmaller\sqsubset}$ be the graph with nodes $V$ and edges $\{p\rightarrow q\:|\:p,q\in V,p\sqsubset q\}$.
Clearly, $\Gamma$ is a subgraph of~$\Gamma^{\mathsmaller\sqsubset}$ (ignoring edge weights).

For a vertex $p\in V$ define the following subsets of $V$: % (viewed as posets with relation $\preceq$):
\begin{subequations}\label{eq:Lneib}
\begin{align}
 \calL^\uparrow_p(\Gamma)=\{q\in V\:|\:p\preceq q\} & \qquad
 \calL^\downarrow_p(\Gamma)=\{q\in V\:|\:q\preceq p\} & 
 \\
 \calL^+_p(\Gamma)=\{q\in V\:|\:p\sqsubseteq q\} & \qquad
 \calL^-_p(\Gamma)=\{q\in V\:|\:q\sqsubseteq p\} & \qquad
 \calL^\pm_p(\Gamma)=\calL^+_p(\Gamma) \cup \calL^-_p(\Gamma)
\end{align}
\end{subequations}
When $\Gamma$ is clear from the context, we will omit it for brevity
(and if $\Gamma$ is not clear, we may write $\preceq_\Gamma$ and $\sqsubseteq_\Gamma$ instead of $\preceq$ and $\sqsubseteq$).
We view $\calL^\uparrow_p,\calL^+_p$ as posets with relation $\preceq$,
and $\calL^\downarrow_p,\calL^-_p$ as posets with the reverse of relation $\preceq$.
Note, if $\Gamma$ is well-oriented then $\calL_p^+=\calL^\uparrow_p$ and $\calL_p^-=\calL^\downarrow_p$ for every  node $p$ of~$\Gamma$.
\begin{lemma}[{\cite[Proposition 4.1, Theorem 4.2, Lemma 4.14]{Hirai:0ext}}]\label{lemma:GammaIdeal}
Let $\Gamma$ be a modular complex.
\begin{itemize}[noitemsep,topsep=0pt]
	\item [{\rm (a)}] 
If elements $a,b$ are upper-bounded then $a\vee b$ exists.
Similarly, if $a,b$ are lower-bounded then $a\wedge b$ exists. 
	\item [{\rm (b)}]
Consider elements $p,q$ with $p\preceq q$. Then 
$\calL^\uparrow_p$, $\calL^\downarrow_p$, $\calL^+_p$, $\calL^-_p$ 
are modular semilattices,
and $[p,q]=\calL^\uparrow_p\cap \calL^\downarrow_p$ is a modular lattice.
Furthermore, these (semi)lattices are convex in $\Gamma$, and
function $v(\cdot)$ defined via $v(a)=\mu_\Gamma(p,a)$
is a valid valuation of these (semi)lattices. 
\end{itemize}
\end{lemma}
We will always view 
$\calL^\uparrow_p$, $\calL^\downarrow_p$, $\calL^+_p$, $\calL^-_p$  
as {\bf valuated} semilattices, where the valuation is defined as in the lemma.

\myparagraph{$L$-convex functions} 
Next, we review the notion of {\em an $L$-convex function on a modular complex~$\Gamma$} introduced by Hirai in~\cite{Hirai:0ext,Hirai:Lconvexity}.
The definition involves the following steps.
\begin{itemize}
\item First, Hirai defines the notion of a {\em submodular function on a valuated modular semilattice~$\calL$}. These are
functions $f:\calL\rightarrow\overline{\mathbb R}$ satisfying certain linear inequalities. Formulating these inequalities is rather lengthy,
and we defer it to Section~\ref{sec:submodularity-on-semilattice}.
\item Second, Hirai defines {\em 2-subdivision} of $\Gamma$ as the directed weighted graph $\Gamma^\ast=(V^\ast,E^\ast,w^\ast)$
constructed as follows: 
\begin{itemize}
\item[(i)] set $V^\ast=\{[p,q]\:|\:p\sqsubseteq q\}$; 
\item[(ii)] for each $[p,q],[p,q']\in V^\ast$ with  $q\rightarrow_\Gamma q'$ add edge $[p,q]\rightarrow [p,q']$ to $\Gamma^\ast$
%(denoted as $[p,qq']$ for brevity) 
with weight $ w(qq')$; 
\item[(iii)] for each $[p',q],[p,q]\in V^\ast$ with  $p'\rightarrow_\Gamma p$ add edge $[p,q]\rightarrow [p',q]$ to $\Gamma^\ast$
%(denoted as $[p'p,q]$ for brevity) 
with weight $ w(p'p)$. 
\end{itemize}
It can be seen that graph $(V^\ast,E^\ast)$ is the Hasse diagram of the poset $(V^\ast,\subseteq)$ (which is the definition used in~\cite{Hirai:Lconvexity}).~\footnote{
Compared to~\cite{Hirai:0ext,Hirai:Lconvexity}, we chose to scale the weights of graph $\Gamma^\ast$ by a factor of 2.
Such scaling will not affect later theorems.}

Hirai proves that graph $\Gamma^\ast$ is oriented modular (\cite[Theorem 4.3]{Hirai:0ext}) and well-oriented (\cite[Lemma 2.14]{Hirai:0ext}).
Consequently, poset $\calL_{[p,p]}^+(\Gamma^\ast)=\calL_{[p,p]}^\uparrow(\Gamma^\ast)$ is a valuated modular semilattice for each $p\in V$. For brevity, this
poset will be denoted as $\calL_{p}^\ast(\Gamma)$, or simply as $\calL_{p}^\ast$.
\item Each function $f:V\rightarrow\overline{\mathbb R}$ is extended to a function $f^\ast:V^\ast \rightarrow\overline{\mathbb R}$ via
$f^\ast([p,q])=f(p)+f(q)$.
\item Function $f:V\rightarrow\overline{\mathbb R}$ is now called {\em $L$-convex on $\Gamma$}
if (i) subset $\dom f\subseteq V$ is connected in $\Gamma^{\mathsmaller\sqsubset}$,
and (ii) for every $p\in V$, the restriction of $f^\ast$ to $\calL^\ast_{p}$ is submodular on valuated modular semilattice $\calL^\ast_{p}$.
\end{itemize}

\myparagraph{Minimum $0$-extension problem and $L$-convex functions} Next, we describe the relation between problem $\ZeroExt{\mu}$ and $L$-convex functions on $\Gamma$
(where $\mu$ is an orientable modular metric and $\Gamma$ is an admissible orientation of $H_\mu$).

If $\Gamma,\Gamma'$ are two modular complexes, then their Cartesian product $\Gamma\times \Gamma'$
is defined as the directed graph with the vertex set $V_\Gamma\times V_{\Gamma'}$ and the following 
edge set: there is an edge $(p,p')\rightarrow(q,q')$ 
iff either (i) $p=q$ and $p'\rightarrow_{\Gamma'} q'$, or (ii) $p'=q'$ and $p\rightarrow_\Gamma q$.
The weight of the edge is $w_{\Gamma'}(p'q')$ in the first case and $w_{\Gamma}(pq)$ in the second case.
The $n$-fold Cartesian product $\Gamma\times\ldots\times \Gamma$ is denoted as $\Gamma^n$.

A Cartesian product $\calL\times\calL'$ of two posets $\calL,\calL'$ is defined in a natural way
(with $(p,p')\preceq (q,q')$ iff $p\preceq p$ and $q\preceq q'$). It is straightforward to verify from definitions
that if $\calL,\calL'$ are modular semilattices then so is $\calL\times\calL'$.
If $\calL,\calL'$ are valuated modular semilattices with valuations $v_{\calL},v_{\calL'}$ then
$\calL\times\calL'$ is also assumed to be valuated with the valuation $v_{\calL\times\calL'}(p,p')=v_{\calL}(p)+v_{\calL'}(p')$.

\begin{lemma}[{\cite[Lemma 4.7]{Hirai:0ext}}]\label{lemma:bp-cartesian}
Consider modular complexes $\Gamma,\Gamma'$ and element $(p,p')\in\Gamma\times\Gamma'$. 
\begin{itemize}[noitemsep,topsep=0pt]
\item[{\rm (a)}] Graph $\Gamma\times \Gamma'$ is a modular complex (i.e.\ oriented modular). 
\item[{\rm (b)}] $\calL^\sigma_{(p,p')}(\Gamma\times\Gamma')=\calL^\sigma_{p}(\Gamma)\times \calL^\sigma_{p'}(\Gamma')$ for $\sigma\in\{-,+\}$. 
%\item[{\rm (c)}] $(\Gamma\times\Gamma')^\ast$ is isomorphic to $\Gamma^\ast\times \Gamma'^\ast$ via isomorphism $[(p,p'),(q,q')]\mapsto ([p,q],[p',q'])$.
%\item[{\rm (d)}] $\calL^\ast_{(p,p')}(\Gamma\times\Gamma')=\calL^\ast_{p}(\Gamma)\times \calL^\ast_{p}(\Gamma)$. 
\end{itemize}
\end{lemma}

\begin{theorem}[{\cite[Theorem 4.8 and Lemma 4.9]{Hirai:0ext};\cite[Lemmas 4.1, 4.2 and Proposition 4.4]{Hirai:Lconvexity}}]\label{th:addition:orig}
Consider modular complexes $\Gamma,\Gamma'$ and functions $f,f':\Gamma\rightarrow \overline{\mathbb R}$ and $\tilde f:\Gamma\times\Gamma'\rightarrow \overline{\mathbb R}$. 
\begin{itemize}[noitemsep,topsep=0pt]
\item[{\rm (a)}] If $f,f'$ are $L$-convex  on $\Gamma$ then $f+f'$ and $c\cdot f$ for $c\in\mathbb R_+$ are also $L$-convex on $\Gamma$. 
\item[{\rm (b)}] If $f$ is $L$-convex on $\Gamma$ and $\tilde f(p,p')=f(p)$ for $(p,p')\in \Gamma\times\Gamma'$ then $\tilde f$ is $L$-convex on $\Gamma\times\Gamma'$. 
\item[{\rm (c)}] If $\tilde f$ is $L$-convex on $\Gamma\times\Gamma'$ and $f(p)=\tilde f(p,p')$ for fixed $p'\in\Gamma'$
then $f$ is $L$-convex on $\Gamma$. 
\item[{\rm (d)}] The indicator function $\delta_U:V\rightarrow \{0,\infty\}$ of a $d_\Gamma$-convex set $U$ is $L$-convex on $\Gamma$. 
\item[{\rm (e)}]  Function $\mu_\Gamma:\Gamma\times \Gamma\rightarrow\mathbb R_+$ is $L$-convex on $\Gamma\times \Gamma$.
\item[{\rm (f)}]  If $f$ is $L$-convex on $\Gamma$ then the restriction of $f$ to $\calL^\sigma_p(\Gamma)$ is submodular on $\calL^\sigma_p(\Gamma)$ for every
$p\in\Gamma$ and $\sigma\in\{-,+\}$.
\end{itemize}
\end{theorem}

It follows that an instance of $\ZeroExt{\mu}$ can be reduced to the problem of minimizing function $f:\Gamma^n\rightarrow\overline{\mathbb R}$, $n\ge 1$
such that (i) $f$ is $L$-convex on $\Gamma^n$, and (ii) $f(x_1,\ldots,x_n)$ can be written as a sum of unary and pairwise terms
that are $L$-convex on $\Gamma$ and $\Gamma\times\Gamma$, respectively. This minimization problem is considered next.

\myparagraph{Minimizing $L$-convex functions}
A key property of $L$-convex functions is that local optimality implies global optimality.
\begin{theorem}[{\cite[Lemma 2.3]{Hirai:0ext}}]
Let $f:V\rightarrow\overline{\mathbb R}$ be an $L$-convex function on a modular complex~$\Gamma$.
If $p$ is a local minimizer of $f$ on $\calL^\pm_p(\Gamma)$, i.e. $f(p)=\min_{q\in \calL^\pm_p(\Gamma)}f(q)$, then it is also a global minimizer of $f$, i.e.\ 
$f(p)=\min_{q\in V}f(q)$.
\end{theorem}
This theorem implies that the {\em Steepest Descent Algorithm} (SDA) given below is guaranteed to produce a global minimizer of $f$
after a finite number of steps.

\begin{algorithm}[H]
  \DontPrintSemicolon
  pick arbitrary $x\in\dom f$ \\
  \While{\tt true}
      {
        compute $y^+\in\argmin \,\{f(y)\:|\:y\in\calL^+_x(\Gamma^n)\} $ and $y^-\in\argmin \, \{f(y)\:|\:y\in\calL^-_x(\Gamma^n)\} $ \\
        pick $y\in\argmin\,\{f(y)\:|\:y\in\{y^-,y^+\}\} $ \\
        if $f(y) = f(x)$ then return $x$, otherwise update $x:= y$
      }
      \caption{Steepest Descent for minimizing $f:\Gamma^n\rightarrow\overline{\mathbb R}$ }
      \label{alg:SDA}
\end{algorithm}
By Lemma~\ref{lemma:bp-cartesian}, for $x=(x_1,\ldots,x_n)$ we have $\calL^\sigma_x(\Gamma^n)=\calL^\sigma_{x_1}(\Gamma)\times\ldots\times\calL^\sigma_{x_n}(\Gamma)$
for $\sigma\in\{-,+\}$. By Theorem~\ref{th:addition:orig}(f), function $f$ is submodular on $\calL^\sigma_x(\Gamma^n)$.
The result below thus implies that the two minimization problems in line 3 can be
solved in polynomial time assuming that $f$ is an instance of $\ZeroExt{\mu}$.
\begin{theorem}[{\cite[Theorem 3.9]{Hirai:0ext}}]\label{th:BLP-solves-submodular}
Consider VCSP instance $f:\calL_1\times\ldots\times\calL_n\rightarrow\overline{\mathbb R}$.
Suppose that $\calL_1,\ldots,\calL_n$ are valuated modular semilattices and function $f$ is submodular
on $\calL=\calL_1\times\ldots\times\calL_n$. Then BLP relaxation solves $f$.
\end{theorem}

To show that $\ZeroExt{\mu}$ can be solved in polynomial time, a few additional definitions and results are needed.
Elements $p,q$ of a modular complex $\Gamma$ are said to be {\em $\DELTANEIB$-neighbors} if $p,q\in[a,b]$ for some $a,b$ with $a\sqsubseteq b$.
Equivalently, $p,q$ are $\DELTANEIB$-neighbors if $p\wedge q,p\vee q$ exist and $p\wedge q\sqsubseteq p\vee q$.
Let $\DELTA{\Gamma}$ be an undirected unweighted graph on nodes $V_\Gamma$ such that $p,q$ are connected by an edge in $\DELTA{\Gamma}$ if and only if $p,q$ are $\DELTANEIB$-neighbors.
This graph is called a {\em thickening of $\Gamma$}. The distance from $p$ to $q$ in $\DELTA{\Gamma}$ will be denoted as $\DELTA{d}_\Gamma(p,q)$, or 
simply as $\DELTA{d}(p,q)$. These distances 
will give
a bound on the number of steps of SDA.

\begin{theorem}[{\cite[Theorem 4.7,Lemma 2.18]{Hirai:Lconvexity}}]\label{th:SDA:orig}
{
Suppose that modular complex $\Gamma$ is well-oriented,
and function $f:\Gamma^n\rightarrow\overline{\mathbb R}$ is $L$-convex on $\Gamma^n$.
SDA terminates after at most $2+\max\limits_{i\in [n]}\DELTA{d}_\Gamma(x_i,{\tt opt}_i(f))$ iterations
where $x$ is the initial vertex and ${\tt opt}_i(f)=\{x^\ast_i\:|\:x^\ast\in\argmin f\}\subseteq V_\Gamma$ is the set of minimizers of $f$ projected to the $i$-th component. 
}
\end{theorem}
\begin{theorem}[{\cite[Proposition 6.10]{Chalopin},\cite[Lemma 2.14]{Hirai:Lconvexity}}]\label{th:well-oriented}
If $\Gamma$ is a modular complex
then graph $\Gamma^\ast$ is well-oriented.
\end{theorem}
\begin{theorem}[{\cite[Proposition 4.9]{Hirai:Lconvexity},\cite[Lemma 4.7]{Hirai:0ext}}]\label{th:fstar-Lconvex}
If $\Gamma$ is a modular complex
and function $f:\Gamma^n\rightarrow\overline{\mathbb R}$ is $L$-convex on $\Gamma^n$ then 
function $f^\ast_\times:(\Gamma^\ast)^n\rightarrow\overline{\mathbb R}$
defined via $f^\ast_\times([x_1,y_1],\ldots,[x_n,y_n])=f(x)+f(y)$
 is $L$-convex on $(\Gamma^\ast)^n$.
\end{theorem}
From the results above one can now conclude that 
any instance $f:\Gamma^n\rightarrow\overline{\mathbb R}$ of $\ZeroExt{\mu}$ for an oriented modular metric $\mu$ can be solved in polynomial time.
Indeed, apply SDA to minimize function $f^\ast_\times:(\Gamma^\ast)^n\rightarrow\overline{\mathbb R}$.
It will produce a minimizer $([x_1,y_1],\ldots,[x_n,y_n])$ after at most ${\tt diam}(\DELTA{(\Gamma^\ast)})+2$ iterations;
both $x$ and $y$ are minimizers of $f$.
Hirai gives the bound ${\tt diam}(\DELTA{(\Gamma^\ast)})\le |V|^2$
in~\cite[Theorem 5.7]{Hirai:Lconvexity}. 
As pointed out by an anonymous reviewer,
this bound can be improved to ${\tt diam}(\DELTA{(\Gamma^\ast)})\le 2 \cdot {\tt diam}(\DELTA{\Gamma}) \le 2\cdot |V|$.
 
Note that in the earlier paper~\cite{Hirai:0ext} Hirai proved that $\ZeroExt{\mu}$ can be solved in weakly polynomial time
via a different algorithm, namely the Steepest Descent Algorithm (applied to $f:\Gamma^n\rightarrow\overline{\mathbb R}$)
but combined with a cost-scaling approach.

\begin{remark}
Our terminology is slightly different from that of~\cite{Hirai:Lconvexity}.
In particular, $\DELTANEIB$-neighbors were called {\em $\Delta$-neighbors} in~\cite{Hirai:Lconvexity},
and a path in $\Gamma^{\mathsmaller\sqsubset}$ was called a {\em $\Delta'$-path}.
\end{remark}

%%%%%%%%%%%%%%%%%%%%%%%%%%%%%%%%%%%%%%%%%%%%%%%%%%%%%%%%%%%%%%%%%%%%%%%%%%%%%%%%%%%%%%%%%%%%%%%%%%%%%%%%%%%%%%%%%%%%%%%%%%%%%%%%%%
%%%%%%%%%%%%%%%%%%%%%%%%%%%%%%%%%%%%%%%%%%%%%%%%%%%%%%%%%%%%%%%%%%%%%%%%%%%%%%%%%%%%%%%%%%%%%%%%%%%%%%%%%%%%%%%%%%%%%%%%%%%%%%%%%%
%%%%%%%%%%%%%%%%%%%%%%%%%%%%%%%%%%%%%%%%%%%%%%%%%%%%%%%%%%%%%%%%%%%%%%%%%%%%%%%%%%%%%%%%%%%%%%%%%%%%%%%%%%%%%%%%%%%%%%%%%%%%%%%%%%
%%%%%%%%%%%%%%%%%%%%%%%%%%%%%%%%%%%%%%%%%%%%%%%%%%%%%%%%%%%%%%%%%%%%%%%%%%%%%%%%%%%%%%%%%%%%%%%%%%%%%%%%%%%%%%%%%%%%%%%%%%%%%%%%%%
%%%%%%%%%%%%%%%%%%%%%%%%%%%%%%%%%%%%%%%%%%%%%%%%%%%%%%%%%%%%%%%%%%%%%%%%%%%%%%%%%%%%%%%%%%%%%%%%%%%%%%%%%%%%%%%%%%%%%%%%%%%%%%%%%%
%%%%%%%%%%%%%%%%%%%%%%%%%%%%%%%%%%%%%%%%%%%%%%%%%%%%%%%%%%%%%%%%%%%%%%%%%%%%%%%%%%%%%%%%%%%%%%%%%%%%%%%%%%%%%%%%%%%%%%%%%%%%%%%%%%
%%%%%%%%%%%%%%%%%%%%%%%%%%%%%%%%%%%%%%%%%%%%%%%%%%%%%%%%%%%%%%%%%%%%%%%%%%%%%%%%%%%%%%%%%%%%%%%%%%%%%%%%%%%%%%%%%%%%%%%%%%%%%%%%%%
%%%%%%%%%%%%%%%%%%%%%%%%%%%%%%%%%%%%%%%%%%%%%%%%%%%%%%%%%%%%%%%%%%%%%%%%%%%%%%%%%%%%%%%%%%%%%%%%%%%%%%%%%%%%%%%%%%%%%%%%%%%%%%%%%%

\section{Extended modular complexes}\label{sec:ebp}
To prove our main tractability result from Theorem~\ref{th:main}, we will introduce the following definition.
\begin{definition}\label{def:ebp}
Let $\Gamma$ be a modular complex. Binary relation $\sqsubseteq$ on $V=V_\Gamma$ is called {\em admissible}
if it coarsens $\preceq$ (i.e.\ $p\sqsubseteq q$ implies $p\preceq q$),
$p\sqsubseteq p$ for every $p\in V$, $p\sqsubseteq q$ for every edge $p\rightarrow q$,
and the following conditions hold.
\begin{itemize}[noitemsep,topsep=0pt,itemindent=10pt]
\item[{\em (\ref{def:ebp}a)}] Suppose that $p\sqsubseteq q$, $a\preceq b$ and $a,b\in[p,q]$. Then $a\sqsubseteq b$.
\item[{\em (\ref{def:ebp}b)}] Suppose that $p\sqsubseteq q_1$, $p\sqsubseteq q_2$ and $q_1\vee q_2$ exists. Then $p\sqsubseteq q_1\vee q_2$.
\item[{\em (\ref{def:ebp}c)}] Suppose that $p_1\sqsubseteq q$, $p_2\sqsubseteq q$ and $p_1\wedge p_2$ exists. Then $p_1\wedge p_2\sqsubseteq q$.
\end{itemize}
\end{definition}

A pair $(\Gamma,\sqsubseteq)$ where $\Gamma$ is a modular complex and $\sqsubseteq$ is admissible will be called an {\em extended modular complex}.
With some abuse of notation, we will use letter $\Gamma$ for an extended modular complex, and treat it as an oriented modular graph when needed.
Note that previously we used notation $\sqsubseteq$ for Boolean pairs in $\Gamma$.
From now on we will write $p\,\BPrel\, q$ if $(p,q)$ is a Boolean pair in~$\Gamma$, and reserve notation $\sqsubseteq$ 
(or $\sqsubseteq_\Gamma$) for the binary relation that comes with an extended modular complex $\Gamma$.
As before, we write $p \sqsubset q$ to mean $p \sqsubseteq q$ and $p\ne q$.

\begin{proposition}\label{prop:ebp:admissible}
Let $\Gamma$ be a modular complex. (a) Relations $\BPrel$ and $\preceq$ are admissible for $\Gamma$.
(b)~If relation $\sqsubseteq$ is admissible for $\Gamma$ then $p\, \BPrel\, q$ implies $p \sqsubseteq q$.
\end{proposition}
This proposition shows that a modular complex is a special case of an extended modular complex
(obtained by setting $\sqsubseteq$ to be $\BPrel$). Also, $\BPrel$ and $\preceq$ are respectively the coarsest and the finest admissible relations.
Next, we will show that many of the results in Section~\ref{sec:background} still hold if relation $\BPrel$ is replaced with an arbitrary admissible relation $\sqsubseteq$.

From now on we fix an extended modular complex $\Gamma$.
We define graphs  $\Gamma^{\mathsmaller\sqsubset}$ and $\DELTA{\Gamma}$ as in Section~\ref{sec:background},
and for $p\in V$ define posets $(\calL^+_p,\calL^-_p,\calL^\pm_p)=(\calL^+_p(\Gamma),\calL^-_p(\Gamma),\calL^\pm_p(\Gamma))$
as in eq.~\eqref{eq:Lneib}. Note, in all cases $\sqsubseteq$ now has the new meaning (it is the relation that comes with $\Gamma$).
\begin{lemma}\label{lemma:GammaIdeal:ebp}
Let $\Gamma$ be an extended modular complex, and $p$ be its element. 
Then 
$\calL^+_p,\calL^-_p$ are modular semilattices that are convex in $\Gamma$.
Furthermore, function $v(\cdot)$ defined via $v(a)=\mu_\Gamma(p,a)$
is a valid valuation of these (semi)lattices. 
\end{lemma}

\begin{figure}[t]
\begin{center}
		\begin{tabular}{@{\hspace{15pt}}c@{\hspace{35pt}}c@{\hspace{35pt}}c} 
			  \includegraphics[scale=0.5,angle=-90,trim=164pt 280pt 190pt 480pt,clip]{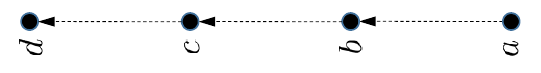}   &
			  \includegraphics[scale=0.5,angle=-90,trim=155pt 200pt 190pt 300pt,clip]{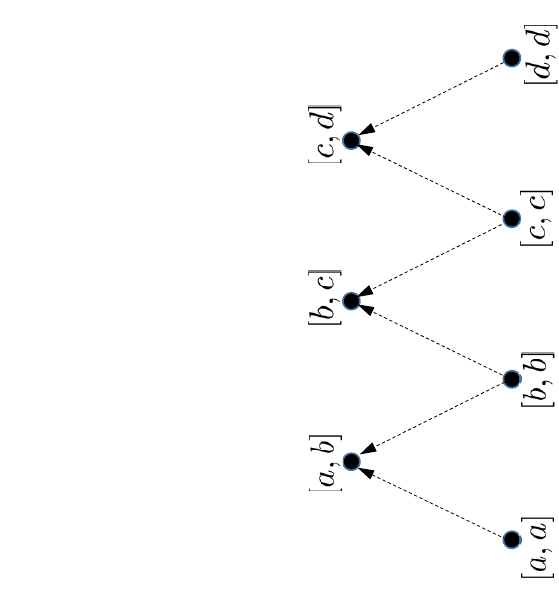} &
			  \includegraphics[scale=0.5,angle=-90,trim=155pt 200pt 190pt 300pt,clip]{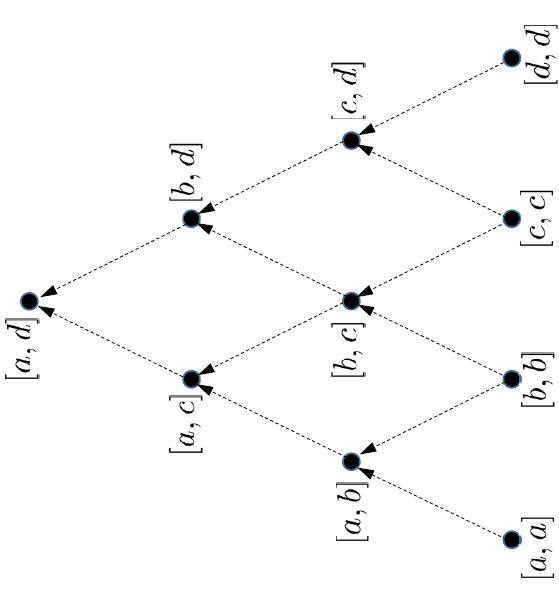}  
~ \vspace{10pt} \\
\small (a): original graph & \small (b): 2-subdivision with $\sqsubseteq\,=\BPrel$ & \small (c): 2-subdivision with $\sqsubseteq\,=\,\preceq$  \vspace{-14pt}
		\end{tabular}
	\end{center}
	\caption{Different 2-subdivisions for the chain $a\prec b\prec c \prec d$.}
\label{fig:subdivisions}
\end{figure}

We define a {\em 2-subdivision of $\Gamma$} exactly as in Section~\ref{sec:background};
it is the graph $\Gamma^\ast=(V^\ast,E^\ast,w^\ast)$
where $V^\ast=\{[p,q]\::\:p,q\in\Gamma,p\sqsubseteq_\Gamma q\}$.
Using condition~(\ref{def:ebp}a), it is straightforward to verify
that $(V^\ast,E^\ast)$ is the Hasse diagram of the poset $(V^\ast,\subseteq)$.
We will show the following result.
\begin{theorem}\label{th:subdivision:modular}
If $\Gamma$ is an extended modular complex then  $(V^\ast,E^\ast,w^\ast)$ is a modular complex
(i.e.\ an oriented 
modular
graph).
\end{theorem}
Recall that if $\Gamma$ is a modular complex (equivalently, an extended modular complex with the relation $\sqsubseteq_\Gamma=\BPrel_\Gamma$)
then $\Gamma^\ast$ is well-oriented, i.e.\ $\BPrel_{\Gamma^\ast}=\,\preceq_{\Gamma^\ast}$.
For extended modular complexes this is not necessarily the case,
as e.g.\ for the example in Fig.~\ref{fig:subdivisions}(c).
We will treat $\Gamma^\ast$ as an extended modular complex with the relation $\sqsubseteq_{\Gamma^\ast}\,=\,\preceq_{\Gamma^\ast}$,
so that we have $\calL_{[p,p]}^+(\Gamma^\ast)=\calL_{[p,p]}^\uparrow(\Gamma^\ast)$ by construction.
Let us set $\calL^\ast_p=\calL^\ast_p(\Gamma)=\calL_{[p,p]}^+(\Gamma^\ast)=\calL_{[p,p]}^\uparrow(\Gamma^\ast)$.
We can now define $L$-convex functions exactly as in the previous section:
\begin{itemize}
\item Function $f:V\rightarrow\overline{\mathbb R}$ is called {\em $L$-convex on extended modular complex $\Gamma$}
if (i) subset $\dom f\subseteq V$ is connected in $\Gamma^{\mathsmaller\sqsubset}$,
and (ii) for every $p\in V$, the restriction of $f^\ast$ to $\calL^\ast_{p}=\calL^\uparrow_{[p,p]}(\Gamma^\ast)$ is submodular on valuated
modular semilattice $\calL^\ast_{p}$.
\end{itemize}

\myparagraph{Cartesian products}
If $(\Gamma,\sqsubseteq),(\Gamma',\sqsubseteq')$ are two extended modular complexes, then their Cartesian product $(\Gamma,\sqsubseteq)\times(\Gamma',\sqsubseteq')$
is defined as the pair $(\Gamma\times\Gamma',\sqsubseteq_{\times})$
where $(p,p') \sqsubseteq_{\times} (q,q')$ iff $p \sqsubseteq q$ and $p' \sqsubseteq' q'$.

\begin{lemma}\label{lemma:ebp-cartesian}
Consider extended modular complexes $\Gamma,\Gamma'$ and element $(p,p')\in\Gamma\times\Gamma'$. 
\begin{itemize}[noitemsep,topsep=0pt]
\item[{\rm (a)}] $\Gamma\times \Gamma'$ is an extended modular complex. 
\item[{\rm (b)}] $\calL^\sigma_{(p,p')}(\Gamma\times\Gamma')=\calL^\sigma_{p}(\Gamma)\times \calL^\sigma_{p'}(\Gamma')$ for $\sigma\in\{-,+\}$. 
\end{itemize}
\end{lemma}

\begin{theorem}\label{th:ebp-addition}
Consider extended modular complexes $\Gamma,\Gamma'$ and functions $f,f':\Gamma\rightarrow \overline{\mathbb R}$ and $\tilde f:\Gamma\times\Gamma'\rightarrow \overline{\mathbb R}$. 
\begin{itemize}[noitemsep,topsep=0pt]
\item[{\rm (a)}] If $f,f'$ are $L$-convex  on $\Gamma$ then $f+f'$ and $c\cdot f$ for $c\in\mathbb R_+$ are also $L$-convex on $\Gamma$. 
\item[{\rm (b)}] If $f$ is $L$-convex on $\Gamma$ and $\tilde f(p,p')=f(p)$ for $(p,p')\in \Gamma\times\Gamma'$ then $\tilde f$ is $L$-convex on $\Gamma\times\Gamma'$. 
\item[{\rm (c)}] If $\tilde f$ is $L$-convex on $\Gamma\times\Gamma'$ and $f(p)=\tilde f(p,p')$ for fixed $p'\in\Gamma'$
then $f$ is $L$-convex on $\Gamma$. 
\item[{\rm (d)}]  The indicator function $\delta_U:V\rightarrow \{0,\infty\}$  is $L$-convex on $\Gamma$
in the following cases: (i) $U$ is a $d_\Gamma$-convex set; (ii) $U=\{p,q\}$ for elements $p,q$ with $p\sqsubseteq q$.
\item[{\rm (e)}]  Function $\mu_\Gamma:\Gamma\times \Gamma\rightarrow\mathbb R_+$ is $L$-convex on $\Gamma\times \Gamma$.
\item[{\rm (f)}]  If $f$ is $L$-convex on $\Gamma$ then the restriction of $f$ to $\calL^\sigma_p(\Gamma)$ is submodular on $\calL^\sigma_p(\Gamma)$ for every
$p\in\Gamma$ and $\sigma\in\{-,+\}$.
\end{itemize}
\end{theorem}

\begin{theorem}\label{th:Lopt}
Let $f:V\rightarrow\overline{\mathbb R}$ be an $L$-convex function on an extended modular complex~$\Gamma$.
If $p$ is a local minimizer of $f$ on $\calL^\pm_p(\Gamma)$, i.e. $f(p)=\min_{q\in \calL^\pm_p(\Gamma)}f(q)$, then it is also a global minimizer of~$f$, i.e.\ 
$f(p)=\min_{q\in V}f(q)$.
\end{theorem}

For extended modular complex $\Gamma$
let $\Phi_\Gamma$ be the language over domain $D=V_\Gamma$ that consists of all functions $f:D^n\rightarrow\overline{\mathbb R}$
such that $f$ is $L$-convex on $\Gamma^n$. The results above imply that any instance $\calI$ of $\Phi_\Gamma$
can be solved by the Steepest Descent Algorithm (Algorithm~\ref{alg:SDA}), and each subproblem in line 2 can be solved
in polynomial time. However, we do not know whether the number of steps would be polynomially bounded.
To get a polynomial bound, we introduce an  alternative algorithm, which we denote {\em $\DELTANEIB$-SDA}.

\begin{algorithm}[H]
  \DontPrintSemicolon
  pick arbitrary $x\in\dom f$ \\
  \While{\tt true}
      {
        compute $x^-\in\argmin \,\{f(y)\:|\:y\in\calL^-_x(\Gamma^n)\} $ and $x^+\in\argmin \,\{f(y)\:|\:y\in\calL^+_x(\Gamma^n)\} $ \\
        compute $\DELTA{x}\in\argmin \,\{f(y)\:|\:y\in\calL^+_{x^-}(\Gamma^n)\cap\calL^-_{x^+}(\Gamma^n)\} $  \\
        if $f(\DELTA{x}) = f(x)$ then return $x$, otherwise update $x:= \DELTA{x}$
      }
      \caption{$\DELTANEIB$-SDA for minimizing $f:\Gamma^n\rightarrow\overline{\mathbb R}$ }
      \label{alg:zigzagSDA}
\end{algorithm}
From Theorems~\ref{th:BLP-solves-submodular} and~\ref{th:ebp-addition}(f) we can conclude that points $x^-,x^+,\DELTA{x}$ in lines 3 and 4 can be computed in polynomial time via the BLP relaxation.
To see this for $\DELTA{x}$, observe that function $f+\delta_U$ for convex set $U=\calL^+_{x^-}\cap\calL^-_{x^+}$ is $L$-convex
on $\Gamma$ by Theorem~\ref{th:ebp-addition}(a,d), and thus its restriction to $\calL^+_{x^-}$ is submodular on $\calL^+_{x^-}$
by Theorem~\ref{th:ebp-addition}(f).
(Alternatively, the claim can be deduced from Theorem~\ref{th:BLPsolvesGamma} given later.)

\begin{theorem}\label{th:SDA}
Let $\Gamma$ be an extended modular complex and 
 $f:\Gamma^n\rightarrow\overline\RR$ be an $L$-convex function on $\Gamma^n$.
$\DELTANEIB$-SDA algorithm applied to function $f$ terminates after generating exactly
 $1+\max\limits_{i\in [n]}\DELTA{d}_\Gamma(x_i,{\tt opt}_i(f))$ distinct points,
where~$x$ is the initial vertex and ${\tt opt}_i(f)$ is as defined in Theorem~\ref{th:SDA:orig}.
\end{theorem}
\begin{remark}
Suppose that $\sqsubset_\Gamma\,=\,\preceq_\Gamma$, and the initial vertex $x$ in $\DELTANEIB$-SDA is computed
as follows: pick some $x_0\in\dom f$ and then set either $x\in\argmin\{f(y)\:|\:y\in\calL^-_{x_0}\}$
or $x\in\argmin\{f(y)\:|\:y\in\calL^+_{x_0}\}$. It can be seen
that in that case $\DELTANEIB$-SDA becomes equivalent to SDA:
we will have $(x^-,x^+,\DELTA{x})=(x,x^+,x^+)$ at even steps and $(x^-,x^+,\DELTA{x})=(x^-,x,x^-)$ at odd steps,
or vice versa.
Thus, Theorem~\ref{th:SDA} generalizes Theorem~\ref{th:SDA:orig} in two ways:
from modular complexes to extended modular complexes, and by allowing relations $\sqsubseteq_\Gamma$ and $\preceq_\Gamma$ to be distinct.

Note that the algorithm for $\ZeroExt{\mu}$ described in the previous section required applying SDA on 2-subdivision $\Gamma^\ast$.
This blows up the size of the graph and the size of LPs that need to be solved at each iteration by an up to a quadratic factor.
$\DELTANEIB$-SDA provides an alternative that avoids such blow-up.  
\end{remark}

\begin{remark}
Algorithm~\ref{alg:zigzagSDA} can also be specialized for minimizing $L^\natural$-convex functions (i.e.\ when $\Gamma$ a directed path on consecutive integers and $a\sqsubset b$ iff $b=a+1$).
In this case $\calL^-_x(\Gamma^n)=[x-{\bf 1},x]$, $\calL^+_x(\Gamma^n)=[x,x+{\bf 1}]$
and $\calL^+_{x^-}(\Gamma^n)\cap\calL^-_{x^+}(\Gamma^n)=U_1\times\ldots\times U_n$
where 
$$U_i=\begin{cases}
\{x_i-1,x_i\} & \mbox{if }(x^-_i,x^+_i)=(x_i-1,x_i) \\
\{x_i,x_i+1\} & \mbox{if }(x^-_i,x^+_i)=(x_i,x_i+1) \\
\{x_i\} & \mbox{otherwise}
\end{cases}.
$$
Note that the resulting algorithm is different from previously proposed versions of SDA~\cite{Murota:SDA:03,KolmogorovShioura:SDA:09,MurotaShioura:SDA:14}.
Thus, $\DELTANEIB$-SDA adds to the toolbox for $L^\natural$-convex minimization.
\end{remark}

We can now show the tractability part of Theorem~\ref{th:main}.
Suppose that graph $H_\mu$ is $F$-orientable modular for a metric space $(V,\mu)$ and subset $F\subseteq \binom{V}2$.
Choose an admissible orientation of $(H_\mu,F)$, and let $\Gamma$ be the corresponding extended modular complex 
with the relation $\sqsubseteq\;=\;\preceq$. Clearly, for any $\{x,y\}\in F$ we have either $x\preceq y$ or $y\preceq x$.
From Theorem~\ref{th:ebp-addition} we conclude that $\ZeroExt{\mu,F}\subseteq \Phi_\Gamma$, and so $\ZeroExt{\mu,F}$ can be solved in polynomial time by the $\DELTANEIB$-SDA algorithm.

More generally, this shows tractability of $\VCSP{\Phi_\Gamma}$ for an extended modular complex  $\Gamma$ assuming that
 a feasible solution of any $\Phi_\Gamma$-instance can be computed in polynomial time.
Our last result shows that $\VCSP{\Phi_\Gamma}$ is tractable even without this assumption.
\begin{theorem}\label{th:BLPsolvesGamma}
If $\Gamma$ is an extended modular complex then BLP relaxation solves $\Phi_\Gamma$, and an optimal solution of any $\Phi_\Gamma$-instance can be computed in polynomial time.
\end{theorem}

Note that previously Hirai remarked that BLP directly solves  $\ZeroExt{\mu}$ for orientable modular metrics
only assuming that ${\tt P}\ne {\tt NP}$, as a consequence of VCSP dichotomy for finite-valued
languages (see Section 6 in~\cite{Hirai:0ext}). Theorem~\ref{th:BLPsolvesGamma} now establishes this fact unconditionally.

All proofs are given in Sections~\ref{sec:proofs} and~\ref{sec:VCSP}. For Lemmas~\ref{lemma:GammaIdeal:ebp}, \ref{lemma:ebp-cartesian}
and Theorems~\ref{th:subdivision:modular}, \ref{th:ebp-addition},  \ref{th:Lopt}
we mostly follow the proofs of the corresponding claims in~\cite{Hirai:0ext,Chalopin,Hirai:Lconvexity}
(replacing properties of relation $\BPrel$ with the properties of an admissible relation~$\sqsubseteq$).
To analyze $\DELTANEIB$-SDA (Theorem~\ref{th:SDA}), we introduce new concepts, such as binary operations $\UP,\DOWN,\diamond$
and the notion of $f$-extremality (see Sections~\ref{sec:NormalPath} and~\ref{sec:f-extremal}).
 Theorem~\ref{th:BLPsolvesGamma} does not have an analogue in~\cite{Hirai:0ext,Chalopin,Hirai:Lconvexity}.

%%%%%%%%%%%%%%%%%%%%%%%%%%%%%%%%%%%%%%%%%%%%%%%%%%%%%%%%%%%%%%%%%%%%%%%%%%%%%%%%%%%%%%%%%%%%%%%%%%%%%%%%%%%%%%%%%%%%%%%%%%%%%%
%%%%%%%%%%%%%%%%%%%%%%%%%%%%%%%%%%%%%%%%%%%%%%%%%%%%%%%%%%%%%%%%%%%%%%%%%%%%%%%%%%%%%%%%%%%%%%%%%%%%%%%%%%%%%%%%%%%%%%%%%%%%%%
%%%%%%%%%%%%%%%%%%%%%%%%%%%%%%%%%%%%%%%%%%%%%%%%%%%%%%%%%%%%%%%%%%%%%%%%%%%%%%%%%%%%%%%%%%%%%%%%%%%%%%%%%%%%%%%%%%%%%%%%%%%%%%
%%%%%%%%%%%%%%%%%%%%%%%%%%%%%%%%%%%%%%%%%%%%%%%%%%%%%%%%%%%%%%%%%%%%%%%%%%%%%%%%%%%%%%%%%%%%%%%%%%%%%%%%%%%%%%%%%%%%%%%%%%%%%%
%%%%%%%%%%%%%%%%%%%%%%%%%%%%%%%%%%%%%%%%%%%%%%%%%%%%%%%%%%%%%%%%%%%%%%%%%%%%%%%%%%%%%%%%%%%%%%%%%%%%%%%%%%%%%%%%%%%%%%%%%%%%%%

\section{Submodular functions on a valuated modular semilattice}\label{sec:submodularity-on-semilattice}
Let $\calL$ be a valuated modular semilattice with valuation $v$.
This section gives the definition of a {\em submodular function on $\calL$},
and thus completes the definition of $L$-convex functions.

Let $\Gamma=(V,E,w)$ be the Hasse diagram of $\calL$ where edge $p\rightarrow q$ is assigned weight $v(q)-v(p)>0$.
As discussed in Section~\ref{sec:background}, graph $\Gamma$ is oriented modular.
Denote $\mu=\mu_\Gamma$ and $d=d_\Gamma$.
 Recall that $\calL$ is viewed as a metric space with the metric $\mu$,
and the definitions of the metric interval $I(p,q)$, convex sets, etc would be the same for $(\calL,\mu)$ and $(\calL,d)$.

For $p\preceq q$ let us denote $v[p,q]=v(q)-v(p)$. The following facts are known.
\begin{lemma}[{\cite[Lemma 2.15]{Hirai:0ext}}]\label{lemma:semilattice:technical}
The following holds for each $p,q\in\calL$ with $s=p\wedge q$. \\
{\rm (a)} $\mu(p,q)=\mu(s,p)+\mu(s,q)=v[s,p]+v[s,q]$. \\
{\rm (b)} The metric interval $I(p,q)$ is equal to the set of elements $u$ that is represented as $u = a \vee b$ 
    for some $(a,b) \in [s,p] \times [s,q]$, where such a representation is unique, and $(a,b)$ equals $(u \wedge p, u \wedge q)$. \\
{\rm (c)} For $u,u'\in I(p,q)$ there holds $u\wedge u'=(u\wedge u'\wedge p)\vee (u\wedge u'\wedge q)$; in particular $u\wedge u'\in I(p,q)$.
\end{lemma}

The construction in~\cite{Hirai:0ext} can be described as follows
(see Figure~\ref{fig:Ipq}(a) for a conceptual diagram).

\begin{figure}[t]
\begin{center}
		\begin{tabular}{@{\hspace{25pt}}c@{\hspace{35pt}}c@{\hspace{35pt}}c@{\hspace{35pt}}c} 
			\includegraphics[scale=0.5,angle=90]{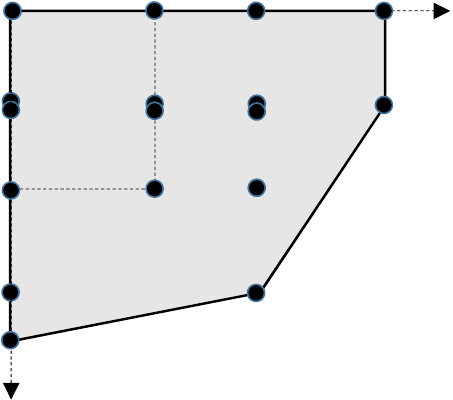} 
			\begin{picture}(0, 0)
				\put(-20,-6){\footnotesize $p$} 
				\put(-108,-6){\footnotesize $p\wedge q$} 
				\put(-106,91){\footnotesize $q$} 

				\put(-57,41){\footnotesize $u$} 
				\put(-64,-6){\footnotesize $p\wedge u$} 
				\put(-122,35){\footnotesize $q\wedge u$} 

				\put(-17,6){\footnotesize $u_0$} 
				\put(-28,64){\footnotesize $u_1$} 
				\put(-82,85){\footnotesize $\ldots$} 
				\put(-95,96){\footnotesize $u_k$} 
			\end{picture}
			\raisebox{26pt}{\hspace{-46pt} \includegraphics[scale=0.5,trim=105pt 80pt 88pt 80pt,clip]{diagram1-eps-converted-to.pdf}}
&
			\includegraphics[scale=0.5,angle=90]{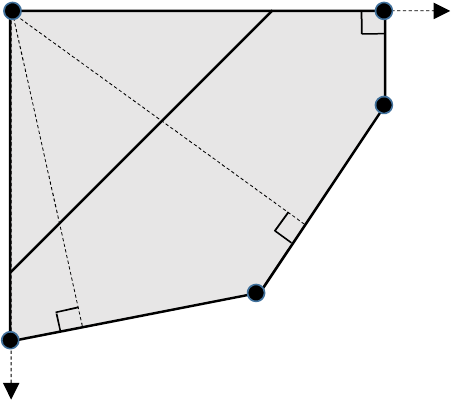} 
			\begin{picture}(0, 0)
				\put(-17,6){\footnotesize $u_0$} 
				\put(-28,64){\footnotesize $u_1$} 
				\put(-82,85){\footnotesize $\ldots$} 
				\put(-95,96){\footnotesize $u_k$} 
				\put(-40,-3){\footnotesize $\alpha_{-1}$} 
				\put(-48,18){\footnotesize $\alpha_0$} 
				\put(-109,63){\footnotesize $\alpha_k$} 
				\put(-96,66){\footnotesize $\alpha_{k-1}$} 
			\end{picture}
&
			\includegraphics[scale=0.5,angle=90]{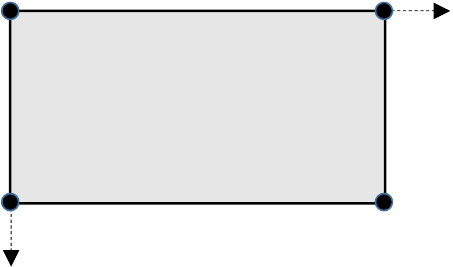} 
			\begin{picture}(0, 0)
				\put(-21,-6){\footnotesize $p$} 
				\put(-73,91){\footnotesize $q$} 
				\put(-16,90){\footnotesize $p\vee q$} 
				\put(-75,-6){\footnotesize $p\wedge q$} 
			\end{picture}
&
			\includegraphics[scale=0.5,angle=90]{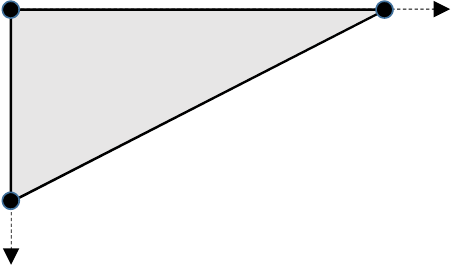} 
			\begin{picture}(0, 0)
				\put(-21,-6){\footnotesize $p$} 
				\put(-73,91){\footnotesize $q$} 
				\put(-75,-6){\footnotesize $p\wedge q$} 
			\end{picture}
~ \vspace{10pt} \\
\small (a) & \small (b) & \small (c) & \small (d) \vspace{-14pt}
		\end{tabular}
	\end{center}
	\caption{(a) Each point in $I(p,q)$ is assigned a coordinate in $\mathbb R^2$. The convex hull of these coordinates ($\Conv I(p,q)$) is in gray.
	Distinct points may have the same coordinates (as some of the points shown in the interior of the gray region), but the coordinates of points in $\calE(p,q)=\{u_0,u_1,\ldots,u_k\}$ are guaranteed to be unique. 
	(b)~Definition of $\{\theta_i\}_i$ and $p\vee_\theta q$. First, define points $\alpha_{-1},\alpha_0,\ldots,\alpha_k$ in $\mathbb R^2$ as follows:
	set $\alpha_{-1}=(\sqrt{2}/2,0)$, $\alpha_{k}=(0,\sqrt{2}/2)$ (so that $||\alpha_k-\alpha_{-1}||=1$),
	and for $i\in[k-1]$ let $\alpha_i$ be the intersection of segment $[\alpha_{-1},\alpha_k]$ and the line that goes through the origin and is perpendicular to the line
	passing through points $v_{pq}(u_{k})$ and $v_{pq}(u_{k+1})$.
	Then $\theta_i=||\alpha_i-\alpha_{-1}||$ and $p\wedge_\theta q=u_i$ for each $i\in[0,k]$ and $\theta\in(\theta_{i-1},\theta_i)$.
	(c)~Bounded pair $(p,q)$. 
	(d)~Antipodal pair $(p,q)$. 
        }
\label{fig:Ipq}
\end{figure}

\begin{itemize}
	\item[{\rm (i)}] For $u \in I(p,q)$, let $v_{pq}(u)$ be the vector in $\mathbb R^2_+$ defined by
	\begin{equation}\label{eq:vpqu}
	v_{pq}(u) = (v[s,u \wedge p], v[s,u \wedge q])\qquad\mbox{where~~~~} s=p\wedge q
	\end{equation}
	\item[{\rm (ii)}] Let $\Conv I(p,q)$ denote 
	the convex hull of vectors $v_{pq}(u)$ for all $u \in I(p,q)$.
	\item[{\rm (iii)}] Let $\calE(p,q)$ be the set of elements $u$ in $I(p,q)$
	such that $v_{pq}(u)$ is a maximal extreme point of $\Conv I(p,q)$. (This set is called ``$(p,q)$-envelope'').
	Note that $p,q\in\calE(p,q)$. Hirai proves that elements of $\calE(p,q)$ receive unique coordinates  \cite[Lemma 3.1]{Hirai:0ext}:
	\begin{equation}\label{eq:v:uniqueness}
	v_{pq}(u)\ne v_{pq}(u') \qquad \forall u\in\calE(p,q),u'\in I(p,q)-\{u\}
%	Then for any $u\in \calE(p,q)$ and $u'\in I(p,q)-\{u\}$ one has $v_{pq}(u)\ne v_{pq}(u')$ $\mbox{\cite[Lemma 3.1]{Hirai:0ext}}$.
	\end{equation}
	\item[{\rm (iv)}] For $\theta\in[0,1]$ define vector $c_\theta=(1-\theta,\theta)$. For points $p,q\in\calL$ let $p\vee_\theta q$
	be the point $u\in I(p,q)$ that maximizes $\langle c_\theta,v_{pq}(u)\rangle$, assuming that the maximizer is unique.
	If a maximizer is not unique then instead set $p\vee_\theta q=\perp$ (``undefined'').
	\item[{\rm (v)}] By the property in eq.~\eqref{eq:v:uniqueness}, there are only a finite number of values $\theta\in[0,1]$
	such that $p\vee_\theta q=\perp$ for some $p,q\in\calL$. Let $\Theta$ be the set of values $\theta\in[0,1]$ such that $p\vee_\theta q\ne \perp$ for all $p,q\in\calL$,
	then $\Theta\subseteq[0,1]$ has measure 1.
\end{itemize}

We are now ready to define submodular functions on $\calL$.
A function $f:\calL \rightarrow \overline{\mathbb R}$ is called {\em submodular} if
every $p,q\in\calL$ satisfies 
\begin{subequations}\label{eq:submodular}
\begin{equation}\label{eq:submodular:a}
f(p) + f(q) \ge f(p \wedge q) + \int_{\theta \in \Theta} f(p\vee_\theta q)d\theta
\end{equation}
This is equivalent to the inequality
\begin{equation}\label{eq:submodular:b}
f(p) + f(q) \ge f(p \wedge q) + \sum_{i=0}^k (\theta_i-\theta_{i-1}) f(u_i)
\end{equation}
\end{subequations}
where $u_0,\ldots,u_k$ are the sorted elements of $\calE(p,q)$ with $(u_0,u_k)=(p,q)$ and $(\theta_{-1},\theta_k)=(0,1)$,
$$
\theta_i=\frac{v[s_{i},u_i]}{v[s_i,u_i]+v[s_i,u_{i+1}]}\qquad\mbox{where~~~} s_i=u_i\wedge u_{i+1}
$$
for $i\in[0,k-1]$. (See Fig.~\ref{fig:Ipq}(b) for a geometric interpretation of values $\theta_i$).

Hirai also gives a simplified characterization of submodularity.
Pair $(p,q)$ is called {\em bounded} if $p\vee q$ exists, implying $\calE(p,q)=\{p,q,p\vee q\}$
(Fig.~\ref{fig:Ipq}(c)).
Pair $(p,q)$ is called {\em antipodal} if $\calE(p,q)=\{p,q\}$ (Fig.~\ref{fig:Ipq}(d)).
Equivalently, $(p,q)$ is antipodal if and only if %every bounded pair $(a,b)$ with $p \succeq a \succeq p \wedge q \preceq b \preceq q$ satisfies
\begin{equation}\label{eq:antipodal:def}
v[a,p] v[b,q] \ge v[p \wedge q, a] v[p \wedge q, b]
\qquad \forall (a,b):p \succeq a \succeq p \wedge q \preceq b \preceq q,a\vee b\mbox{ exists}
\end{equation}
We say that $(p,q)$ is {\em special} if it is either bounded or antipodal.
For special pairs inequality~\eqref{eq:submodular} reduces to
\begin{subequations}\label{eq:submodular:special}
\begin{align}
f(p) + f(q) &\ge f(p \wedge q) + f(p\vee q)  && \mbox{if $(p,q)$ bounded} \label{eq:submodular:bounded} \\
v[p \wedge q,q] f(p) +  v[p \wedge q,p] f(q) &\ge (v[p \wedge q,p] + v[p \wedge q,q])f(p \wedge q) && \mbox{if $(p,q)$ antipodal} \label{eq:submodular:antipodal}
\end{align}
These inequalities are called respectively {\em submodularity} and {\em $\wedge$-convexity} inequalities for $(p,q)$.
\end{subequations}
\begin{theorem}[{\cite[Theorem 3.5]{Hirai:0ext}}]\label{th:submodularity-characterization}
Function $f:\calL\rightarrow\overline{\mathbb R}$ is submodular if and only if it satisfies
\begin{itemize}
\item [{\rm (1)}] $\calE(p,q) \subseteq \dom f$ for $p,q \in \dom f$;
\item [{\rm (2)}] the submodularity inequality for every bounded pair $(p,q)$;
\item [{\rm (3)}] the $\wedge$-convexity inequality for every antipodal pair $(p,q)$.
\end{itemize}
\end{theorem}

We observe that condition {\rm (1)} can be strengthened further (though we will not use this observation). For elements $p,q,u,u'\in\calL$ 
we write $(p,q)\lhd (u,u')$ if $u,u'\in I(p,q)-\{p,q\}$ and point $\frac{1}{2}(v_{pq}(u)+v_{pq}(u'))$ lies strictly above the segment $[v_{pq}(p),v_{pq}(q)]$ in $\mathbb R^2$.
Note that all elements $u\in\calE(p,q)-\{p,q\}$ satisfy $(p,q)\lhd (u,u)$,
and pair $(p,q)$ is antipodal if and only if there is no $u\in\calL$ with $(p,q)\lhd (u,u)$.
The following theorem is proven in~\cite[Appendix A]{DK:arxiv:v3}.

\begin{theorem}\label{th:submodularity-strong-characterization}
Function $f:\calL\rightarrow\overline{\mathbb R}$ satisfies condition {\rm (1)} of Theorem~\ref{th:submodularity-characterization}
if and only if it satisfies the following condition:
\begin{itemize}
\item [{\rm (1$'$)}] Suppose that $p,q\in\dom f$ and $\calE(p,q)\ne\{p,q\}$. Let $u^-$ and $u^+$ be the elements
in $\calE(p,q)-\{p,q\}$ that are closest to $p$ and to $q$, respectively. Then there exists $t\in \dom f$ such
that either $(p,q)\lhd(u^-,t)$ or $(p,q)\lhd(u^+,t)$.
\end{itemize}
\end{theorem}

%%%%%%%%%%%%%%%%%%%%%%%%%%%%%%%%%%%%%%%%%%%%%%%%%%%%%%%%%%%%%%%%%%%%%%%%%%%%%%%%%%%%%%%%%%%%%%%%%%%%%%%%%%%%%%%%%%%%%%%
%%%%%%%%%%%%%%%%%%%%%%%%%%%%%%%%%%%%%%%%%%%%%%%%%%%%%%%%%%%%%%%%%%%%%%%%%%%%%%%%%%%%%%%%%%%%%%%%%%%%%%%%%%%%%%%%%%%%%%%
%%%%%%%%%%%%%%%%%%%%%%%%%%%%%%%%%%%%%%%%%%%%%%%%%%%%%%%%%%%%%%%%%%%%%%%%%%%%%%%%%%%%%%%%%%%%%%%%%%%%%%%%%%%%%%%%%%%%%%%
%%%%%%%%%%%%%%%%%%%%%%%%%%%%%%%%%%%%%%%%%%%%%%%%%%%%%%%%%%%%%%%%%%%%%%%%%%%%%%%%%%%%%%%%%%%%%%%%%%%%%%%%%%%%%%%%%%%%%%%
%%%%%%%%%%%%%%%%%%%%%%%%%%%%%%%%%%%%%%%%%%%%%%%%%%%%%%%%%%%%%%%%%%%%%%%%%%%%%%%%%%%%%%%%%%%%%%%%%%%%%%%%%%%%%%%%%%%%%%%

\section{Proofs}\label{sec:proofs}

Throughout the proofs we usually denote $d,\mu$ to be the distance functions in an extended modular complex~$\Gamma$,
%$d^\ast,\mu^\ast$ to be the distance functions in its 2-subdivision $\Gamma^\ast$,
and $d^{\mathsmaller\sqsubset}$ and $\DELTA{d}$ to be the distance functions in $\Gamma^{\mathsmaller\sqsubset}$ and in $\DELTA{\Gamma}$ respectively.
We write $xy\in \Gamma^{\mathsmaller\sqsubset}$ to indicate that nodes $x,y$ are neighbors in $\Gamma^{\mathsmaller\sqsubset}$
(equivalently, that either $x\sqsubset y$ or $y\sqsubset x$). Note that $xy\in \Gamma^{\mathsmaller\sqsubset}$ if and only if $d^{\mathsmaller\sqsubset}(x,y)=1$.

A sequence $(u_0,u_1,\ldots,u_k)$ of nodes in $\Gamma$ is called a {\em shortest subpath}
if $\mu(u_0,u_k)=\mu(u_0,u_1)+\mu(u_1,u_2)+\ldots+\mu(u_{k-1},u_k)$.
Recall that for a modular graph $\Gamma$ this is equivalent to the condition
$d(u_0,u_k)=d(u_0,u_1)+d(u_1,u_2)+\ldots+d(u_{k-1},u_k)$.
We will often implicitly use the following fact.
\begin{proposition}\label{prop:wedge-shortest}
Consider elements $p,q$ in a modular complex $\Gamma$.
{\rm (a)} If $p\wedge q$ exists then $(p,p\wedge q,q)$ is a shortest subpath.
Conversely, if $p\succeq x\preceq q$ and $(p,x,q)$ is a shortest subpath then $x=p\wedge q$.
{\rm (b)} If $p\vee q$ exists then $(p,p\vee q,q)$ is a shortest subpath.
Conversely, if $p\preceq x\succeq q$ and $(p,x,q)$ is a shortest subpath then $x=p\vee q$.
\end{proposition}
\begin{proof}
To see the claim, combine Lemma~\ref{lemma:GammaIdeal} and Lemma~\ref{lemma:semilattice:technical}(a).
\end{proof}

We say that sequence $(x_1,x_2,x_3,x_4)$ is an {\em isometric rectangle} if $(x_{i-1},x_i,x_{i+1})$ 
are shortest subpaths for all $i\in[4]$, where $x_j=x_{j\!\!\mod 4}$. Isometric rectangles have the following properties.
\begin{proposition}\label{prop:isometric-rectangle}
If $(x_1,x_2,x_3,x_4)$ is an isometric rectangle in graph $\Gamma$ in then $d(x_1,x_2)=d(x_3,x_4)$ and $d(x_1,x_4)=d(x_2,x_3)$.
Furthermore, if graph $\Gamma$ is modular then for any $y_1\in I(x_1,x_4)$ there exists $y_2=I(x_2,x_3)$
(namely, a median of $y_1,x_2,x_3$)
such that sequences $(x_1,x_2,y_2,y_1)$ and $(y_1,y_2,x_3,x_4)$ are isometric rectangles.
\end{proposition}
\begin{proof}
The definition of an isometric rectangle implies that $d(x_1,x_2)+d(x_2,x_3)=d(x_4,x_1)+d(x_3,x_4)$
and $d(x_2,x_3)+d(x_3,x_4)=d(x_1,x_2)+d(x_4,x_1)$,
which in turns implies  that $d(x_1,x_2)=d(x_3,x_4)$ and $d(x_1,x_4)=d(x_2,x_3)$.
Now suppose that $y_1\in I(x_1,x_4)$ and $y_2$ is a median of $y_1,x_2,x_3$.
$(x_1,y_1,x_4,x_3)$ and $(y_1,y_2,x_3)$ are shortest subpaths,
and thus so is $(x_1,y_1,y_2,x_3)$. By a symmetric argument, $(x_4,y_1,y_2,x_2)$ is also a shortest subpaths.
This implies the claim.
\end{proof}

\subsection{Preliminaries}\label{sec:proofs-prelim}
First, we state some known facts about orientable modular graphs that will be needed later in the proofs.

\begin{lemma}[{\cite[Proposition 1.7]{Bandelt:93}, see \cite[Proposition 6.2.6, Chapter I]{vanDeVel:book}}]\label{lemma:modular:characterization}
A connected graph $\Gamma$ with distances $d=d_\Gamma$ is modular if and only if 
\begin{itemize}[noitemsep,topsep=0pt]
\item[{\rm (1)}] $\Gamma$ is bipartite, and 
\item[{\rm (2)}] for vertices $p,q$ and neighbors $p_1,p_2$ of $p$ with $d(p_1,q)=d(p_2,q)=d(p,q)-1$, 
there exists a common neighbor $p^\ast$ of $p_1,p_2$ with $d(p^\ast, q)=d(p,q)-2$.
\end{itemize}
\end{lemma}

\myparagraph{Convex sets and gated sets} 
Consider metric space $(V,\mu)$. 
Recall that subset $U\subseteq V$ is called {\em convex} if $I(x,y)\subseteq U$
for every $x,y\in U$.  
It is called {\em gated}
if for every $p\in V$ there exists
unique $p^\ast \in U$, called the {\em gate of $p$ at $U$}, such that $\mu(p,q)=\mu(p,p^\ast)+\mu(p^\ast,q)$
holds for every $q\in U$. The gate $p^\ast$ will be denoted as $\Pr_U(p)$ (``projection of $p$ onto $U$'').
Thus, $\Pr_U$ is a map $V\rightarrow U$.

\begin{theorem}[{\cite{DressScharlau:87}}]\label{th:gatedset}
Let $A$ and $A'$ 
be gated subsets of $(V,\mu)$ and 
let $B := \Pr_{A}(A')$ and $B' := \Pr_{A'}(A)$. 
\begin{itemize}[noitemsep,topsep=0pt]
\item[{\rm (a)}] $\Pr_{A}$ and $\Pr_{A'}$ induce isometries, inverse to each other, between $B'$ and $B$. 
\item[{\rm (b)}] For $p \in A$ and $p' \in A'$, the following conditions are equivalent:
\begin{itemize}[noitemsep,topsep=0pt]
\item[{\rm (i)}] $\mu(p,p') = \mu(A,A')$. 
\item[{\rm (ii)}] $p = \Pr_{A}(p')$ and $p' = \Pr_{A'}(p)$.
\end{itemize}
\item[{\rm (c)}] $B$ and $B'$ are gated, and $\Pr_{B} = \Pr_{A} \circ \Pr_{A'}$ and $\Pr_{B'} = \Pr_{A'} \circ \Pr_{A}$.
\end{itemize}
\end{theorem}

Now let us consider a modular graph $\Gamma=(V,E,w)$. Such graph induces two natural metrics on $V$,
namely $\mu=\mu_\Gamma$ and $d=d_\Gamma$ (shortest path metrics w.r.t.\ edge lengths $w$ and $1$, respectively).
In the light of Theorem~\ref{th:orbits}, the definitions of convex sets, gated sets, gates and maps $\Pr_U$
would be the same for both metric spaces ($(V,\mu)$ and $(V,d)$). In addition, convex and gated sets for such metrics coincide.

\begin{lemma}[{\cite{Chepoi:89}, see \cite[Lemma 2.9]{Hirai:0ext}}]\label{lemma:convex:gated}
Let $\Gamma$ be a modular graph. % with metric $\mu_\Gamma$. 
For $U\subseteq V$, the following conditions are equivalent: 
\begin{itemize}[noitemsep,topsep=0pt]
\item[{\rm (1)}] $U$ is convex. 
\item[{\rm (2)}] $U$ is gated. 
\item[{\rm (3)}] $\Gamma[U]$ is connected and $I(x,y)\subseteq U$ holds for every $p,q\in U$ with $d_\Gamma(p,q)=2$. 
\end{itemize}
\end{lemma}

Next, we review properties of modular complexes. 
A path $(p_0,p_1,p_2,\ldots,p_k)$ in a directed acyclic graph $\Gamma=(V,E,w)$ is said to be {\em ascending} 
if $p_0 \prec p_1\prec\ldots\prec p_k$.
\begin{lemma}[{\cite[Lemma 4.13]{Hirai:0ext}}]\label{lemma:ascending}
Let $\Gamma$ be a modular complex.
For $p,q \in V$ with $p \preceq q$,
a $(p,q)$-path $P$ is shortest 
if and only if $P$ is an ascending path from $p$ to $q$.
In particular, $I(p,q) = [p,q]$,  any maximal chain in $[p,q]$ 
has the same length, and the rank $r$ of $[p,q]$ is given by $r(a) = d(a,p)$.
\end{lemma}

Since set $[p,q]$ is convex in $\Gamma$ by Lemma~\ref{lemma:GammaIdeal}, one can define  projection $\Pr_{[p,q]}: V \rightarrow [p,q]$.
The lemma below describes some properties of this projection.
\begin{lemma}[{\cite[Lemma 4.15]{Hirai:0ext}}]\label{lemma:pq-pq-prime}
Let $\Gamma$ be a modular complex.
For elements $p,q,p',q'$ with $p \preceq q$ and $p' \preceq q'$ define
\begin{equation}\label{eq:uvu'v'}
u = \Pr\nolimits_{[p, q]} (p') \quad
v = \Pr\nolimits_{[p,q]} (q') \quad
u' = \Pr\nolimits_{[p',q']} (p) \quad
v' = \Pr\nolimits_{[p',q']} (q)
\end{equation}
Then we have:
\begin{itemize}[noitemsep,topsep=0pt]
\item[{\rm (1)}]  $u \preceq v$, $u' \preceq v'$, $\Pr_{[p',q']}([p,q]) = [u', v']$, and $\Pr_{[p,q]}([p',q']) = [u, v]$.
\item[{\rm (2)}] $[u,v]$ is isomorphic to $[u',v']$ by map $w \mapsto \Pr_{[p',q']}(w)$. 
\end{itemize}
\end{lemma}

%%%%%%%%%%%%%%%%%%%%%%%%%%%%%%%%%%%%%%%%%%%%%%%%%%%%%%%%%%%%%%%%%%%%%%%%%%%%%%%%%%%%%%%%%%%%%%%%%%%%%%%%%%%%%%%%%%%%%%%%%
%%%%%%%%%%%%%%%%%%%%%%%%%%%%%%%%%%%%%%%%%%%%%%%%%%%%%%%%%%%%%%%%%%%%%%%%%%%%%%%%%%%%%%%%%%%%%%%%%%%%%%%%%%%%%%%%%%%%%%%%%
%%%%%%%%%%%%%%%%%%%%%%%%%%%%%%%%%%%%%%%%%%%%%%%%%%%%%%%%%%%%%%%%%%%%%%%%%%%%%%%%%%%%%%%%%%%%%%%%%%%%%%%%%%%%%%%%%%%%%%%%%

\subsection{Properties of extended modular complexes}
In this section 
$\Gamma$ is always assumed to be an extended modular complex on nodes $V$ with relation~$\sqsubseteq$.

\begin{lemma}\label{lemma:common-neighbor}
Suppose that $p\sqsubseteq q_1$, $p\sqsubseteq q_2$, $q_1\ne q_2$
and $q$ is a common neighbor of $q_1,q_2$ (implying that $d(q_1,q_2)=2$).
Then $p\sqsubseteq q$.
\end{lemma}
\begin{proof}
Modulo symmetry, three cases are possible.
\begin{itemize}
\item $q_1\rightarrow q\rightarrow q_2$. Since $p\sqsubseteq q_2$, condition~(\ref{def:ebp}a) gives $p\sqsubseteq q$.
\item $q_1\rightarrow q\leftarrow q_2$. Then $q=q_1\vee q_2$ by Lemma~\ref{lemma:GammaIdeal}(a), and condition~(\ref{def:ebp}b) gives $p\sqsubseteq q$.
\item $q_1\leftarrow q\rightarrow q_2$.  Then $q=q_1\wedge q_2$ by Lemma~\ref{lemma:GammaIdeal}(a), implying $p\preceq q$ (since $p$ lower-bounds $q_1,q_2$). Condition~(\ref{def:ebp}a) gives $p\sqsubseteq q$. \qedhere
\end{itemize}
\end{proof}

\begin{lemma}\label{lemma:romb}
Consider elements $p,q,a,b$ such that $(p,a,b)$ and $(a,b,q)$ are shortest subpaths and $p\sqsubseteq q$. Then $a\sqsubseteq b$.
\end{lemma}
\begin{proof}
We use induction on $d(p,q)+d(p,a)+d(q,b)$. First, assume that $(a,p,q)$ is not a shortest subpath.
Let $p'$ be a median of $a,p,q$, then $p'\in[p,q]-\{p\}$ and so $p'\sqsubseteq q$.
The induction hypothesis for $p',q,a,b$ gives the claim. We can thus assume that $(a,p,q)$ is a shortest subpath.
By a symmetric argument we can also assume that $(b,q,p)$ is a shortest subpath.
Sequence $(p,q,b,a)$ is thus an isometric rectangle, and so $d(p,q)=d(a,b)$ and $d(p,a)=d(q,b)$.
We now consider 5 possible cases.
\begin{itemize}
\item $d(p,a)=d(q,b)=0$. The claim then holds trivially.
\item $d(p,a)=d(q,b)=1$, $q\rightarrow b$.
Then $p\preceq b$ and $(p,a,b)$ is a shortest subpath, and so $p\rightarrow a\preceq b$ by Lemma~\ref{lemma:ascending}.
We have $a\vee q\preceq b$ and $(a,b,q)$ is a shortest subpath; this implies that $a\vee q=b$.
Since $p\sqsubseteq a$ and $p\sqsubseteq q$, we obtain $p\sqsubseteq a\vee q=b$ by condition~(\ref{def:ebp}b), and thus $a\sqsubseteq b$ since $a\in[p,b]$.
\item $d(p,a)=d(q,b)=1$, $p\leftarrow a$. This case is symmetric to the previous one.
\item $d(p,a)=d(q,b)=1$, $p\rightarrow a$, $b\rightarrow q$. We will show that $a\preceq b$;
this will imply that  $(p,a,b,q)$ is a shortest subpath by Lemma~\ref{lemma:ascending}, contradicting condition $d(p,q)=d(a,b)$.

If $d(p,q)=d(a,b)=0$ then the claim is trivial. If $d(p,q)=d(a,b)=1$ then $(p,q,b,a)$ is a 4-cycle,
and so $p\rightarrow q$ implies that $a\rightarrow b$.
Suppose that $d(p,q)=d(a,b)\ge 2$.
By Proposition~\ref{prop:isometric-rectangle},
there exist $x\in I(p,q)-\{p,q\}$ and $y\in I(a,b)$ such that $(p,x,y,a)$ and $(q,x,y,b)$
are isometric rectangles. We have $p\sqsubseteq x\sqsubseteq q$,
so by the induction hypothesis  $a\sqsubseteq y\sqsubseteq b$, implying the claim.

\item $d(p,a)=d(q,b)\ge 2$. By Proposition~\ref{prop:isometric-rectangle},
there exist $p'\in I(p,a)$ and $q'\in I(q,b)$
such that $(p,q,q',p')$ and $(p',q',b,a)$ are isometric rectangles. 
The induction hypothesis for elements $p,q,p',q'$ gives $p'\sqsubseteq q'$.
The induction hypothesis for elements $p',q',a,b$ gives $a\sqsubseteq b$. \qedhere
\end{itemize}
\end{proof}

\begin{corollary}\label{cor:projection-preserves-ebp}
Suppose that $p\sqsubseteq q$ and $A'\subseteq V$ is a convex set in $\Gamma$.
Let $p'=\Pr_{A'}(p)$ and $q'=\Pr_{A'}(q)$.
Then $p'\sqsubseteq q'$.
\end{corollary}

\renewcommand{\thelemmaRESTATED}{\ref{lemma:GammaIdeal:ebp}}
\begin{lemmaRESTATED}[restated]
Let $\Gamma$ be an extended modular complex, and $p$ be its element. 
Then 
$\calL^+_p,\calL^-_p$ are modular semilattices that are convex in $\Gamma$.
Furthermore, function $v(\cdot)$ defined via $v(a)=\mu_\Gamma(p,a)$
is a valid valuation of these (semi)lattices. 
\end{lemmaRESTATED}
\begin{proof}
By symmetry, it suffices to consider only $\calL^+_p$. Note that $\calL^+_p$ is a subset of $\calL^\uparrow_p$
such that for any $a,b\in\calL_p^+$ we have (i) $a\wedge b\in \calL_p^+$ and (ii) $a\vee b\in\calL_p^+$ assuming that $a\vee b$ exists (by Definition~\ref{def:ebp}).
Thus, all claims except convexity follow from the corresponding properties of $\calL^\uparrow_p$ (see Lemma~\ref{lemma:GammaIdeal}).
 Let us show the convexity.
In view of Lemma~\ref{lemma:convex:gated}, it suffices to show that for any $a,b\in\calL_p^+$ with $d(a,b)=2$ we have $I(a,b)\subseteq \calL^+_p$.
This claim follows directly from Lemma~\ref{lemma:common-neighbor}.
\end{proof}

In the next two results we use the same notation $d$ for $d_\Gamma$ and $d_{\Gamma^\ast}$ (since they can be distinguished by the arguments).
Similarly, we use $\mu$ for both $\mu_\Gamma$ and $\mu_{\Gamma^\ast}$.
\begin{lemma}\label{lemma:d:star}
Let $\Gamma^\ast$ be the 2-subdivision of $\Gamma$ on nodes $V^\ast$.
Then
$$
d([p,q],[p',q'])=d(p,p')+d(q,q')\qquad\forall [p,q],[p',q']\in V^\ast
$$
\end{lemma}
\begin{proof}
If $P=([p_0,q_0],\ldots,[p_k,q_k])$ is a path in $\Gamma^\ast$ from $[p_0,q_0]=[p,q]$ to $[p_k,q_k]=[p',q']$
then the length of $P$ in $\Gamma^\ast$ equals $\sum_{i=0}^{k-1}(d(p_i,p_{i+1})+d(q_i,q_{i+1}))\ge d(p,p')+d(q,q')$;
hence $(\ge)$ holds.

To show equality, we use induction on $D=d(p,p')+d(q,q')$. The base case $D=0$ is trivial; suppose that $D\ge 1$.
It suffices to show that one of the following holds:
\begin{itemize}[noitemsep,topsep=0pt]
\item[{\rm (i)}]  $p$ has neighbor $a$ in $\Gamma$ such that $a\sqsubseteq q$ and $d(a,p')=d(p,p')-1$.
\item[{\rm (ii)}] $q$ has neighbor $b$ in $\Gamma$ such that $p\sqsubseteq b$ and $d(b,q')=d(q,q')-1$.
%\item[{\rm (i)}]  $p'$ has neighbor $a'$ in $\Gamma$ such that $a'\sqsubseteq q'$ and $d(a,p')=d(p,p')-1$.
%\item[{\rm (ii)}] $q'$ has neighbor $b'$ in $\Gamma$ such that $p'\sqsubseteq b'$ and $d(b,q')=d(q,q')-1$.
\end{itemize}
Indeed, if {\rm (i)} holds then
$d([a,q],[p',q'])=d(a,p')+d(q,q')=D-1$ by induction,
and hence $d([p,q],[p',q'])\le (D-1)+1=D$, as required. Case {\rm (ii)} is similar.
Alternatively, it would also suffice to show symmetrical cases when $p'$ or $q'$ have appropriate neighbors.

Define $u,v,u',v'$ by~\eqref{eq:uvu'v'}. Note that $p\preceq u\preceq v\preceq q$ and $p'\preceq u'\preceq v'\preceq q'$. Modulo symmetry, two cases are possible.
\begin{itemize}
\item $p\ne u$, implying $p\prec u$.
Let $a\in[p,u]\subseteq[p,q]$ be an out-neighbor of $p$ in $\Gamma$ (i.e.\ $p\rightarrow a\preceq u$),
then  $a\sqsubseteq q$ by condition~(\ref{def:ebp}a),
and $(p,a,p')$ is a shortest subpath since $(p,a,u)$ and $(p,u,p')$ are shortest subpaths. Thus, case {\rm (i)} holds.
\item $(p,q,p',q')=(u,v,u',v')$. Then $(p,q,q',p')$ is an isometric rectangle.
By Proposition~\ref{prop:isometric-rectangle} there exist $a\in I(p,p')$ and $b\in I(q,q')$
such that $(p,q,b,a)$ and $(p',q',b,a)$ are isometric rectangles and $d(q,b)=d(p,a)=1$.
We have $a\sqsubseteq b$ by Lemma~\ref{lemma:romb}.

If $a\rightarrow p$ then  Lemma~\ref{lemma:ascending} for elements $a\preceq q$ %and the properties above 
gives $a\preceq b\rightarrow q$.
By Lemma~\ref{lemma:GammaIdeal}(a), $p\wedge b$ exists. Since $(p,q,b)$ is a shortest
subpath, we have $b\notin I(p,q)=[p,q]$, therefore
 $p\not\preceq b$ and $p\wedge b=a$.
Condition~(\ref{def:ebp}c) gives $a\sqsubseteq q$, and so case {\rm (i)} holds.

If $q\rightarrow b$ then by a symmetric argument we conclude that case {\rm (ii)} holds.

The last remaining case $p\rightarrow a,b\rightarrow q$ is impossible by  Lemma~\ref{lemma:ascending} for elements $p\preceq q$.
 \qedhere
\end{itemize}
\end{proof}

\renewcommand{\thetheoremRESTATED}{\ref{th:subdivision:modular}}
\begin{theoremRESTATED}[restated]
If $\Gamma$ is an extended modular complex then  $(V^\ast,E^\ast,w^\ast)$ is a modular complex
(i.e.\ an oriented modular graph).
\end{theoremRESTATED}

\begin{proof}
Any 4-cycle in $\Gamma^*$ is represented 
as $([p,q], [p,q'], [p',q'], [p',q])$ 
or $([p',q], [p',q'], [p,q'], [p,q])$
for some edges $p\rightarrow p',q\rightarrow q'$ in $\Gamma$, or 
$([p,x], [p,y], [p,z], [p,w])$ or $([x,p], [y,p], [z,p], [w,p])$ 
for 4-cycle $(x,y,z,w)$ and vertex $p$ in $\Gamma$.
This immediately implies that the orientation of $\Gamma^\ast$
is admissible and orientation 
and $w^*$ is orbit-invariant.

To show that $\Gamma^*$ is modular,
we are going to verify that $\Gamma^*$ 
satisfies the two conditions of Lemma~\ref{lemma:modular:characterization}.
If $[p,q]$ and $[p',q']$ are joined by an edge,
then $d_{\Gamma}(p,q)$ and $d_{\Gamma}(p',q')$ 
have different parity. This implies that $\Gamma^*$ is bipartite.

We next verify condition (2) of Lemma~\ref{lemma:modular:characterization}.
Take intervals $[p,q],[p',q']\in V^\ast$, and denote $D_p=d(p,p')$, $D_q=d(q,q')$, $D=d([p,q],[p',q'])=D_p+D_q$.
 Suppose further that we are given two distinct neighbors $[p_1,q_1],[p_2,q_2]$ of $[p,q]$ with
$d([p_1,q_1],[p',q'])=d([p_2,q_2],[p',q'])=D-1$.
Our goal is to show the existence of a common neighbor $[p^\ast,q^\ast]$ of $[p_1,q_1],[p_2,q_2]$
with $d([p^\ast,q^\ast],[p',q'])=D-2$.

Modulo symmetry, two cases are possible.
\begin{itemize}
\item $p_1=p=p_2$. Condition $d([p_i,q_i],[p',q'])=D-1$ implies that $d(q_i,q')=D_q-1$ for $i=1,2$.
By Lemma~\ref{lemma:modular:characterization}(2) for $\Gamma$, there is a common neighbor $q^\ast$ of $q_1,q_2$ with $d(q^\ast,q')=D_q-2$.
By Lemma~\ref{lemma:common-neighbor}, $p\sqsubseteq q^\ast$.
Thus, $[p,q^\ast]$ is a desired common neighbor of $[p,q_1],[p,q_2]$.
\item $p_1=p$, $q_2=q$.
For better readability let us denote $s=p_2$, $t=q_1$. To summarize, we know that
$(p,s,p')$ and $(q,t,q')$ are shortest subpaths, $d(p,s)=d(q,t)=1$,
 $s\sqsubseteq q$, $p\sqsubseteq t$, $p\sqsubseteq q$, $p'\sqsubseteq q'$.
It suffices to show $s\sqsubseteq t$; this will imply that $[s,t]$ is a desired common neighbor of $[p_1,q_1],[p_2,q_2]$.
Modulo symmetry, three subcases are possible.

\underline{Case 1}: $s\rightarrow p$, $t\rightarrow q$. Then $s\sqsubseteq q$ and $s\preceq p\preceq t\preceq q$,
so condition~(\ref{def:ebp}a) gives $s\sqsubseteq t$.

\underline{Case 2}: $p\rightarrow s$, $t\rightarrow q$. By Lemma~\ref{lemma:pq-pq-prime}, $\Pr_{[p,q]}([p',q'])$ is
equal to interval $[a,b]$ for $a=\Pr_{[p,q]}(p')$ and $b=\Pr_{[p,q]}(q')$.
Note that $s,t,a,b\in[p,q]$.
Since $(p,s,p')$ and $(s,a,p')$ are shortest subpaths, so is $(p,s,a,p')$ and thus  $s\preceq a$.
Similarly $b\preceq t$. Thus $p\preceq s \preceq a\preceq b\preceq t\preceq q$ and $p\sqsubseteq q$
imply $s\sqsubseteq t$ (by condition~(\ref{def:ebp}a)), as desired.

\underline{Case 3}: $s\rightarrow p$, $q\rightarrow t$.
Observe that $x\vee y$, $x\wedge y$ are defined and belong to $[s,t]$ for all $x,y\in[s,t]$ by Lemma~\ref{lemma:GammaIdeal}.

Consider set $\Pr_{[s,t]}([p',q'])$, which is equal to $[u,v]$ for $u=\Pr_{[s,t]}(p')$ and $v=\Pr_{[s,t]}(q')$ (Lemma~\ref{lemma:pq-pq-prime}).
We must have $p\not\preceq u$ (implying $p\wedge u=s$);
otherwise $(s,p,u)$, $(s,u,p')$ and thus $(s,p,u,p')$, $(s,p,p')$ 
would be shortest subpaths,
contradicting the assumption that $(p,s,p')$ is a shortest subpath and $p\ne s$.
Similarly, $v\vee q=t$.
Note that $u\sqsubseteq v$ by Corollary~\ref{cor:projection-preserves-ebp}.

Define $a=u\wedge q$. We claim that $a\sqsubseteq t$.
Indeed, we have $a\sqsubseteq q$ since $s\preceq a\preceq q$ and $s\sqsubseteq q$. 
It now suffices to show that there exists $b$ such that $a \sqsubseteq b$ and $b\vee q=t$;
the claim will then follow by condition~(\ref{def:ebp}b).
If $a=u$ then we can take $b=v$.
Otherwise, if $a\prec u$, take $b$ to be an out-neighbor of $a$ in $[a,u]$
(i.e.\ $a\rightarrow b\preceq u$); we have $b\notin[a,q]$ since $a=u\wedge q$, and hence $b\vee q=t$.

We have $p\sqsubseteq t$, $a\sqsubseteq t$ and $p\wedge a=p\wedge u \wedge q=s\wedge q=s$, so condition~(\ref{def:ebp}c) gives $s\sqsubseteq t$.
 \qedhere

\end{itemize}

\end{proof}

\noindent By combining Lemma~\ref{lemma:d:star}, Theorem~\ref{th:subdivision:modular} and Theorem~\ref{th:orbits}(a) we obtain:
\begin{corollary}\label{cor:mu:star}
Let $\Gamma^\ast$ be the 2-subdivision of $\Gamma$ on nodes $V^\ast$.
Then
$$
\mu([p,q],[p',q'])=\mu(p,p')+\mu(q,q')\qquad\forall [p,q],[p',q']\in V^\ast
$$
\end{corollary}

\begin{lemma}\label{lemma:restriction-submodular}
Let $f:\Gamma\rightarrow\overline{\mathbb R}$ be an $L$-convex function
on an extended modular complex $\Gamma$. Then for every $p\in\Gamma$
the restrictions of $f$ to $\calL^-_p$ and to $\calL^+_p$ are submodular functions.
\end{lemma}
\begin{proof}
Let us define $\calL_p^{\ast +}=\{[p,q]\::\:p\sqsubseteq q\}\subseteq\calL_p^\ast$. 
By Lemma~\ref{lemma:d:star},
any vertex in any shortest path between $[p,q]$ and $[p,q']$ is of the form $[p,u]$.
Hence $\calL_p^{\ast +}$ is convex in $\Gamma$ and 
in $\calL_p^{\ast}$. 
Therefore, submodularity of $f^\ast$ on $\calL_p^\ast$ implies
$f^\ast$ is submodular on $\calL_p^{\ast +}$
(by~\cite[Lemma 3.7(4)]{Hirai:0ext}).
Obviously $\calL_p^{\ast +}$ is isomorphic to $\calL_p^+$ by $[p,q] \mapsto q$.
By using relation $f(q) = f^\ast([p,q]) - f(p)$  
$(q \in \calL^+_p)$, 
we see the submodularity of $f$ on $\calL^+_p$.
The proof of submodularity of $f$ on $\calL^-_p$ is symmetric.
\end{proof}

%%%%%%%%%%%%%%%%%%%%%%%%%%%%%%%%%%%%%%%%%%%%%%%%%%%%%%%%%%%%%%%%%%%%%%%%%%%%%%%%%
%%%%%%%%%%%%%%%%%%%%%%%%%%%%%%%%%%%%%%%%%%%%%%%%%%%%%%%%%%%%%%%%%%%%%%%%%%%%%%%%%

\subsection{Proof of Proposition~\ref{prop:ebp:admissible}}

\renewcommand{\thepropositionRESTATED}{\ref{prop:ebp:admissible}}
\begin{propositionRESTATED}[restated]
Let $\Gamma$ be a modular complex. (a) Relations $\BPrel$ and $\preceq$ are admissible for $\Gamma$.
(b)~If relation $\sqsubseteq$ is admissible for $\Gamma$ then $p\, \BPrel\, q$ implies $p \sqsubseteq q$.
\end{propositionRESTATED}

To prove this proposition, we will need the following result.
A modular lattice $\calL$ is called {\em complemented} if the maximal element $1_\calL$ is a join of atoms.
(This is one possible characterization of complemented modular lattices, see~\cite[Chapter IV, Theorem 4.1]{Birkhoff:67}).
\begin{proposition}[{\cite[Proposition 6.5]{Chalopin}}]\label{prop:complemented}
Consider elements $p,q$ in a modular complex $\Gamma$ with $p\preceq q$.
Then $p\,\BPrel\, q$ if and only if $[p,q]$ is a complemented modular lattice.
\end{proposition}

We now proceed with the proof of Proposition~\ref{prop:ebp:admissible}.

\myparagraph{(a)}
Checking admissibility of $\preceq$ is straightforward. Clearly, if $B$ is a cube graph and $p,q$ are elements of $B$ with $p\preceq q$
then the subgraph of $B$ induced by $[p,q]$ is also a cube graph; this implies that $\BPrel$ satisfies
condition~(\ref{def:ebp}a). Let us show that condition~(\ref{def:ebp}b) holds (condition~(\ref{def:ebp}c) is symmetric).
Suppose that $p\,\BPrel\, a$, $p\,\BPrel\, b$ and $a\vee b$ exists.
Let $a_1,\ldots,a_k$ be the atoms of $a$ and $b_1,\ldots,b_\ell$  be the atoms of $b$,
then $a=a_1\vee\ldots\vee a_k$ and $b=b_1\vee\ldots\vee b_\ell$ by Proposition~\ref{prop:complemented}. Thus, $a\vee b=a_1\vee\ldots\vee a_k\vee b=b_1\vee\ldots\vee b_\ell$,
and so $a\vee b$ is a join of atoms of $[p,a\vee b]$. Proposition~\ref{prop:complemented} gives that $p\,\BPrel\, a\vee b$.

\myparagraph{(b)} We use induction on $d(p,q)$.
Consider elements $p,q$ with  $p\, \BPrel\, q$, $d(p,q)\ge 2$.
Using Proposition~\ref{prop:complemented}, we conclude that $q=a\vee b$ for some $a,b\in[p,q]-\{q\}$.
We have $p\,\BPrel\, a$ and $p\,\BPrel\, b$ by part (a), and so $p\sqsubseteq a$ and $p\sqsubseteq b$ by the induction hypothesis.
Condition~(\ref{def:ebp}b) gives $p\sqsubseteq q$.

%%%%%%%%%%%%%%%%%%%%%%%%%%%%%%%%%%%%%%%%%%%%%%%%%%%%%%%%%%%%%%%%%%%%%%%%%%%%%%%%%

\subsection{Properties of the Cartesian product and $L$-convexity of the metric function} %Proof of Lemma~\ref{lemma:ebp-cartesian}}
In this section we prove several properties related to the Cartesian product of extended modular complexes $\Gamma\times\Gamma'$.
We denote $(V_\Gamma,V_{\Gamma'},V_{\Gamma\times\Gamma'})=(V,V',V_\times)$ for brevity,
and use similar notation for other objects (e.g.\ relations $\preceq,\sqsubseteq,\rightarrow$, distances $d,\mu$, etc).
Recall that $(p,p') \sqsubseteq_\times (q,q')$ iff $p \sqsubseteq q$ and $p' \sqsubseteq' q'$.

\renewcommand{\thelemmaRESTATED}{\ref{lemma:ebp-cartesian}}
\begin{lemmaRESTATED}[restated]
Consider extended modular complexes $\Gamma,\Gamma'$ and element $(p,p')\in\Gamma\times\Gamma'$. 
\begin{itemize}[noitemsep,topsep=0pt]
\item[{\rm (a)}] $\Gamma\times \Gamma'$ is an extended modular complex. 
\item[{\rm (b)}] $\calL^\sigma_{(p,p')}(\Gamma\times\Gamma')=\calL^\sigma_{p}(\Gamma)\times \calL^\sigma_{p'}(\Gamma')$ for $\sigma\in\{-,+\}$. 
\end{itemize}
\end{lemmaRESTATED}

\begin{proof}
\myparagraph{(a)} In the light of Lemma~\ref{lemma:bp-cartesian}(a), it suffices to verify that relation $\sqsubseteq_{\times}$
is admissible, i.e.\ satisfies the properties in Definition~\ref{def:ebp}.
We need to show the following:
\begin{itemize}[noitemsep,topsep=0pt,leftmargin=14pt]
\item $(p,p')\sqsubseteq_\times (q,q')$ implies $(p,p')\preceq_\times (q,q')$.
\item $(p,p')\sqsubseteq_\times (p,p')$ for every $(p,p')\in V_\Gamma\times V_{\Gamma'}$.
\item $(p,p')\sqsubseteq_\times (q,q')$ for every edge $(p,p')\rightarrow_\times (q,q')$.
\item If $(p,p')\sqsubseteq_\times (q,q')$, $(a,a')\preceq_\times (b,b')$ and $(a,a'),(b,b')\in[(p,p'),(q,q')]$, then $(a,a')\sqsubseteq_\times (b,b')$.
\item If $(p,p')\sqsubseteq (q_1,q'_1)$, $(p,p')\sqsubseteq (q_2,q'_2)$ and $(q_1,q'_1)\vee (q_2,q'_2)$ exists, then $(p,p')\sqsubseteq_\times (q_1,q'_1)\vee (q_2, q'_2)$.
\end{itemize}
Checking each property is mechanical, and is omitted.

\myparagraph{(b)} 
From definitions,  $\calL^+_{(p,p')}(\Gamma\times\Gamma')=\{(q,q')\::\:p\sqsubseteq q,p'\sqsubseteq' q'\}=\calL^+_{p}(\Gamma)\times \calL^+_{p'}(\Gamma')$
as sets. Furthermore, the partial order is the same in both cases, and so
$\calL^+_{(p,p')}(\Gamma\times\Gamma')$ and $\calL^+_{p}(\Gamma)\times \calL^+_{p'}(\Gamma')$
also equal as posets (which are modular semilattices by Lemma~\ref{lemma:GammaIdeal:ebp}). 
Finally, both semilattices are assigned the same valuation, namely $v_\times(q,q')=\mu_\times((p,p'),(q,q'))=\mu(p,q)+\mu'(p',q')$.
 The case $\sigma=-$ is symmetric.
 \qedhere

\end{proof}
%%%%%%%%%%%%%%%%%%%%%%%%%%%%%%%%%%%%%%%%%%%%%%%%%%%%%%%%%%%%%%%%%%%%%%%%%%%%%%%%%

\begin{lemma}
Consider extended modular complexes $\Gamma,\Gamma'$.
Graphs $(\Gamma\times\Gamma')^\ast$ and $\Gamma^\ast\times{\Gamma'}^\ast$ 
are isomorphic with the isomorphism given by $[(p,p'),(q,q')]\mapsto ([p,q],[p',q'])$
for $p,q\in \Gamma,p',q'\in\Gamma'$ with $p\sqsubseteq q,p'\sqsubseteq' q'$.
Consequently, $\calL^\ast_{(p,p')}(\Gamma\times\Gamma')$ and $\calL^\ast_p(\Gamma)\times \calL^\ast_{p'}(\Gamma')$
are isomorphic valuated modular semilattices for any $(p,p')\in\Gamma\times\Gamma'$.
\end{lemma}
\begin{proof}
Clearly, the mapping defined in the lemma is a bijection between the nodes of 
$(\Gamma\times\Gamma')^\ast$ and the nodes of $\Gamma^\ast\times{\Gamma'}^\ast$.
Checking that this bijection preserves edges and edge weights is mechanical.
\end{proof}

\begin{lemma}\label{lemma:ftilde}
Consider extended modular complexes $\Gamma,\Gamma'$ and functions $f,f':\Gamma\rightarrow \overline{\mathbb R}$ and $\tilde f:\Gamma\times\Gamma'\rightarrow \overline{\mathbb R}$
such that  $\dom f$ is connected in $\Gamma^{\mathsmaller\sqsubset}$.
\begin{itemize}[noitemsep,topsep=0pt]
\item[{\rm (a)}] If $f$ is $L$-convex on $\Gamma$ and $\tilde f(p,p')=f(p)$ for $(p,p')\in \Gamma\times\Gamma'$ then $\tilde f$ is $L$-convex on $\Gamma\times\Gamma'$. 
\item[{\rm (b)}] If $\tilde f$ is $L$-convex on $\Gamma\times\Gamma'$ and $f(p)=\tilde f(p,p')$ for fixed $p'\in\Gamma'$  then $f$ is $L$-convex on $\Gamma$. 
\end{itemize}
\end{lemma}
\begin{proof}
Since $(\Gamma\times\Gamma')^\ast$ and $\Gamma^\ast\times {\Gamma'}^\ast$ are isomorphic,
for a function ${\mbox{$\tilde f$}}^{\,\ast}:(\Gamma\times\Gamma')^\ast\rightarrow\overline\RR$
we can define function  ${\mbox{$\tilde f$}}^{\,\ast}_\times:\Gamma^\ast\times {\Gamma'}^\ast\rightarrow\overline\RR$
via ${\mbox{$\tilde f$}}^{\,\ast}_\times([p,q],[p',q'])={\mbox{$\tilde f$}}^{\,\ast}([(p,p'),(q,q')])$.
Clearly, ${\mbox{$\tilde f$}}^{\,\ast}$ is submodular on 
$\calL^\ast_{(x,x')}(\Gamma\times\Gamma')$ 
if and only if ${\mbox{$\tilde f$}}^{\,\ast}_\times$ is submodular on 
 $\calL^\ast_x(\Gamma)\times \calL^\ast_{x'}(\Gamma')$.

\myparagraph{(a)}
Checking connectivity of $\dom \tilde f$ in $\Gamma^{\mathsmaller\sqsubset}$ is straightforward.
We can write
${\mbox{$\tilde f$}}^{\,\ast}_\times([p,q],[p',q'])=\tilde f(p,p')+\tilde f(q,q')=f(p)+f(q)= f^\ast([p,q])$. 
$L$-convexity of $f$ means that $f^\ast$ is submodular on $\calL^\ast_x$ for any $x\in\Gamma$.
Therefore, by~\cite[Lemma 3.7(2)]{Hirai:0ext}, function ${\mbox{$\tilde f$}}^{\,\ast}_\times$
is submodular on $\calL^\ast_x(\Gamma)\times\calL^\ast_{x'}(\Gamma')$ for any $(x,x')\in\Gamma\times\Gamma'$,
and thus ${\mbox{$\tilde f$}}^{\,\ast}$ is submodular on $\calL^\ast_{(x,x')}(\Gamma\times\Gamma')$.

\myparagraph{(b)}
$f^\ast([p,q])=f(p)+f(q)=\tilde f(p,p')+\tilde f(q,p')={\mbox{$\tilde f$}}^{\,\ast}([(p,p'),(q,p')])={\mbox{$\tilde f$}}^{\,\ast}_\times([p,q],[p',p'])$. 
$\mbox{$L$-convexity}$ of $\tilde f$ means that ${\mbox{$\tilde f$}}^{\,\ast}_\times$ is submodular on $\calL^\ast_x(\Gamma)\times \calL^\ast_{p'}(\Gamma')$.
Therefore, by~\cite[Lemma 3.7(3)]{Hirai:0ext}, function $f^\ast$ is submodular on $\calL^\ast_{x}(\Gamma)$
for any $x\in\Gamma$.
 \qedhere

\end{proof}

\begin{lemma}[{\cite[Lemma 4.18]{Hirai:0ext}}]\label{lemma:mu-is-Lconvex-orig}
Let $\Lambda$ be a modular complex.
The distance function $\mu_\Lambda$ is submodular on $\calL_{(a,b)}^\uparrow(\Lambda\times\Lambda)=\calL_{a}^\uparrow(\Lambda)\times \calL_{b}^\uparrow(\Lambda)$
for every $(a,b)\in\Lambda\times\Lambda$.
\end{lemma}

\begin{lemma}\label{lemma:mu-is-Lconvex}
 Let $\Gamma$ be an extended modular complex.
{\rm (a)} Function $\mu_\Gamma:\Gamma\times \Gamma\rightarrow\mathbb R_+$ is $L$-convex on $\Gamma\times \Gamma$.
{\rm (b)} For each $p\in\Gamma$, function $\mu_{\Gamma,p}:\Gamma\rightarrow\mathbb R_+$ defined via $\mu_{\Gamma,p}(x)=\mu(x,p)$ is $L$-convex on~$\Gamma$.
\end{lemma}
\begin{proof}
It suffices to prove {\rm (a)}; claim {\rm (b)} will then follow by Lemma~\ref{lemma:ftilde}(b).
To be consistent with the notation in Lemma~\ref{lemma:ftilde}, denote $\Gamma'=\Gamma$ and $\tilde f=\mu_\Gamma$.
By Corollary~\ref{cor:mu:star}, for each $[p,q],[p',q']\in V^\ast$ we have 
${\mbox{$\tilde f$}}^{\,\ast}_\times([p,q],[p',q'])={\mbox{$\tilde f$}}^{\,\ast}([(p,p'),(q,q')])=
\mu_\Gamma(p,p')+\mu_\Gamma(q,q')=\mu_{\Gamma^\ast}([p,q],[p',q'])$,
i.e.\ ${\mbox{$\tilde f$}}^{\,\ast}_\times=\mu_{\Gamma^\ast}$. 
Consider $(x,y)\in\Gamma\times\Gamma$.
Lemma~\ref{lemma:mu-is-Lconvex-orig} for $\Lambda=\Gamma^\ast$ and $a=[x,x]$, $b=[y,y]$ gives that ${\mbox{$\tilde f$}}^{\,\ast}_\times=\mu_{\Gamma^\ast}$ is submodular on 
$\calL^\uparrow_{[x,x]}(\Gamma^\ast)\times\calL^\uparrow_{[y,y]}(\Gamma^\ast)=\calL^\ast_x(\Gamma)\times \calL^\ast_y(\Gamma)$,
and therefore ${\mbox{$\tilde f$}}^{\,\ast}$ is submodular on $\calL^\ast_{(x,y)}(\Gamma\times\Gamma)$.
\end{proof}

%%%%%%%%%%%%%%%%%%%%%%%%%%%%%%%%%%%%%%%%%%%%%%%%%%%%%%%%%%%%%%%%%%%%%%%%%%%%%%%%%
%%%%%%%%%%%%%%%%%%%%%%%%%%%%%%%%%%%%%%%%%%%%%%%%%%%%%%%%%%%%%%%%%%%%%%%%%%%%%%%%%

%%%%%%%%%%%%%%%%%%%%%%%%%%%%%%%%%%%%%%%%%%%%%%%%%%%%%%%%%%%%%%%%%%%%%%%%%%%%%%%%%
%%%%%%%%%%%%%%%%%%%%%%%%%%%%%%%%%%%%%%%%%%%%%%%%%%%%%%%%%%%%%%%%%%%%%%%%%%%%%%%%%

%%%%%%%%%%%%%%%%%%%%%%%%%%%%%%%%%%%%%%%%%%%%%%%%%%%%%%%%%%%%%%%%%%%%%%%%%%%%%%%%%
%%%%%%%%%%%%%%%%%%%%%%%%%%%%%%%%%%%%%%%%%%%%%%%%%%%%%%%%%%%%%%%%%%%%%%%%%%%%%%%%%

\subsection{Proof of Theorem~\ref{th:Lopt}: Local optimality implies global optimality}

In this section we prove the following result.
\renewcommand{\thetheoremRESTATED}{\ref{th:Lopt}}
\begin{theoremRESTATED}[equivalent formulation]
Let $f:V\rightarrow\overline{\mathbb R}$ be an $L$-convex function on an extended modular complex~$\Gamma$.
Consider element $p$ with  $\min_{q\in \Gamma}f(q)<f(p)<\infty$. There exists element $u$ with $f(u)<f(p)$ and $pu\in \Gamma^{\mathsmaller\sqsubset}$.
\end{theoremRESTATED}

We use the same argument as in~\cite{Hirai:0ext} (slightly rearranged).

\begin{lemma}[{\cite[Lemma 4.20]{Hirai:0ext}}] \label{lemma:LASKFHASGASFASF}
Let $f$ be a submodular function on a modular semilattice $\calL$.
For $p,q\in\calL$ and $\alpha\in\mathbb R$ with $f(p)\le \alpha$, $f(q)<\alpha$,
there exists sequence $(u_0,u_1,\ldots,u_k)$ 
with $(u_0,u_k)=(p,q)$ such that for each $i\in[k]$ we have $f(u_i)<\alpha$
and elements $u_{i-1},u_i$ are comparable.
\end{lemma}

\begin{lemma}\label{lemma:ALISFHLAGSLASG}
Let $f$ be an $L$-convex function on an extended modular complex $\Gamma$. 
 Consider triplet $(x,y,z)$ and $\alpha\in\mathbb R$ with $x\sqsubseteq y\sqsubseteq z$, $x\not \sqsubseteq z$, $f(x)\le\alpha$, $f(y)<\infty$, $f(z)< \alpha$.
There exists sequence $(u_0,u_1,\ldots,u_k)$ with $(u_0,u_k)=(x,z)$
such that for each $i\in[k]$ we have  $f(u_i)<\alpha$ and $u_{i-1}\sqsubseteq u_i$.
\end{lemma}
\begin{proof}
We denote the desired sequence as $P(x,y,z)$, if exists.
We prove that $P(x,y,z)$ exists using induction on $d(x,z)+d(x,y)$. 
We know that function $f^\ast$ is submodular on $\calL^\ast_y$. Let us apply inequality~\eqref{eq:submodular:b}
to elements $p=[x,y]$ and $q=[y,z]$ (with $p\wedge q=[y,y]$). Since $f^\ast(p) + f^\ast(q)=f(x)+f^\ast(p\wedge q)+f(z)$, 
there must exist $[a,b]\in \calE(p,q)-\{p,q\}$
with $f(a)+f(b)=f^\ast([a,b])\le f(x)+f(z)$.
Since $[a,b]\in I(p,q)$, we must have $a\in[x,y]$, $b\in[y,z]$ by Lemma~\ref{lemma:d:star} (and also $a\sqsubseteq b$).
Condition $x\not \sqsubseteq z$ implies that $(a,b)\ne (x,z)$. Therefore, three cases are possible.
\begin{itemize}
\item $a=x$, $y\prec b\prec z$. Then we have $f(b)\le f(z)< \alpha$ and $x\sqsubseteq b\sqsubseteq z$.
Thus, we can set $P(x,y,z)=(x,b,z)$.

\item $x\prec a\prec y$, $b=z$. Then we have $f(a)\le f(x)\le\alpha$  and $x\sqsubseteq a\sqsubseteq z$.
Thus, we can set $P(x,y,z)=P(x,a,z)$, where we used the induction hypothesis
(note that $d(x,a)<d(x,y)$).

\item $x\prec a\preceq b\prec z$. Since $f(a)+f(b)\le f(x)+f(z)<2\alpha$, one of the following must hold: 

\noindent \underline{$f(a)<\alpha$.} Then we set $P(x,y,z)=(x,a,z)$ if $a\sqsubseteq z$, and $P(x,y,z)=(x,P(a,y,z))$ otherwise.

\noindent \underline{$f(b)<\alpha$.} Then we set $P(x,y,z)=(x,b,z)$ if $x\sqsubseteq b$, and $P(x,y,z)=(P(x,y,b),z)$ otherwise.

\end{itemize}
\end{proof}

\begin{corollary}\label{cor:ALSLFKSAHG}
Consider elements $x,y,z$ in an extended modular complex $\Gamma$ such that $xy,yz\in \Gamma^{\mathsmaller\sqsubset}$,
$xz\notin \Gamma^{\mathsmaller\sqsubset}$, $f(x)\le \alpha$, $f(y)<\infty$, $f(z)<\alpha$.
Then there exists path $(u_0,u_1,\ldots,u_k)$ in $\Gamma^{\mathsmaller\sqsubset}$ with $(u_0,u_k)=(x,z)$ such that $f(u_i)<\alpha$ for all $i\in [k]$.
\end{corollary}
\begin{proof}
Modulo symmetry, two cases are possible:
\begin{itemize}
\item $x\sqsupset y\sqsubset z$. The claim then follows from Lemma~\ref{lemma:LASKFHASGASFASF},
since function $f$ is submodular on $\calL_y^+$ by Lemma~\ref{lemma:restriction-submodular}.
\item $x\sqsubset y\sqsubset z$. The claim then follows from Lemma~\ref{lemma:ALISFHLAGSLASG}. \qedhere
\end{itemize}
\end{proof}

We now proceed with the proof of Theorem~\ref{th:Lopt}. 
For value $\alpha\in\mathbb R$
let $\Gamma^{\mathsmaller\sqsubset}_{\alpha}$ and $\Gamma^{\mathsmaller\sqsubset}_{<\alpha}$ be the subgraphs of $\Gamma^{\mathsmaller\sqsubset}$ induced by
nodes $p$ with $f(p)\le \alpha$ and with $f(p)< \alpha$, respectively.
\begin{lemma}\label{lemma:levelset:connectivity}
$\Gamma^{\mathsmaller\sqsubset}_{\alpha}$ is connected for any $\alpha\ge \alpha^\ast\eqdef \min_{p\in\Gamma} f(p)$.
\end{lemma}
\begin{proof}
Suppose the claim is false. For a sufficiently large $\alpha$ graph $\Gamma^{\mathsmaller\sqsubset}_{\alpha}$ is connected,
since $\dom f$ is connected in $\Gamma^{\mathsmaller\sqsubset}$.
Thus, there must exist $\alpha > \alpha^\ast$ such that $\Gamma^{\mathsmaller\sqsubset}_{\alpha}$ is connected
but $\Gamma^{\mathsmaller\sqsubset}_{<\alpha}$ is disconnected.
There must exist a pair of vertices $p,p'$ belonging to different components of $\Gamma^{\mathsmaller\sqsubset}_{<\alpha}$;
in particular, $f(p)<\alpha$ and $f(p')<\alpha$, and there exists path $P=(p_0,p_1,\ldots,p_k)$ in $\Gamma^{\mathsmaller\sqsubset}$
with $(p_0,p_k)=(p,p')$, $k\ge 2$ and $f(p_i)\le \alpha$ for all $i\in[k-1]$. Pick such $p,p',P$ so that $k\ge 2$ is minimum.
The minimality of $k$ implies that $p_{k-2}p_k\notin\Gamma^{\mathsmaller\sqsubset}$.
Apply Corollary~\ref{cor:ALSLFKSAHG} to elements $(x,y,z)=(p_{k-2},p_{k-1},p_k)$.
We obtain a path $(u_0,\ldots,u_\ell)$ in $\Gamma^{\mathsmaller\sqsubset}$ between $u_0=p_{k-2}$ and $u_\ell=p_k$
with $f(u_i)<\alpha$ for all $i\in[\ell]$. Note that $u_1,p'$ are in the same connected component of $\Gamma^{\mathsmaller\sqsubset}_{<\alpha}$
(which is different from that of $p$), and path $(p_0,\ldots,p_{k-2},u_1)$ has shorter length compared to $P$.
This contradicts the minimality of~$k$.
\end{proof}

Now consider element $p\in\Gamma$ which is not a global minimum, i.e.\ $f(p)>\alpha^\ast$.
Connectivity of $\Gamma^{\mathsmaller\sqsubset}_{\alpha}$ for $\alpha=f(p)$ implies
that there exists path $P=(p_0,p_1,\ldots,p_k)$ in $\Gamma^{\mathsmaller\sqsubset}$
with $p_0=p$, $f(p_i)\le \alpha$ for $i\in[k]$, and $f(p_k)<\alpha$.
Let us pick such $P$ so that $k$ is minimum. It suffices to show that $k=1$; clearly, this would imply Theorem~\ref{th:Lopt}.
Suppose not, i.e.\ $k\ge 2$. The minimality of $k$ implies that $p_{k-2}p_k\notin \Gamma^{\mathsmaller\sqsubset}$.
Applying Corollary~\ref{cor:ALSLFKSAHG} to elements $(x,y,z)=(p_{k-2},p_{k-1},p_k)$
gives element $u_1$ so that $f(u_1)<\alpha$ and $p_{k-2}u_1\in \Gamma^{\mathsmaller\sqsubset}$.
This contradicts the minimality of $k$.

We can strengthen Theorem~\ref{th:Lopt} as follows.
\begin{lemma}\label{lemma:Lopt:pq}
Let $f:V\rightarrow\overline{\mathbb R}$ be an $L$-convex function on an extended modular complex~$\Gamma$
and $p,q$ be distinct elements in $\dom f$.
{\rm (a)}~If $f(p)\ge f(q)$ then there exists element $u\in I(p,q)$ such that $f(u)\le f(p)$ and $pu\in \Gamma^{\mathsmaller\sqsubset}$.
Additionally, if $f(p)> f(q)$ then $f(u)<f(p)$.
{\rm (b)}~There exists element $u\in I(p,q)\cap\dom f$ such that $pu\in \Gamma^{\mathsmaller\sqsubset}$.
\end{lemma}
\begin{proof}
\myparagraph{(a)}
For element $x\in\Gamma$ define function $\mu_x:\Gamma\rightarrow\mathbb R$ via $\mu_x(u)=\mu(x,u)$.
By Lemma~\ref{lemma:mu-is-Lconvex}, 
function $\mu_x$ is $L$-convex on $\Gamma$. Now for value $C\ge 0$ define function $f_C:\Gamma\rightarrow\overline{\mathbb R}$ via $f_C(u)=f(u)+C(\mu_p(u)+\mu_q(u)-\mu(p,q))$.
It is straightforward to check that $f_C$ is $L$-convex on $\Gamma$. (Note that $\dom f_C=\dom f$, since functions $\mu_p,\mu_q$ are finite-valued).
We have the following properties: (i)~$f_C(u)=f(u)$ for all $u\in I(p,q)$;
(ii)~if $C$ is sufficiently large then $f_C(u)>f(p)$ for all $u\in V_\Gamma-I(p,q)$.
Applying Lemma~\ref{lemma:levelset:connectivity} to function $f_C$ gives the first claim,
while applying Theorem~\ref{th:Lopt} to function $f_C$ gives the second claim.

\myparagraph{(b)} We can assume w.l.o.g.\ that $f(p)>f(q)$, since adding function of the form $C\mu_q$, $C\ge 0$ to~$f$
preserves $L$-convexity of $f$ and does not affect the statement. The claim now follows from part~(a).
\end{proof}

%%%%%%%%%%%%%%%%%%%%%%%%%%%%%%%%%%%%%%%%%%%%%%%%%%%%%%%%%%%%%%%%%%%%%%%%%%%%%%%%%
%%%%%%%%%%%%%%%%%%%%%%%%%%%%%%%%%%%%%%%%%%%%%%%%%%%%%%%%%%%%%%%%%%%%%%%%%%%%%%%%%
%%%%%%%%%%%%%%%%%%%%%%%%%%%%%%%%%%%%%%%%%%%%%%%%%%%%%%%%%%%%%%%%%%%%%%%%%%%%%%%%%

\subsection{Graph thickening and operations $\UP,\DOWN,\diamond$}\label{sec:NormalPath}

Let us recall some definitions from Sections~\ref{sec:background}-\ref{sec:ebp}.
Elements $p,q$ of an extended modular complex $\Gamma$ are said to be $\DELTANEIB$-neighbors if $p,q\in[a,b]$ for some $a,b$ with $a\sqsubseteq b$.
Equivalently, $p,q$ are $\DELTANEIB$-neighbors if $a\wedge b,a\vee b$ exist and $a\wedge b\sqsubseteq a\vee b$.
Let $\DELTA{\Gamma}$ be an undirected unweighted graph on nodes $V_\Gamma$ such that $p,q$ are connected by an edge in $\DELTA{\Gamma}$ if and only if $p,q$ are $\Delta$-neighbors.
This graph is called a {\em thickening of $\Gamma$}. 
A path in $\DELTA{\Gamma}$ is called a {\em $\DELTANEIB$-path}.
The shortest path distance between $p$ and $q$ in $\DELTA{\Gamma}$ will be denoted as $\DELTA{d}(p,q)$.
These concepts were introduced in~\cite{Chalopin}
in the case when $\sqsubseteq\,=\BPrel$; we now use them for an arbitrary admissible relation $\sqsubseteq$.
In order to work with $\DELTA{\Gamma}$, in this section we introduce binary operations $\UP,\DOWN,\diamond$ on an extended modular complex $\Gamma$ and establish some of their properties.
They will be used, in particular, for proving Theorem~\ref{th:SDA}
(bound on the number of steps of $\DELTANEIB$-SDA), and for proving that the sum of $L$-convex functions is $L$-convex.

For vertices $p,q\in\Gamma$ let $\preceq_{pq}$ be a partial order on $I(p,q)$ defined as follows:
$x\preceq_{pq} y$ iff $(p,x,y,q)$ is a shortest subpath. For vertices $x,y\in I(p,q)$
this is equivalent to the condition $x\in I(p,y)$, or to the condition $y\in I(x,q)$.
Clearly, $(I(p,q),\preceq_{pq})$ is a poset with the minimal element $p$ and the maximal element $q$.
We will need the following result.
\begin{theorem}[{\cite[Theorem 6.1]{Chalopin}}]\label{th:Ipq:modular-lattice}
Let $\Gamma$ be a modular complex. For every $p,q\in\Gamma$ poset $\mbox{$(I(p,q),\preceq_{pq})$}$ is a modular lattice.
\end{theorem}
We denote the meet and join operations in this poset as $\wedge_{pq}$ and $\vee_{pq}$, respectively.
Clearly, for each $x,y\in I(p,q)$ the meet $x \wedge_{pq}y$ is a median of $p,x,y$,
and in fact it is the unique median  
(since every median $m$ of $p,x,y$ belongs to $I(p,q)$ and satisfies $m\preceq_{pq}x$ and $m\preceq_{pq}y$).
Similarly, $x \vee_{pq}y$ is the unique median of $x,y,q$.

For elements $p,q\in\Gamma$
we define $p\UP q$ to be the gate of $q$ at  $\calL^+_p$ (which is a convex subset of $\Gamma$ by Lemma~\ref{lemma:GammaIdeal:ebp}):
$p\UP q=\Pr_{\calL^+_p}(q)$.
Similarly, we let $p\DOWN q=\Pr_{\calL^-_p}(q)$.
Clearly, by the definition of the gate and by Lemma~\ref{lemma:ascending} we have 
\begin{subequations}\label{eq:arrow}
\begin{align}
I(p,q)\cap \calL^+_p&=I(p,p\UP q)=[p,p\UP q]=\{u\:|\:(p,u,p\UP q,q)\mbox{ is a shortest subpath}\} \\
I(p,q)\cap \calL^-_p&=I(p,p\DOWN q)=[p\DOWN q,p]=\{u\:|\:(p,u,p\DOWN q,q)\mbox{ is a shortest subpath}\}
\end{align}
\end{subequations}
We also define $p\diamond q = (p\UP q)\vee_{pq}(p\DOWN q)$.

\begin{lemma}\label{lemma:diamond:pq}
 Consider elements $p,q$ of an extended modular complex $\Gamma$.
There holds $p\wedge(p\diamond q)=p\DOWN q\sqsubseteq p\UP q=p\vee(p\diamond q)$
(and consequently $p,p\diamond q$ are $\DELTANEIB$-neighbors).
If $p\ne q$ then $p\ne p\diamond q$.
\end{lemma}
\begin{proof}
Denote $(a,b,c)=(p\DOWN q,p\UP q,p\diamond q)$.
We know that $a\sqsubseteq p\sqsubseteq b$
and $(p,a,c),(p,b,c)$, $(a,c,b)$ are shortest subpaths
(since $c$ is a median of $a,b,q$ and $a,b\in I(p,q)$).
From Lemma~\ref{lemma:romb} we conclude that $a\sqsubseteq c\sqsubseteq b$.
Since $p\succeq a \preceq c$ and $(p,a,c)$ is a shortest subpath, we must have $a=p\wedge c$.
Similarly, $b=p\vee c$. Condition~(\ref{def:ebp}b) gives $a\sqsubseteq b$.
Now suppose that $p\ne q$, and consider shortest $p$-$q$ path $(p,u,\ldots,q)$.
If $p\rightarrow u$ then $p\UP q\ne p$ and $p\diamond q\ne p$,
and if $p\leftarrow u$ then $p\DOWN q\ne p$ and $p\diamond q\ne p$.
\end{proof}

For elements $p,q$ an integer $k\ge 0$ define $p\diamond^k q$ as follows:
$p\diamond^0 q=p$, and $p\diamond^{k+1}q=(p\diamond^k q)\diamond q$ for $k\ge 0$.
Clearly, $p\diamond^k q=q$ for some index $k\ge 0$. 
If $k$ is the minimum such index
then the sequence $(p\diamond^0 q,p\diamond^1 q,\ldots,p\diamond^k q)$ will be called the {\em normal $p$-$q$ path}.
By the previous lemma, it is a $\DELTANEIB$-path (i.e.\ a path in $\DELTA{\Gamma}$).
%Element $p\diamond^{k-1} q$ will be denoted as $p \DIAMOND q$; if $p=q$ then we define $p\DIAMOND q=p$.

For node $p\in \Gamma$ and integer $k\ge 0$  denote $\DELTA{B}_k(p)=\{q\:|\:\DELTA{d}(p,q)\le k\}$
(a ball of radius $k$ in~$\DELTA{\Gamma}$).
We will show later that $\DELTA{B}_k(p)$ is convex in $\Gamma$ for any $p$ and $k$. First, we establish this for $k=1$.
\begin{lemma}\label{lemma:B1ball}
For each $p\in\Gamma$, set $\DELTA{B}_1(p)$ is convex in $\Gamma$.
\end{lemma}
\begin{proof}
By Lemma~\ref{lemma:convex:gated}, it suffices to show the following: if $x,x'$ are distinct $\DELTANEIB$-neighbors of $p$
and $y$ is a common neighbor of $x,x'$ in $\Gamma$ then $p,y$ are $\DELTANEIB$-neighbors.
Denote $(a,b)=(p\wedge x,p\vee x)$ and $(a',b')=(p\wedge x',p\vee x')$,
then $a\sqsubseteq b$ and $a'\sqsubseteq b'$.
Recall that sets $[a,b]$ and $[a',b']$ are convex in $\Gamma$ by Lemma~\ref{lemma:GammaIdeal}.
If $d(p,y)=d(p,x)-1$ then $y\in I(p,x)$, and so by convexity of $[a,b]$ we have $y\in[a,b]$,
implying that $p,y$ are $\DELTANEIB$-neighbors. We can thus assume that $d(p,y)=d(p,x)+1$,
and also that $d(p,y)=d(p,x')+1$ (by a symmetric argument). 
Modulo symmetry,
two cases are possible.
\begin{itemize}
\item $x\rightarrow y\rightarrow x'$.
Since $(p,a,x)$ and $(p,x,y)$ are shortest subpaths, so is $(p,a,y)$.
Since $(p,b',x')$ and $(p,x',y')$ are shortest subpaths, so is $(p,b',y)$.
$(a,p,b')$ and $(a,y,b')$ are also shortest subpaths (since $a\preceq p\prec b'$ and $a\preceq y\prec b'$).
Thus, we have $(a,b')=(p\wedge y,p\vee y)$ and $a\sqsubseteq y\sqsubseteq b'$ 
(by Lemma~\ref{lemma:romb}). Condition~(\ref{def:ebp}b) 
for pairs $a\sqsubseteq p$ and $a\sqsubseteq y$
gives $a\sqsubseteq p\vee y=b'$, and so $p,y$ are $\DELTANEIB$-neighbors.
\item $x\rightarrow y\leftarrow x'$.
By Lemma~\ref{lemma:modular:characterization}, nodes $x,x'$ have a common neighbor $z$ with $d(p,z)=d(p,x)-1=d(p,x')-1$. 
Since $z\in I(p,x)$, we must have $z\in[a,b]$ by convexity of $[a,b]$. Similarly, $z\in[a',b']$.
Since $(x,y,x',z)$ is a 4-cycle, we must have $x\leftarrow z\rightarrow x'$.
Since $a=p\wedge x$ and $z\in[a,x]$, % and $(p,z,x)$ is a shortest subpath, 
we must have $a=p\wedge z$.
By a symmetric argument, $a'=p\wedge z$, and so $a=a'$.
We know that $p,x,x',b,b'\in\calL^+_a$. By Lemma~\ref{lemma:GammaIdeal:ebp}, $\calL^+_a$ is a modular semilattice which is convex in $\Gamma$.
Since $y\in I(x,x')$, we must have $y=x\vee x'\in \calL^+_a$ by convexity.
Since joins $p\vee x=b$, $p\vee x'=b'$, $x\vee x'=y$ exist, the join $\hat b=p\vee x\vee x'$ must also exist in $\calL^+_a$
by definition of modular semilattices. We obtain $a\sqsubseteq \hat b$ and $p,y\in[a,\hat b]$, and so $p,y$ are $\DELTANEIB$-neighbors.  \qedhere
\end{itemize}
\end{proof}

\begin{lemma}\label{lemma:DeltaNeighbors:1}
If $p,x$ are $\DELTANEIB$-neighbors then $(x,p\diamond q,q)$ is a shortest subpath.
\end{lemma}
\begin{proof}
We use induction on $d(x,p\diamond q)$. Let $x'$ be a median of $x,p\diamond q,q$.
First, assume that $x'\ne x$.  By Lemma~\ref{lemma:B1ball}, $p,x'$ are $\DELTANEIB$-neighbors
(since $x'\in I(x,p\diamond q)$).
By the induction hypothesis, $(x',p\diamond q,q)$ is a shortest subpath.
$(x,x',q)$ is also a shortest subpath, and thus so is $(x,x',p\diamond q,q)$. Now assume that $x'=x$, i.e. $(p\diamond q, x, q)$ is a shortest subpath.
 $(p,p\diamond q,q)$ is a shortest subpath, and thus so is $(p,p\diamond q,x,q)$.
 Let $(a,b)=(p\wedge x,p\vee x)$. $(p,a,x)$ and $(p,b,x)$ are shortest subpaths, and thus so are $(p,a,x,q)$ and $(p,b,x,q)$,
implying $a,b\in I(p,q)$. $(a,x,b)$ is a shortest subpath since $a\preceq x\preceq b$,
and so $x$ is a median of $a,b,q$. Thus, we must have $x=a \vee_{pq} b$.
We also have $a\sqsubseteq b$, and thus $a\in[p\DOWN q,p]$ by eq.~\eqref{eq:arrow}.
 By a similar argument, $b\in[p,p\UP q]$.
We have $ a \preceq_{pq} p\DOWN q$
and $ b \preceq_{pq} p\UP q$,
and so $x=a \vee_{pq} b \preceq_{pq} (p\DOWN q) \vee_{pq} (p\UP q)=p\diamond q$.
This implies that $(x,p\diamond q,q)$ is a shortest subpath.
\end{proof}
\begin{lemma}\label{lemma:DeltaNeighbors:2}
If $(p,u,q)$ is a shortest subpath then the following sequences are also shortest subpaths:
{\rm (i)}~$(p,p\UP q,u\UP q,q)$; 
{\rm (ii)}~$(p,p\DOWN q,u\DOWN q,q)$; 
{\rm (iii)}~$(p,p\diamond q,u\diamond q,q)$.
\end{lemma}
\begin{proof}
First, let us show {\rm (i)}.
For brevity, denote  $b=p\UP q\in I(p,q)$. Define $z=b\vee_{pq} u$.
$(b,z,u)$ is a shortest subpath (since $z$ is the median of $b,u,q$).
Also, $u \preceq_{pq} z$, meaning that $(p,u,z,q)$ is a shortest subpath. 
Lemma~\ref{lemma:romb} for elements $p,b,u,z$ gives that $u\sqsubseteq z$.
This shows that $z\in\calL^+_u\cap I(u,q)$, and thus $(u,z,u\UP q,q)$ is a shortest subpath by eq.~\eqref{eq:arrow}.
Since $(z,u\UP q,q)$ and $(p,b,z,q)$ are shortest
subpaths, so is $(p,b,z,u\UP q,q)$.

A symmetric argument gives {\rm (ii)}. We can now show {\rm (iii)} as follows:
%Note that {\rm (i),(ii)} imply {\rm (iii)} since we would then have
$u\diamond q=(u\DOWN q)\vee_{uq}(u\UP q)={\tt median}(u\DOWN q,u\UP q,q)=(u\DOWN q)\vee_{pq}(u\UP q)\succeq_{pq} (p\DOWN q)\vee_{pq}(p\UP q)=p\diamond q$.
\end{proof}

%\begin{lemma}\label{lemma:DeltaNeighbors:step}

\begin{theorem}\label{th:NormalPath}
If $(p_0,p_1,\ldots,p_k)$ is a $\DELTANEIB$-path with $p_0=p$
then $(p_i,p\diamond^i q,q)$ is a shortest subpath for any $i\in[0,k]$.
%In particular, if $q=p_k$ then $p\diamond^k q=q$, and hence the normal $p$-$q$ path $(p,p\diamond^1 q,p\diamond^2 q,\ldots, q)$ is a shortest $\DELTANEIB$-path from $p$ to $q$.
\end{theorem}
\begin{proof}
We use induction on $i$. For $i=0$ the claim holds by Lemma~\ref{lemma:DeltaNeighbors:1}.
Let us show it for $i\in[k]$.
Denote $\bar p=p_{i-1}$, $x=p_i$,
$u=p\diamond^{i-1} q$.
By the induction hypothesis, $(p_{i-1},p\diamond^{i-1} q,q)=(\bar p,u,q)$ is a shortest subpath.
By Lemma~\ref{lemma:DeltaNeighbors:1}, $(x,\bar p\diamond q,q)$ is a shortest subpath.
By Lemma~\ref{lemma:DeltaNeighbors:2}, $(\bar p \diamond q,u\diamond q,q)$ is a shortest subpath.
Thus, $(x,\bar p\diamond q,u\diamond q,q)$ and so $(x,u\diamond q,q)=(p_i,p\diamond^i q,q)$ are shortest subpaths.
\end{proof}

\begin{corollary}\label{cor:NormalPath}
The normal $p$-$q$ path $(p,p\diamond^1 q,p\diamond^2 q,\ldots, q)$ is a shortest $\DELTANEIB$-path from $p$ to $q$.
\end{corollary}
\begin{proof}
Let $(p_0,p_1,\ldots,p_k)$ be a shortest $\DELTANEIB$-path with $(p_0,p_k)=(p,q)$.
 By Theorem~\ref{th:NormalPath} with $q=p_k$, $(p_k,p\diamond^k q,q)=(q,p\diamond^k q,q)$ is a shortest subpath,
 therefore $p\diamond^k q=q$. This means that the length of the normal $p$-$q$ path is at most $k$.
\end{proof}

Our next goal is to prove  that set $\DELTA{B}_k(p)$ is convex in $\Gamma$ for any $p$ and $k$.
For that we will first need a technical result.

\begin{lemma}\label{lemma:GAKSFHALGHAS}
Suppose that $(p,p',q',q)$ is an isometric rectangle 
in an extended modular complex $\Gamma$
and $p,q$ are $\DELTANEIB$-neighbors. Then $p',q'$ are also $\DELTANEIB$-neighbors.
\end{lemma}
\begin{proof}
Note that $I(p,q)\cup I(p',q')\subseteq I(p,q')\cap I(p',q)$.
For a node $x\in I(p,q)$ let $x'\in I(p',q')$ be the median of $x,p',q'$.
(This median is unique and equals $x\vee_{pq'} p'$, since $x,p'\in I(p,q')$.
Furthermore, $p'$ is a median of $p,p',q'$ and $q'$ is a median of $q,p',q'$,
so this notation is consistent).
Since $(p,x,q,q')$ and $(x,x',q')$ are shortest subpaths, so is $(p,x,x',q')$.
Analogously, since $(p',p,x,q)$ and $(p',x',x)$ are shortest subpaths, so is $(p',x',x,q)$.
We can now conclude that $(p,p',x',x)$ and $(q,q',x',x)$ are isometric rectangles.

Denote $(a,b)=(p\wedge q,p\vee q)$, then $a\sqsubseteq b$ and $a,b\in I(p,q)$.
Since $(p,p',a',a)$ is an isometric rectangle and $a\sqsubseteq p$,
we get $a'\sqsubseteq p'$ by Lemma~\ref{lemma:romb}.
By similar arguments we get $p'\sqsupseteq a'\sqsubseteq q'$ and 
$p'\sqsubseteq b'\sqsupseteq q'$.
$(p',a',q')$ and $(p',b',q')$ are shortest subpaths, and thus $(a',b')=(p'\wedge q',p'\vee q')$.
We have $a'\sqsubseteq b'$ by condition~(\ref{def:ebp}b), and so $p',q'$ are $\DELTANEIB$-neighbors.
\end{proof}

\begin{theorem}\label{th:Bkball}
For each $p\in\Gamma$ and $k\ge 0$, set $\DELTA{B}_k(p)$ is convex in $\Gamma$.
\end{theorem}
\begin{proof}

We use induction on $k$. Consider $k\ge 1$. 
By Lemma~\ref{lemma:modular:characterization}, it suffices to show that for
distinct nodes $q_1,q_2\in \DELTA{B}_{k}(p)$  and a common neighbor $q'$ of $q_1,q_2$ we have $q'\in \DELTA{B}_{k}(p)$.
Denote $p_1=p\diamond^{k-1} q_1$ and $p_2=p\diamond^{k-1} q_2$. By Theorem~\ref{th:NormalPath}, $(p_1,p_2,q_2)$ and $(p_2,p_1,q_1)$
are shortest subpaths. Let $p'$ be a median of $p_1,p_2,q'$.
Since $p'\in I(p_1,p_2)$, we have $\DELTA{d}(p,p')\le k-1$ by the induction hypothesis.
It thus suffices to show that $p',q'$ are $\DELTANEIB$-neighbors.
We can assume that $p_1\ne p_2$ (otherwise $p_1=p_2=p'$ and the claim holds by Lemma~\ref{lemma:B1ball}).
By symmetry, we can assume that $p_1\ne p'$. We know that $(p_2,p',p_1,q_1)$ and $(p_1,p',q')$ are shortest subpaths.
We also have $D\eqdef d(p_1,p') \ge 1$, $d(p',q')+D\le d(p_1,q_1)+1$ and $d(p_1,q_1)+D\le d(p',q')+1$.
This implies that $D=1$, $d(p_1,q_1)=d(p',q')$ and $(p_1,p',q',q_1)$ is an isometric rectangle.
Lemma~\ref{lemma:GAKSFHALGHAS} now gives that $p',q'$ are $\DELTANEIB$-neighbors, as desired.
\end{proof}

To conclude this section, we state one property of operations $\UP,\DOWN,\diamond$ that will be needed later.
\begin{lemma}\label{lemma:circ-preserved}
Consider elements $p,u,q$ and operation $\circ\in\{\UP,\DOWN,\diamond\}$
such that $(p,p\circ q,u,q)$ is a shortest subpath.
Then $p\circ u=p\circ q$.
\end{lemma}
\begin{proof}
To see the claim for $\circ=\!\UP$, observe that for every $x\in\calL^+_p$ sequence $(p,x,p\UP q,u,q)$ is a shortest subpath
and thus there exists a shortest $x$-$u$ path going through $p\UP q$.
This means that $p \UP q$ is the gate of $u$ at $\calL^+_p$, i.e.\ $p\UP u=p\UP q$. A symmetric argument gives the claim for~$\circ=\!\DOWN$.
Now suppose that $(p,p\diamond q,u,q)$ is a shortest subpath.
We know that $(p,p\DOWN q,p\diamond q)$ and thus $(p,p\DOWN q,p\diamond q,u,q)$ are shortest subpaths,
so the previous result gives $a\eqdef p\DOWN u=p\DOWN q$. Similarly, we have $b\eqdef p\UP u=p\UP q$.
By definition, $m=p\diamond u$ is the unique median of $a,b,u$.
It can be seen that $m$ is also a median of $a,b,q$.
(Note that $(a,u,q)$, $(a,m,u)$ and thus $(a,m,u,q)$ are shortest subpaths).
Therefore, $m=p\diamond q$.
\end{proof}

\begin{remark}\label{remark:NormalPath}
For  nodes $p,q$ with $\DELTA{d}(p,q)=k\ge 1$ let us define $p \DIAMOND q=p\diamond^{k-1} q$.
It can be deduced from Theorem~\ref{th:NormalPath} that
$p\DIAMOND q$ is the gate of $q$ at the convex set $\DELTA{B}_{k-1}(p)$ where $k=\DELTA{d}(p,q)$.
This operation was introduced in~\cite{Chalopin} under the name ``$\Delta$-gate of $p$ at $q$'' (in the case when $\sqsubseteq\,=\BPrel$).
The normal $p$-$q$ path $(p_0,p_1,\ldots,p_k)$ was also used in~\cite{Chalopin} (and was shown to be a shortest $\DELTANEIB$-path from $p$ to $q$).
However, it was defined by a different construction, namely via a recursion $p_{i-1}=p\DIAMOND p_{i}$ for $i=k,k-1,\ldots,1$. 
To our knowledge, operations $\DOWN,\UP,\diamond$ and Theorem~\ref{th:NormalPath} have not appeared in~\cite{Chalopin};
instead, the proof techniques in~\cite{Chalopin} relied on the notion of ``Helly graphs''
and on the operation $\langle\langle x,y\rangle\rangle$ (the minimal convex set containing $x,y$),
which we do not use.
\end{remark}
%%%%%%%%%%%%%%%%%%%%%%%%%%%%%%%%%%%%%%%%%%%%%%%%%%%%%%%%%%%%%%%%%%%%%%%%%%%%%%%%%%%%
%%%%%%%%%%%%%%%%%%%%%%%%%%%%%%%%%%%%%%%%%%%%%%%%%%%%%%%%%%%%%%%%%%%%%%%%%%%%%%%%%%%%
%%%%%%%%%%%%%%%%%%%%%%%%%%%%%%%%%%%%%%%%%%%%%%%%%%%%%%%%%%%%%%%%%%%%%%%%%%%%%%%%%%%%
%%%%%%%%%%%%%%%%%%%%%%%%%%%%%%%%%%%%%%%%%%%%%%%%%%%%%%%%%%%%%%%%%%%%%%%%%%%%%%%%%%%%

%%%%%%%%%%%%%%%%%%%%%%%%%%%%%%%%%%%%%%%%%%%%%%%%%%%%%%%%%%%%%%%%%%%%%%%%%%%%%%%%%%%%
%%%%%%%%%%%%%%%%%%%%%%%%%%%%%%%%%%%%%%%%%%%%%%%%%%%%%%%%%%%%%%%%%%%%%%%%%%%%%%%%%%%%
%%%%%%%%%%%%%%%%%%%%%%%%%%%%%%%%%%%%%%%%%%%%%%%%%%%%%%%%%%%%%%%%%%%%%%%%%%%%%%%%%%%%
%%%%%%%%%%%%%%%%%%%%%%%%%%%%%%%%%%%%%%%%%%%%%%%%%%%%%%%%%%%%%%%%%%%%%%%%%%%%%%%%%%%%

\subsection{Connectivity of $\dom f$}

In this section we will show that following.
\begin{theorem}\label{th:domf:arrow}
If $f:\Gamma\rightarrow\overline{\mathbb R}$ is an $L$-convex function on an extended modular complex $\Gamma$ and $p,q\in\dom f$ then {\em $p\UP q,p\DOWN q\in \dom f$}.
\end{theorem}
\begin{corollary}\label{cor:domf:arrow}
For any $p,q\in\Gamma$ there exists a $p$-$q$ path $(p_0,p_1,\ldots,p_k)$ in $\Gamma^{\mathsmaller\sqsubset}$
with $(p_0,p_k)=(p,q)$
which is a shortest subpath and has the following property:
for any $L$-convex function $f$ on $\Gamma$ with $p,q\in \dom f$
one has $p_i\in\dom f$ for all $i\in[0,k]$.
(This path is obtained by setting either $p_{i+1}=p_i\DOWN q$ or $p_{i+1}=p_i\UP q$;
one of these two elements is distinct from $p_i$ if $p_i\ne q$).
\end{corollary}

To prove Theorem~\ref{th:domf:arrow}, we start with technical observations.

\begin{lemma}\label{lemma:GALKFAGAS} Consider elements $p,q,u\in \dom f$. \\
{\rm (a)} If $p\sqsupseteq u \sqsubseteq q$ then $p\UP q\in\dom f$.
Similarly, if $p\sqsubseteq u \sqsupseteq q$ and $p,q\in \dom f$ then $p\DOWN q\in\dom f$. \\
{\rm (b)} If $p\sqsubseteq u \sqsubseteq q$ then $p\UP q\in\dom f$.
Similarly, if $p\sqsupseteq u \sqsupseteq q$ then $p\DOWN q\in\dom f$.
\end{lemma}
\begin{proof}
By symmetry, it suffices to prove only the first claim in (a) and in (b). 
In both cases we denote $a=p\UP q$.

\myparagraph{(a)} Denote $s=p\wedge q$, then $p\sqsupseteq s \sqsubseteq q$.
We know that $f$ is submodular on $\calL^+_s$.
It suffices to show that $a\in\calE(p,q)$; this would imply that $a\in \dom f$.

Point $p$ has coordinates $v_{pq}(p)=(X,0)$ where $X=\mu(s,p)$.
We have $a\in I(p,q)$, $p\sqsubseteq a$ and $a\wedge p=p$, therefore $a$ has coordinates $v_{pq}(a)=(X,\ldots)$.
Suppose that $a\notin \calE(p,q)$, then there exists $a' \in I(p,q)$ with $v_{pq}(a')=(X,\ldots)$ and $\mu(p,a')>\mu(p,a)$.
We have $a'\wedge p\in[s,p]$ and $\mu(s,a'\wedge p)=X=\mu(s,p)$,
thus $a'\wedge p=p$ or equivalently $p\preceq a'$.
For any $x\in I(p,q)$ we have $s\sqsubseteq x$ 
(this follows from condition $p\sqsupseteq s \sqsubseteq q$, Lemma~\ref{lemma:semilattice:technical}(b)
and properties~(\ref{def:ebp}a,b)). Thus, we have $s\sqsubseteq a'$ and hence $p\sqsubseteq a'$ (since $s\preceq p\preceq a'$).
Eq.~\eqref{eq:arrow} gives $a'\in[p,p\UP q]=[p,a]$,
contradicting condition $\mu(p,a')>\mu(p,a)$.

\myparagraph{(b)}
We know that $f^\ast$ is submodular on $\calL^\ast_u$. Note that $[x,y]\in I([p,u],[u,q])$ if and only if $x\sqsubseteq y$, $x\in[p,u]$, $y\in[u,q]$
(by Lemma~\ref{lemma:d:star}). For such $[x,y]$ we have $v_{[p,u], [u,q]}([x,y])=(\mu(x,u),\mu(u,y))$.
%In particular, we have $[p,a]\in I([p,u],[u,q])$.
It suffices to show that $[p,a]\in\calE([p,u],[u,q])$;
this would imply that $[p,a]\in \dom f^\ast$ and $a\in\dom f$.

From the lemma's assumption and eq.~\eqref{eq:arrow} we conclude that 
$p\sqsubseteq u\sqsubseteq a\sqsubseteq q$ and $p\sqsubseteq a$,
and hence $[p,a]\in \calL^\ast_u\cap I([p,u],[u,q])$.
Elements $[p,u]$ and $[p,a]$ have coordinates respectively
$(X,0)$ and $(X,Y)$ where $X=\mu(p,u)$ and $Y=\mu(u,a)$.
If $[p,a]\notin \calE([p,u],[u,q])$ then there exists element $[p,a']\in \calL^\ast_u\cap I([p,u],[u,q])$ whose
second coordinate is $Y'=\mu(u,a')>\mu(u,a)=Y$.
This contradicts eq.~\eqref{eq:arrow}.
\end{proof}

\begin{lemma}\label{lemma:updown}
Consider elements $p,q,u\in\Gamma$ such that $u\in[p,p\UP q]$. Define $u^+=u\UP q$ and $u^-=u\DOWN q$.
Then $p\UP q=p\UP u^+$ and $p\DOWN q=p\DOWN u^-$.
\end{lemma}
\begin{proof}
Clearly, we have $u \sqsubseteq p\UP q$, and sequence $(p,u,p\UP q,q)$
is a shortest subpath. Thus, we have $p\UP q\in \calL^+_u\cap I(u,q)$, and so $p\UP q\in[u,u\UP q]=[u,u^+]$ by eq.~\eqref{eq:arrow}.
Sequence $(p,u,p\UP q,u^+,q)$ is thus a shortest subpath, so Lemma~\ref{lemma:circ-preserved} gives the first claim: $p\UP q=p\UP u^+$.

To show the second claim, let $x=u\vee_{pq} (p\DOWN q)$ be the median of $u,p\DOWN q,q$.
We have $x\sqsubseteq u$ (by the same argument as in the proof of Lemma~\ref{lemma:diamond:pq})
and $x\in I(u,q)$, therefore $x\in[u\DOWN q,u]=[u^-,u]$  by eq.~\eqref{eq:arrow}.
Sequence $(u,x,u^-,q)$ is a thus shortest subpath.
$(p,p\DOWN q,x,q)$ is also a shortest subpath, and hence so is $(p,p\DOWN q,x,u^-,q)$.
Lemma~\ref{lemma:circ-preserved} now gives $p\DOWN q=p\DOWN u^-$.
\end{proof}

\begin{lemma}\label{lemma:GHAKDJGHALKSDFG}
If $p,q\in\dom f$ then $p\UP q,p\DOWN q\in \dom f$.
\end{lemma}
\begin{proof}
We use induction on $d(p,q)$. Suppose that $d(p,q)>0$. By Lemma~\ref{lemma:Lopt:pq},
there exists $u\in I(p,q)\cap \dom f$ such that $pu\in\Gamma^{\mathsmaller\sqsubset}$. By symmetry, we can assume that $p\sqsubset u$.
We then have $u\in[p,p\UP q]-\{p\}$. Denote $u^+=u\UP q$ and $u^-=u\DOWN q$,
then we have $u^+,u^-\in\dom f$ by the induction hypothesis.

Let  us show that $p\DOWN q\in \dom f$. We have $p\DOWN q=p\DOWN u^-$ by Lemma~\ref{lemma:updown}. If
$u^-\ne q$ then the claim holds by the induction hypothesis.
If $u^-=q$ then $u\sqsupseteq q$, and the claim follows from Lemma~\ref{lemma:GALKFAGAS}(a).

Let us show that $p\UP q\in \dom f$. We have $p\UP q=p\UP u^+$ by Lemma~\ref{lemma:updown}. If
$u^+\ne q$ then the claim  holds by the induction hypothesis.
If $u^+=q$ then $u\sqsubseteq q$, and the claim follows from Lemma~\ref{lemma:GALKFAGAS}(b).
\end{proof}

%%%%%%%%%%%%%%%%%%%%%%%%%%%%%%%%%%%%%%%%%%%%%%%%%%%%%%%%%%%%%%%%%%%%%%%%%%%%%%%%%%%%%%%%%%%%%%%%%%%%%%%%%%%%%%%%%%%%%%%%%%%%%%%%%%%%
%%%%%%%%%%%%%%%%%%%%%%%%%%%%%%%%%%%%%%%%%%%%%%%%%%%%%%%%%%%%%%%%%%%%%%%%%%%%%%%%%%%%%%%%%%%%%%%%%%%%%%%%%%%%%%%%%%%%%%%%%%%%%%%%%%%%
%%%%%%%%%%%%%%%%%%%%%%%%%%%%%%%%%%%%%%%%%%%%%%%%%%%%%%%%%%%%%%%%%%%%%%%%%%%%%%%%%%%%%%%%%%%%%%%%%%%%%%%%%%%%%%%%%%%%%%%%%%%%%%%%%%%%
%%%%%%%%%%%%%%%%%%%%%%%%%%%%%%%%%%%%%%%%%%%%%%%%%%%%%%%%%%%%%%%%%%%%%%%%%%%%%%%%%%%%%%%%%%%%%%%%%%%%%%%%%%%%%%%%%%%%%%%%%%%%%%%%%%%%
%%%%%%%%%%%%%%%%%%%%%%%%%%%%%%%%%%%%%%%%%%%%%%%%%%%%%%%%%%%%%%%%%%%%%%%%%%%%%%%%%

\subsection{Proof of Theorem~\ref{th:ebp-addition}}

\renewcommand{\thetheoremRESTATED}{\ref{th:ebp-addition}}
\begin{theoremRESTATED}[restated]
Consider extended modular complexes $\Gamma,\Gamma'$ and functions $f,f':\Gamma\rightarrow \overline{\mathbb R}$ and $\tilde f:\Gamma\times\Gamma'\rightarrow \overline{\mathbb R}$. 
\begin{itemize}[noitemsep,topsep=0pt]
\item[{\rm (a)}] If $f,f'$ are $L$-convex  on $\Gamma$ then $f+f'$ and $c\cdot f$ for $c\in\mathbb R_+$ are also $L$-convex on $\Gamma$. 
\item[{\rm (b)}] If $f$ is $L$-convex on $\Gamma$ and $\tilde f(p,p')=f(p)$ for $(p,p')\in \Gamma\times\Gamma'$ then $\tilde f$ is $L$-convex on $\Gamma\times\Gamma'$. 
\item[{\rm (c)}] If $\tilde f$ is $L$-convex on $\Gamma\times\Gamma'$ and $f(p)=\tilde f(p,p')$ for fixed $p'\in\Gamma'$
then $f$ is $L$-convex on $\Gamma$. 
\item[{\rm (d)}]  The indicator function $\delta_U:V\rightarrow \{0,\infty\}$  is $L$-convex on $\Gamma$
in the following cases: (i) $U$ is a $d_\Gamma$-convex set; (ii) $U=\{p,q\}$ for elements $p,q$ with $p\sqsubseteq q$.
\item[{\rm (e)}]  Function $\mu_\Gamma:V\times V\rightarrow\mathbb R_+$ is $L$-convex on $\Gamma\times \Gamma$.
\item[{\rm (f)}]  If $f$ is $L$-convex on $\Gamma$ then the restriction of $f$ to $\calL^\sigma_p(\Gamma)$ is submodular on $\calL^\sigma_p(\Gamma)$ for every
$p\in\Gamma$ and $\sigma\in\{-,+\}$.
\end{itemize}
\end{theoremRESTATED}

\myparagraph{(a)} We need to show that $\dom (f+f')=(\dom f)\cap (\dom f')$ is connected in $\Gamma^{\mathsmaller\sqsubset}$;
checking other conditions of $L$-convexity is straightforward.
Consider $p,q\in (\dom f)\cap (\dom f')$, and let $(p_0,p_1,\ldots,p_k)$ be the path in $\Gamma^{\mathsmaller\sqsubset}$
constructed in Corollary~\ref{cor:domf:arrow}. By the corollary, for each $i\in[0,k]$ we have $p_i\in \dom f$ and $p_i\in \dom f'$,
which implies the claim.

\myparagraph{(b)} The claim holds by Lemma~\ref{lemma:ftilde}(a).

\myparagraph{(c)} It suffices to prove that $\dom f$ is connected in $\Gamma^{\mathsmaller\sqsubset}$,
the claim will then follow from  Lemma~\ref{lemma:ftilde}(b).
Consider $p,q\in \dom f$. We know that $(p,p'),(q,p')\in \dom \tilde f$, and thus
there exists path $(\tilde p_0,\tilde p_1,\ldots,\tilde p_k)$ in $(\Gamma\times\Gamma')^{\mathsmaller\sqsubset}$
with $(\tilde p_0,\tilde p_k)=((p,p'),(q,p'))$ and $\tilde p_i\in\dom \tilde f$ for all $i$.
By Corollary~\ref{cor:domf:arrow}, such path can be chosen so that $\tilde p_i\in I((p,p'),(q,p'))$ for all $i$.
Thus, we have $\tilde p_i=(p_i,p')$ for all $i$.
Sequence $(p_0,p_1,\ldots,p_k)$ is now a path in $\Gamma^{\mathsmaller\sqsubset}$
satisfying $p_i\in\dom f$ for all $i$.

\myparagraph{(d)} Clearly, in both cases $U$ is connected in $\Gamma^{\mathsmaller\sqsubset}$;
in the case (i) this holds since $U$ is connected in $\Gamma$, and in the case (ii) the claim is trivial.

Let $U^\ast$ be the set of vertices $[p,q]$ in $\Gamma^\ast$ with $p,q\in U$, and for $x\in\Gamma$
denote $U^\ast_x=U^\ast\cap \calL^\ast_x$.
Note that $\delta_U^\ast=\delta_{U^\ast}$,
and the restriction of $\delta_{U^\ast}$ to $\calL^\ast_x$ equals $\delta_{U^\ast_x}$.
 It suffices to show that $U^\ast_x$ is convex in $\calL^\ast_x$;
the submodularity of $\delta_{U^\ast_x}$ on $\calL^\ast_x$ will then follow from~\cite[Lemma 3.7(4)]{Hirai:0ext}.

If $U=\{p,q\}$ where $p\sqsubseteq q$ then $U^\ast_x=\{[p,q]\}$ if $x\in[p,q]$, and $U^\ast_x=\varnothing$ otherwise;
in both cases $U^\ast_x$ is clearly convex.
Suppose that $U$ is a convex set in $\Gamma$. It suffices to show that $U^\ast$ is convex in $\Gamma^\ast$.
(Since $\calL^\ast_x$ is convex in $\Gamma^\ast$ by Lemma~\ref{lemma:GammaIdeal:ebp}, the intersection $U^\ast\cap\calL^\ast_x$ would then
be convex in $\Gamma^\ast$ and thus also in $\calL^\ast_x$).

We use Lemma~\ref{lemma:convex:gated} to show that $U^\ast$ is convex in $\Gamma^\ast$.
Let $[u,v]\in\Gamma^\ast$ be common neighbor of distinct $[p,q],[p',q']\in U^\ast$;
we need to show that $[u,v]\in U^\ast$. 
Modulo symmetry, 4 cases are possible:
\begin{itemize}[noitemsep,topsep=0pt]
\item $p=u=p'$ and $q\rightarrow v\rightarrow q'$.
\item $p=u=p'$ and $q\rightarrow v\leftarrow q'$.
\item $p=u=p'$ and $q\leftarrow v\rightarrow q'$.
\item $p\rightarrow u=p'$ and $q\rightarrow v=q'$.
\end{itemize} 
In each of these cases conditions $p,q,p',q'\in U$ and convexity of $U$ implies that $u,v\in U$, and so $[u,v]\in U^\ast$.

\myparagraph{(e)} The claim holds by Lemma~\ref{lemma:mu-is-Lconvex}(a).

\myparagraph{(f)} The claim holds by Lemma~\ref{lemma:restriction-submodular}.

\begin{corollary}\label{cor:NAGLKGLAKSGH}
If function $f$ is an $L$-convex on $\Gamma$ and $U$ is convex in $\Gamma$
then function $f+\delta_U$ is $L$-convex on $\Gamma$. In particular,
$f+\delta_{\DELTA{B}_k(x)}$ is $L$-convex for any $x\in\Gamma$ and $k\ge 0$.
\end{corollary}

%%%%%%%%%%%%%%%%%%%%%%%%%%%%%%%%%%%%%%%%%%%%%%%%%%%%%%%%%%%%%%%%%%%%%%%%%%%%%%%%%
%%%%%%%%%%%%%%%%%%%%%%%%%%%%%%%%%%%%%%%%%%%%%%%%%%%%%%%%%%%%%%%%%%%%%%%%%%%%%%%%%%%%
%%%%%%%%%%%%%%%%%%%%%%%%%%%%%%%%%%%%%%%%%%%%%%%%%%%%%%%%%%%%%%%%%%%%%%%%%%%%%%%%%%%%

%%%%%%%%%%%%%%%%%%%%%%%%%%%%%%%%%%%%%%%%%%%%%%%%%%%%%%%%%%%%%%%%%%%%%%%%%%%%%%%%%%%%
%%%%%%%%%%%%%%%%%%%%%%%%%%%%%%%%%%%%%%%%%%%%%%%%%%%%%%%%%%%%%%%%%%%%%%%%%%%%%%%%%%%%
%%%%%%%%%%%%%%%%%%%%%%%%%%%%%%%%%%%%%%%%%%%%%%%%%%%%%%%%%%%%%%%%%%%%%%%%%%%%%%%%%%%%
%%%%%%%%%%%%%%%%%%%%%%%%%%%%%%%%%%%%%%%%%%%%%%%%%%%%%%%%%%%%%%%%%%%%%%%%%%%%%%%%%%%%
%%%%%%%%%%%%%%%%%%%%%%%%%%%%%%%%%%%%%%%%%%%%%%%%%%%%%%%%%%%%%%%%%%%%%%%%%%%%%%%%%%%%
%%%%%%%%%%%%%%%%%%%%%%%%%%%%%%%%%%%%%%%%%%%%%%%%%%%%%%%%%%%%%%%%%%%%%%%%%%%%%%%%%%%%

\subsection{$f$-extremality}\label{sec:f-extremal}

Given an $L$-convex function $f:\Gamma\rightarrow\overline\RR$, we say that pair $(p,q)$ is {\em $f$-extremal}
if $p,q\in\dom f$ and $f(p)=\min \{f(x)\::\:x\in I(p,q)\}$.
We say that it is {\em strictly $f$-extremal} if $p,q\in\dom f$ and $\argmin \{f(x)\::\:x\in I(p,q)\}=\{p\}$.
In this section we establish some key facts about $f$-extremal pairs that
will be used in the analysis of the $\DELTANEIB$-SDA algorithm.

%\begin{lemma}
%Consider elements $p,q\in \dom f$. If $p\DOWN q\ne p$ then there exists $a\in [p\DOWN q,p]-\{p\}$ with $a\in \dom f$.
%Similarly, if $p\UP q\ne p$ then there exists $b\in [p,p\UP q]-\{p\}$ with $b\in \dom f$.
%\end{lemma}
%\begin{proof}
%We use induction on $d(p,q)$. Suppose that $p\ne q$.
%By Lemma~\ref{lemma:Lopt:pq}, there exists $a\in I(p,q)$ with $pa\in \Gamma^{\mathsmaller\sqsubset}$ and $a\in\dom f$.
%By symmetry, we can assume that $p\sqsupset a$, then $a\in [p\DOWN q,p]-\{p\}$.
%If $p\UP q=p$ then we are done. Assume that $b\eqdef p\UP q\ne p$. Consider $c=a\vee_{pq} b={\tt median}(a,b,q)$.
%Since $(p,a,c)$ and $(a,c,b)$ are shortest subpaths, we get $a\squsbseteq c$ by Lemma~\ref{lemma:romb}.
%Clearly, we have $c\in [a,a\UP q]$.
%We cannot have $a=c$ (then $(b,a,q)$, $(b,p,a,q)$, $(b,p,q)$, $(b,p,b,q)$ would be shortest subpaths implying $b=p$ - a contradiction).
%This shows that $a\UP q\ne a$. By the induction hypothesis, there exists $q'\in[a,a\UP q']-\{a\}$ with 
%\end{proof}

\begin{lemma}\label{lemma:OAUSGAISFHASGALSFHAKJSGHAKSFALGHAKJDSHAKDJFH}
Suppose that $(p,q)$ is $f$-extremal, $p\DOWN q\ne p$ and
 $s\in\argmin\{f(s)\::\:s\in[p\DOWN q,p]-\{p\}\}$. Then $(s,q)$ is $f$-extremal.
\end{lemma}
\begin{proof}
Note that $p\DOWN q\in\dom f$ by Theorem~\ref{th:domf:arrow}, and thus $s\in\dom f$.
Suppose the claim is false, then there exists $q'\in I(s,q)-\{s\}$ with $f(q')<f(s)$.
By Lemma~\ref{lemma:Lopt:pq}, there exists $u\in I(s,q')\subseteq I(p,q)$ with $us\in \Gamma^{\mathsmaller\sqsubset}$ and $f(u)<f(s)$.
First, suppose that $u\sqsubset s$. 
Note that $u\preceq s\preceq p$ and $f(u)\ge f(p)$.
We claim that there exists $x\in I(u,p)=[u,p]$ with $x\sqsubset p$ and $f(x)\le f(u)$.
Indeed, define sequence $u_0\sqsubset u_1 \sqsubset \ldots \sqsubset u_k\sqsubset p$ 
 with $u_0=u$ and $f(u_0)\ge f(u_1)\ge \ldots f(u_k)\ge f(p)$ via the following rule:
 if $u_i\sqsubset p$ then set $k=i$ and stop, otherwise let $u_{i+1}$ be an element in $[u_i,p]$
 with $u_i\sqsubset u_{i+1}$ and $f(u_i)\ge f(u_{i+1})$,
 which exists by Lemma~\ref{lemma:Lopt:pq}.
 Taking $x=u_k$ now proves the claim.

Clearly, $x\in I(p,q)$, and thus $x\in[p\DOWN q,p]$ and $f(x)\le f(u)<f(s)$, which contradicts the choice of $s$.
Thus, we must have $u\sqsupset s$.
Note that $p\wedge u=s$ since $p\succeq s\preceq u$ and $(p,s,u)$ is a shortest subpath.
We know that $f$ is submodular on a modular semilattice $\calL_s^+$, and $p,u\in\calL_s^+$.
Thus, $f(p)+f(u)\ge f(s)+\sum_{a\in\calE(p,u)}\lambda_a f(a)$
where $\lambda_a$ are nonnegative numbers with $\sum_{a\in\calE(p,u)}\lambda_a =1$.
This gives a contradiction since $f(u)<f(s)$  and $f(p)\le f(a)$ for all $a\in\calE(p,u)\subseteq I(p,q)$.
\end{proof}
\begin{corollary}\label{cor:GLKAJS}
If $(p,q)$ is $f$-extremal and $p\wedge q$ exists then $(p\wedge q,q)$ is $f$-extremal.
\end{corollary}
\begin{proof}
The claim follows by repeatedly applying Lemma~\ref{lemma:OAUSGAISFHASGALSFHAKJSGHAKSFALGHAKJDSHAKDJFH} 
(with an induction argument).
\end{proof}

\begin{lemma}\label{lemma:f-extremal}
If $(p,q)$  is $f$-extremal then  $f(p\diamond q)\le f(q)$.
\end{lemma}
\begin{proof}
We use induction on $d(p,q)$. We can assume that $p\diamond q\ne q$, otherwise the claim is trivial. 
It would suffice to show the following:
\begin{itemize}[noitemsep,topsep=0pt]
\item[{($\star$)}] {\em There exists $y\ne q$ such that $f(y)\le f(q)$ and $(p\diamond q,y,q)$ is a shortest subpath.}
\end{itemize}
Indeed, by Lemma~\ref{lemma:circ-preserved} we would then have $p\diamond q=p\diamond y$,
and so the induction hypothesis for $(p,y)$ would give $f(p\diamond q)=f(p\diamond y)\le f(y)\le f(q)$.

One of $p\DOWN q,p\UP q$ must be different from $p$. By symmetry, we can assume that $p\DOWN q\ne p$.
We can further assume the following: if $p\DOWN q\ne p$ and $p\UP q\ne p$
then condition $p\sqsupseteq p\wedge q \sqsubseteq q$ does {\bf not} hold,
where the existence of $p\wedge q$ is a part of the condition.
Indeed, conditions $p\sqsupseteq p\wedge q \sqsubseteq q$ and $p\sqsubseteq p\vee q \sqsupseteq q$
cannot hold simultaneuosly (otherwise we would have $p\diamond q= q$),
so we can pick the condition that does not hold and then change the orientation if necessary.

Denote $x=p\DOWN q$, and pick $s\in\argmin\{f(s)\::\:s\in[x,p]-\{p\}\}$.
By Lemma~\ref{lemma:OAUSGAISFHASGALSFHAKJSGHAKSFALGHAKJDSHAKDJFH}, $(s,q)$ is $f$-extremal. 
Four cases are possible.

\vspace{5pt}\noindent
\fbox{Case 1: $\DELTA{d}(s,q)\ge 2$}
Denote $y=s\diamond q$, then $y\ne q$. By Lemma~\ref{lemma:DeltaNeighbors:2}, $(p\diamond q,y,q)$ is a shortest subpath.
%By Lemma~\ref{lemma:circ-preserved} we have $p\diamond u=p\diamond q$.
By the induction hypothesis for $(s,q)$ we have $f(y)\le f(q)$.
Thus, $(\star)$ holds.
%By the induction hypothesis for $(p,u)$ we have $f(p \diamond q)=f(p \diamond u)\le f(u)\le f(q)$.

\medskip

In the remaining cases we have $\DELTA{d}(s,q)\le 1$. We will denote $a=s\wedge q=p\wedge q$.
Note that $p\sqsupseteq s\sqsupseteq a\sqsubseteq q$.
Furthermore, $x\in[a,s]$. 
Indeed, $(p,s,x,q)$ and thus $(s,x,q)$ are shortest subpaths
and $s\sqsupseteq x$, so $x\succeq s\DOWN q=s\wedge q=a$ by eq.~\eqref{eq:arrow}.
One can check that $(p,s,x,a,q)$ is a shortest subpath
(since $(p,s,q)$, $(s,a,q)$, $(s,x,a)$ are shortest subpaths).
Furthermore, $(p\wedge q,q)=(a,q)$ is $f$-extremal by Corollary~\ref{cor:GLKAJS}.

\vspace{5pt}\noindent
\fbox{Case 2: $\DELTA{d}(s,q)\le 1$, $p\sqsupseteq a$}
We have $p\sqsupseteq a\sqsubseteq q$, which implies that $p\UP q=p$ and thus $p\diamond q=p\DOWN q=x$.
$f$-extremality of $(a,q)$ gives $f(a)\le f(q)$. We have $a\ne q$ since we assumed that $p\diamond q\ne q$.
Thus, element $y=a$ satisfies condition ($\star$).

\vspace{5pt}\noindent
\fbox{Case 3: $\DELTA{d}(s,q)\le 1$, $p\not\sqsupseteq a\prec q$}
We have $p\diamond a=p\DOWN a=p\DOWN q=x$ where the first two equalities hold since $p\succeq a$ and by Lemma~\ref{lemma:circ-preserved}, respectively.
Since $(a,q)$ is $f$-extremal, we get  $a\in \dom f$.
Since $a\ne q$, we can apply the induction hypothesis for $(p,a)$  and get $f(x)=f(p\diamond a)\le f(a)$.
 %(by observation ($\star$) in the proof of Lemma~\ref{lemma:updown}).
 Conditions $\DELTA{d}(s,q)\le 1$ and  $x\in[a,s]$
imply that $y\eqdef x\vee q$ exists  and $a\sqsubseteq x$.
 Note that $x\wedge q=a$.
By submodularity inequality for $x,q$ in modular semilattice $\calL_a^+$ we get 
$f(x)+f(q)\ge f(a)+f(y)$, and so $f(y)\le f(q)$.
Clearly, $(p,x,y,q)$ is a shortest subpath.
%We can thus apply induction hypothesis to $(p,y)$ and get $f(
By applyng Lemma~\ref{lemma:semilattice:technical} to modular semilattice $\calL_a^\uparrow$ 
we obtain that $(p,p\UP q,y,q)$ is a shortest subpath
(since $y\in I(p,q)$ and $y\succeq q$).
Since $y\succeq_{pq} x=p\DOWN q$ and $y\succeq_{pq} p\UP q$,
we get $y \succeq_{pq} (p\DOWN q)\vee_{pq} (p\UP q)=p\diamond q$,
i.e.\ $(p,p\diamond q,y,q)$ is a shortest subpath.
We have $x\ne a$ (since $p\sqsupseteq x$ and $p\not\sqsupseteq a$)
and thus $y\ne q$. 
Therefore, $(\star)$ holds.

\vspace{5pt}\noindent
\fbox{Case 4: $\DELTA{d}(s,q)\le 1$, $p\not\sqsupseteq a=q$}
We have $q\preceq x\preceq s\preceq p$, $q\sqsubseteq s$ (since $\DELTA{d}(s,q)\le 1$)
and $x\sqsubseteq p$.
Note that $p\diamond q=p\DOWN q=x$. Also, $q\ne x$ (since $q\not\sqsubseteq p$ and $x\sqsubseteq p$).
By Theorem~\ref{th:domf:arrow} we have $x\in \dom f$.
We claim that $(\star)$ holds. Indeed, suppose this is false, then
 $f(y)>f(q)$ for all $y\in[q,x]-\{q\}$.
Denote $\alpha=[q,s]$, $\beta=[x,s]$, $\gamma=[x,p]$,
then $\alpha,\beta,\gamma\in \calL^\ast_s$ and $\alpha\wedge \gamma=\alpha\cap \gamma=\beta$.
We know that function $f^\ast:\calL^\ast_s\rightarrow\overline\RR$ is submodular on $\calL^\ast_s$.
Condition $f(q)<f(x)$ implies that $f^\ast(\alpha)<f^\ast(\beta)$.
Condition $f(p)\le f(s)$ implies that $f^\ast(\gamma)\le f^\ast(\beta)$.

Consider $\delta=[u,v]\in \calE(\alpha,\gamma)-\{\gamma\}$.
By Lemma~\ref{lemma:d:star} we must have $u\in[q,x]$ and $v\in[s,p]$.
We must have $v\ne p$
(otherwise condition $\delta\ne\gamma$ would imply that $u\in[q,x]-\{x\}$;
we would also have $u\sqsubseteq v=p$ which is impossible since $x=p\DOWN q$).
Conditions $u\in[q,x]$ and $v\in[s,p]-\{p\}$ have  two implications: (i) the second coordinate of $v_{\alpha\gamma}(\delta)$ is strictly smaller 
than the second coordinate of $v_{\alpha\gamma}(\gamma)$;
(ii) $f(u)\ge f(q)$, $f(v)\ge f(s)$ and hence $f^\ast(\delta)\ge f^\ast(\alpha)$.

By submodularity, $f^\ast(\alpha)+f^\ast(\gamma)\ge f^\ast(\beta)+\sum_{\delta\in\calE(\alpha,\gamma)}\lambda_\delta f^\ast(\delta)$
where $\lambda_\delta$ are nonnegative numbers with $\sum_{\delta\in\calE(\alpha,\gamma)}\lambda_\delta =1$.
By implication (i), we have $\lambda_\gamma>0$. 
The submodularity inequality can be rewritten as
$$
[f^\ast(\gamma)-f^\ast(\beta)]+\sum_{\delta\in\calE(\alpha,\gamma):\lambda_\delta>0}\lambda_\delta [f^\ast(\alpha)-f^\ast(\delta)]\ge 0
$$
All numbers in square brackets are non-positive, thus they must all be 0. This implies that $f^\ast(\gamma)=f^\ast(\beta)$
and $f^\ast(\alpha)=f^\ast(\gamma)$, which contradicts condition $f^\ast(\alpha)<f^\ast(\beta)$.

%Indeed, we have $f^\ast(\alpha)+f^\ast(\gamma)\ge f^\ast(\beta)+1\cdot f^\ast(\alpha)$
%
% (by considering cases $f^\ast(\alpha)>f^\ast(\gamma)$ and $f^\ast(\alpha)\le f^\ast(\gamma)$).
%Assume for simplicity that $\min\{f^\ast(\alpha),f^\ast(\gamma)\}=0$, and denote $A=\max\{f^\ast(\alpha),f^\ast(\gamma)\}$.

\end{proof}

We conclude this section with a couple of results that will help to deal with cases when values of $f$ are not unique.
\begin{lemma}\label{lemma:f-extremal:1}
If $(p,q)$ is strictly $f$-extremal and $p\diamond q\ne q$ then $f(p\diamond q)<f(q)$.
\end{lemma}
\begin{proof}
For $\varepsilon>0$ define function $f_\varepsilon:\Gamma\rightarrow\overline\RR$ via
$f_\varepsilon(x)=f(x)+\varepsilon\mu(x,q)$. Note that $f_\varepsilon$ is $L$-convex on $\Gamma$.
Clearly, there exists $\varepsilon>0$ such that $(p,q)$ is $f_\varepsilon$-extremal.
Using Lemma~\ref{lemma:f-extremal}, we obtain 
$f(p\diamond q)=f_\varepsilon(p\diamond q)-\varepsilon\mu(p\diamond q,q)<f_\varepsilon(p\diamond q)\le f_\varepsilon(q)=f(q)$.
\end{proof}

\begin{lemma}\label{lemma:f-extremal:2}
Suppose that $(p,q)$ is $f$-extremal, $p\preceq q$, $p\not\sqsubseteq q$, $q\diamond p\preceq p\diamond q$ and
$f(q')\ge f(q)$ for all $q'\in[p\diamond q,q]$. Then there exists $p'\in[p,q\diamond p]-\{p\}$ with $f(p')=f(p)$.
\end{lemma}
\begin{proof}
Denote $(a,b)=(q\diamond p,p\diamond q)$, then $p\preceq a\preceq b\preceq q$, $p\sqsubseteq b$, $a\sqsubseteq q$ and hence $a\sqsubseteq b$.
By Theorem~\ref{th:domf:arrow}, we have $a,b\in\dom f$.
%We can assume w.l.o.g.\ that $f(p)=0$ (by adding a constant to $f$, if necessary).
Denote $\alpha=[p,b]$, $\beta=[a,b]$, $\gamma=[a,q]$, then $\alpha,\beta,\gamma\in\calL^\ast_a$ and $\alpha\wedge \gamma=\beta$.
Note that $f^\ast(\alpha)+f^\ast(\gamma)=f^\ast(\beta)+f(p)+f(q)$.
By submodularity, $f^\ast(\alpha)+f^\ast(\gamma)\ge f^\ast(\beta)+\sum_{\delta\in\calE(\alpha,\gamma)}\lambda_\delta f^\ast(\delta)$,
or equivalently $f(p)+f(q)\ge \sum_{\delta\in\calE(\alpha,\gamma)}\lambda_\delta f^\ast(\delta)$
where $\lambda_\delta$ are nonnegative numbers with $\sum_{\delta\in\calE(\alpha,\gamma)}\lambda_\delta =1$.
For every $\delta=[u,v]\in\calE(\alpha,\gamma)$ we have $u\in[p,a]$, $v\in[b,q]$ by Lemma~\ref{lemma:d:star}.
Furthermore, $f(u)\ge f(p)$ and $f(v)\ge f(q)$ by the lemma's assumption, and so $f^\ast(\delta)\ge f(p)+f(q)$.
This implies that $f(u)=f(p)$, $f(v)=f(q)$ for every $\delta=[u,v]\in\calE(\alpha,\gamma)$ with $\lambda_\delta>0$.
Note that $[p,q]\notin\calL^\ast_a$ since $p\not\sqsubseteq q$, and so $\alpha\vee\gamma$ does not exist in $\calL^\ast_a$.
This implies that $\alpha\ne \beta\ne\gamma$.
Therefore, there exists $\delta=[u,v]\in\calE(\alpha,\gamma)-\{\alpha\}$ with $\lambda_\delta>0$ (and hence with $f(u)=f(p)$).
We claim that $u\ne p$ (and so $u\in[p,a]-\{p\}$, implying the lemma).
Indeed, if $u=p$ then condition $\delta\ne\alpha$ gives that $v\ne b$, and hence $v\in[b,q]-\{b\}$.
This is impossible since $p=u\sqsubseteq v$ and $b=p\diamond q=p\UP q$.
\end{proof}

%%%%%%%%%%%%%%%%%%%%%%%%%%%%%%%%%%%%%%%%%%%%%%%%%%%%%%%%%%%%%%%%%%%%%%%%%%%%%%%%%%%%
%%%%%%%%%%%%%%%%%%%%%%%%%%%%%%%%%%%%%%%%%%%%%%%%%%%%%%%%%%%%%%%%%%%%%%%%%%%%%%%%%%%%
%%%%%%%%%%%%%%%%%%%%%%%%%%%%%%%%%%%%%%%%%%%%%%%%%%%%%%%%%%%%%%%%%%%%%%%%%%%%%%%%%%%%
%%%%%%%%%%%%%%%%%%%%%%%%%%%%%%%%%%%%%%%%%%%%%%%%%%%%%%%%%%%%%%%%%%%%%%%%%%%%%%%%%%%%
%%%%%%%%%%%%%%%%%%%%%%%%%%%%%%%%%%%%%%%%%%%%%%%%%%%%%%%%%%%%%%%%%%%%%%%%%%%%%%%%%%%%
%%%%%%%%%%%%%%%%%%%%%%%%%%%%%%%%%%%%%%%%%%%%%%%%%%%%%%%%%%%%%%%%%%%%%%%%%%%%%%%%%%%%

\subsection{Proof of Theorem~\ref{th:SDA}: analysis of the $\DELTANEIB$-SDA algorithm}

\renewcommand{\thetheoremRESTATED}{\ref{th:SDA}}
\begin{theoremRESTATED}[restated]
Let $\Gamma$ be an extended modular complex and 
 $f:\Gamma^n\rightarrow\overline\RR$ be an $L$-convex function on $\Gamma^n$.
$\DELTANEIB$-SDA algorithm applied to function $f$ terminates after generating exactly
 $1+\max\limits_{i\in [n]}\DELTA{d}_\Gamma(x_i,{\tt opt}_i(f))$ distinct points,
where~$x$ is the initial vertex and ${\tt opt}_i(f)$ is as defined in Theorem~\ref{th:SDA:orig}.
\end{theoremRESTATED}

First, we show the following fact.
\begin{lemma}\label{lemma:max-d}
Suppose that $\Gamma$ is a Cartesian product of extended modular complexes: $\Gamma=\Gamma_1\times\ldots\times\Gamma_n$.
Then $\DELTA{d}_\Gamma(x,y)=\max_{i\in[n]} \DELTA{d}_{\Gamma_i}(x_i,y_i)$.
\end{lemma}
\begin{proof}
It suffices to show the following fact for extended modular complexes $\Gamma,\Gamma'$.
\begin{itemize}[noitemsep,topsep=0pt]
\item[{($\ast$)}] {\em The following conditions are equivalent for elements $x,y\in\Gamma$ and $x',y'\in\Gamma'$: \\
{\rm (a)} $(x,x'),(y,y')$ are $\DELTANEIB$-neighbors; \\
{\rm (b)} $x,y$ are $\DELTANEIB$-neighbors and $x',y'$ are $\DELTANEIB$-neighbors.
 }
\end{itemize}
Let us define $(a,b)=(x\wedge y,x\vee y)$ and $(a',b')=(x'\wedge y',x'\vee y')$.
(These expressions are defined if either {\rm (a)} or {\rm (b)} holds).
By the definition of the Cartesian product $\Gamma\times\Gamma'$, we have the following implications:
\begin{align*}
{\rm (a)}
	&\quad\Leftrightarrow\quad
(x,x')\sqsupseteq(a,a')\sqsubseteq(y,y') \mbox{ and }(x,x')\sqsubseteq(b,b')\sqsupseteq(y,y') \\
	&\quad\Leftrightarrow\quad
x\sqsupseteq a\sqsubseteq y \mbox{ and }x'\sqsupseteq a'\sqsubseteq y' \mbox{ and }x\sqsubseteq b\sqsupseteq y \mbox{ and }x'\sqsubseteq b'\sqsupseteq y' 
	\quad\Leftrightarrow\quad
{\rm (b)} \qedhere
\end{align*} 
\end{proof}
In the light of this lemma, it suffices to prove Theorem~\ref{th:SDA}
in the case when $n=1$. Accordingly, from on we assume that $\DELTANEIB$-SDA  is applied to minimize function $f:\Gamma\rightarrow\overline\RR$
which is $L$-convex on an extended modular complex $\Gamma$.

Let $\bar x$ be the initial vertex of the $\DELTANEIB$-SDA algorithm, and for $k\ge 0$ define function
$f_k=f+\delta_{\DELTA{B}_k(\bar x)}$. By Corollary~\ref{cor:NAGLKGLAKSGH}, $f_k$ is $L$-convex on $\Gamma$ for any $k\ge 0$.
\begin{lemma}\label{lemma:NGAGASGASGA}
For any $q\in\argmin f_{k-1}$ there exists $p\in\argmin f_{k}$ with $\DELTA{d}(p,q)\le 1$.
\end{lemma}
\begin{proof}
Let $p$ be an element of $\argmin f_{k}$ which is closest to $q$.
Clearly, $(p,q)$ is $f_k$-extremal, and in fact
 strictly $f_{k}$-extremal.
 (If there exists $p'\in I(p,q)-\{p\}$ with $f_k(p')=f_k(p)$
 then $p'$ also belongs to $\argmin f_{k}$
 and is closer to $q$ than $p$, contradicting the choice of $p$).
Since $\DELTA{d}(\bar x,p)\le k$, there exists $u\in \DELTA{B}_{k-1}(\bar x)$ with $\DELTA{d}(u,p)\le 1$.
 $(u,p\diamond q,q)$ is a shortest subpath by Theorem~\ref{th:NormalPath}.
 We have $u,q\in \DELTA{B}_{k-1}(\bar x)$ and thus $p\diamond q\in \DELTA{B}_{k-1}(\bar x)$ by convexity of $\DELTA{B}_{k-1}(\bar x)$ (Theorem~\ref{th:Bkball}).
If $p\diamond q\ne q$ then $f_k(p\diamond q)<f_k(q)$ by Lemma~\ref{lemma:f-extremal:1}, which contradicts the choice of $q$.
Thus, we must have $p\diamond q=q$ and so $\DELTA{d}(p,q)\le 1$.

\end{proof}

\begin{lemma}\label{lemma:GALJSGAKSG}
Consider $x\in\dom f$, and let $p$ be an element of $\argmin \{f(y)\:|\:y\in \DELTA{B}_1(x)\}$
which is closest to $x$.
Let $x^-\in\argmin \,\{f(y)\:|\:y\in\calL^-_x(\Gamma)\} $
and $x^+\in\argmin \,\{f(y)\:|\:y\in\calL^+_x(\Gamma)\} $.
Then $x^-\sqsubseteq p \sqsubseteq x^+$.
\end{lemma}
\begin{proof}
By symmetry, it suffices to prove that $p\sqsubseteq q$ where we denoted $q=x^+$.
Define $(a,b)=(p\wedge x,p\vee x)$. Note that $[a,b]$ is a modular lattice.
Since $q\succeq x$, the meet $s=p\wedge q$ exists and satisfies $s\in[a,p]$.
Define $y=s\vee x$, then $y\in[s,q]$ and $p\vee y=b$.
Since $p\sqsupseteq a\sqsubseteq x$, we get $p\sqsupseteq s\sqsubseteq y$
(using, in particular, Lemma~\ref{lemma:romb}).
This means that $p\DOWN q=s$, $p\UP q\succeq b$ and $p\diamond q\in[y,q]\subseteq[x,q]$.

Note that $p,q\in\DELTA{B}_1(x)$ and thus $I(p,q)\subseteq \DELTA{B}_1(x)$ by convexity of $\DELTA{B}_1(x)$.
The choice of $p$ thus implies  that $(p,q)$ is $f$-extremal.
We claim that $f(p)=f(s)$.
Indeed, suppose that $f(p)<f(s)$. Pair $(s,q)$ is $f$-extremal by Corollary~\ref{cor:GLKAJS},
and thus $f(p)<f(s)\le f(q)$.
Let $p'$ be an element of $I(p,q)$ with $f(p')=f(p)$ which is closest element to $q$.
By construction, $p'\ne q$ and $(p',q)$ is strictly $f$-extremal. If $p'\diamond q\ne q$
then $f(p'\diamond q)<f(q)$ by Lemma~\ref{lemma:f-extremal:1},
and $(p,p\diamond q,p'\diamond q,q)$ is a shortest subpath by Lemma~\ref{lemma:DeltaNeighbors:2}.
Condition $p\diamond q\in[x,q]$ thus implies that $x\preceq p\diamond q\preceq p'\diamond q\preceq q$.
Condition $x\sqsubseteq q$ now gives $x\sqsubseteq p'\diamond q$, contradicting the choice of $q$.
Thus, we must have $p'\diamond q=q$.
By the choice of $p'$ we have $f(p')<f(p'\wedge q)$, so the submodularity inequality
for $(p',q)$ gives that $f(p'\vee q)<f(q)$
where  $p'\vee q$ exists since $p'\diamond q=q$.
 Note that $x\preceq q\preceq p'\vee q$ 
and $p'\vee q\in I(p,q)\subseteq\DELTA{B}_1(x)$. This implies that $x\sqsubseteq p'\vee q$,
contradicting the choice of $q$. We proved that $f(p)=f(s)$,
and thus $p=s$ by the choice of $p$, since $(p,s,x)$ is a shortest subpath.

We now know that $p\sqsubseteq y\sqsubseteq q$ where $y=b=x\vee p$.
Suppose that $p\not \sqsubseteq q$.
Clearly, we have $p\diamond q\in[y,q]$, $q\diamond p\in[p,y]$ and $q\diamond p\preceq y\preceq p\diamond q$.
By choice of $q$ we must have $f(q')\ge f(q)$ for any $q'\in[p\diamond q,q]\subseteq [y,q]\subseteq[x,q]$.
By Lemma~\ref{lemma:f-extremal:2},
there exists $p'\in[p,q\diamond p]-\{p\}\subseteq [p,y]-\{p\}$ with $f(p')=f(p)$.
This contradicts the choice of $p$, since $(p,p',y,x)$ is a shortest subpath.
\end{proof}

\begin{corollary}\label{cor:NLAKJDNA}
Element $\DELTA{x}$ computed in line 4 of Algorithm~\ref{alg:zigzagSDA} satisfies $\DELTA{x}\in\argmin \{f(y)\:|\:y\in \DELTA{B}_1(x)\}$.
\end{corollary}
\begin{proof}
Clearly, any $y\in\calL^+_{x^-}(\Gamma)\cap\calL^-_{x^+}(\Gamma)$ satisfies
$y\in \DELTA{B}_1(x)$. The claim now follows directly from Lemma~\ref{lemma:GALJSGAKSG}.
\end{proof}

Lemma~\ref{lemma:NGAGASGASGA} and Corollary~\ref{cor:NLAKJDNA} imply that after $k$ iterations point $x$ in $\DELTANEIB$-SDA
satisfies $x\in\argmin f_k$ (via an induction argument).
This yields Theorem~\ref{th:SDA}.

%%%%%%%%%%%%%%%%%%%%%%%%%%%%%%%%%%%%%%%%%%%%%%%%%%%%%%%%%%%%%%%%%%%%%%%%%%%%%%%%%%%%
%%%%%%%%%%%%%%%%%%%%%%%%%%%%%%%%%%%%%%%%%%%%%%%%%%%%%%%%%%%%%%%%%%%%%%%%%%%%%%%%%%%%
%%%%%%%%%%%%%%%%%%%%%%%%%%%%%%%%%%%%%%%%%%%%%%%%%%%%%%%%%%%%%%%%%%%%%%%%%%%%%%%%%%%%
%%%%%%%%%%%%%%%%%%%%%%%%%%%%%%%%%%%%%%%%%%%%%%%%%%%%%%%%%%%%%%%%%%%%%%%%%%%%%%%%%%%%
%%%%%%%%%%%%%%%%%%%%%%%%%%%%%%%%%%%%%%%%%%%%%%%%%%%%%%%%%%%%%%%%%%%%%%%%%%%%%%%%%%%%
%%%%%%%%%%%%%%%%%%%%%%%%%%%%%%%%%%%%%%%%%%%%%%%%%%%%%%%%%%%%%%%%%%%%%%%%%%%%%%%%%%%%

\section{VCSP proofs}\label{sec:VCSP}
In this section we prove two results: Theorem~\ref{th:BLPsolvesGamma} (BLP relaxation solves language $\Phi_\Gamma$ for an extended modular complex~$\Gamma$)
and the hardness direction of Theorem~\ref{th:main}.
Below we give some background on Valued Constraint Satisfaction Problems (VCSPs) which will be needed for these proofs.

Let us fix finite set $D$. Let $\calO^{(m)}$ be the set of operations $g:D^m\rightarrow D$.
Operation $g\in \calO^{(m)}$ is called {\em symmetric} if $g(x_1,\ldots,x_m)=g(x_{\pi(1)},\ldots,x_{\pi(m)})$ for any tuple $(x_1,\ldots,x_m)\in D^m$ and any permutation $\pi:[m]\rightarrow[m]$.
A {\em fractional operation of arity $m$} is a probability distribution over $\calO^{(m)}$, i.e.\ vector $\omega\in[0,1]^{\calO^{(m)}}$ with $\sum_g \omega(g)=1$.
Fractional operation $\omega$ is called {\em symmetric} if all operations in $\supp(\omega)=\{g\in\calO^{(m)}\:|\:\omega(g)>0\}$ are symmetric.

A cost function $f:D^n\rightarrow\overline\RR$ is said to {\em admit $\omega$} (or $\omega$ is a {\em fractional polymorphism of $f$}) if
$$
\sum_{g\in\supp(\omega)} \omega(g) f(g(x^1,\ldots,x^m))\le \frac{1}{m} \sum_{i=1}^m f(x^i)\qquad\forall x^1,\ldots,x^m\in D^n
$$
where operation $g(\cdot)$ is applied componentwise, i.e.\ 
$$g(x^1,\ldots,x^m)=(g(x^1_1,\ldots,x^m_1),\ldots,g(x^1_n,\ldots,x^m_n))\in D^n$$
Language $\Phi$ over $D$ is said to admit $\omega$ if all functions $f\in\Phi$ admit $\omega$.
In this case $\omega$ is called a {\em fractional polymorphism of $\Phi$}.

The {\em expressive power $\langle\Phi\rangle$ of language $\Phi$} is defined as the set of all cost functions $f:D^n\rightarrow\overline\RR$
of the form
$
f(x)=\min_{y\in D^k} f_\calI(x,y)
$ where $\calI$ is a $\Phi$-instance with $n+k$ variables. It is known that if $\Phi$ admits a fractional polymorphism $\omega$ then  $\langle\Phi\rangle$ also admits $\omega$.

\begin{definition}
Language $\Phi$ on domain $D$ is said to satisfy condition (MC) if there exist distinct $a,b\in D$
and binary function $f\in\langle\Phi\rangle$ such that $\argmin f=\{(a,b),(b,a)\}$.
\end{definition}

\begin{theorem}[{\cite[Theorem 3.4]{tz16:jacm},\cite[Theorem 5]{kolmogorov15:power}}]\label{th:MC}
Let $\Phi$ be a finite-valued language on domain $D$ such that for every $a\in D$
there exists a unary cost function $g_a\in\langle\Phi\rangle$ with $\argmin g_a=\{a\}$.
If $\Phi$ does not satisfy (MC) then $\Phi$ admits a symmetric fractional polymorphism of every arity $m\ge 2$.
\end{theorem}

\begin{theorem}[{\cite[Theorem 1, Proposition 8]{kolmogorov15:power}}]\label{th:BLP:characterization}
Let $\Phi$ be a general-valued language.
BLP solves $\Phi$ if and only if $\Phi$ admits a symmetric fractional polymorphism of every arity $m\ge 2$.
If this condition holds then an optimal solution of any $\Phi$-instance can be computed in polynomial time.
\end{theorem}

\subsection{Proof of Theorem~\ref{th:BLPsolvesGamma}}\label{sec:VCSP1}

Recall that $\Phi_\Gamma$ is the language over domain $D=V_\Gamma$ that consists of all functions $f:D^n\rightarrow\overline{\mathbb R}$
such that $f$ is $L$-convex on $\Gamma^n$.

\begin{theorem}\label{th:GNLAKSFLASF}
If $\Gamma$ is an extended modular complex then $\Phi_\Gamma$ does not satisfy (MC).
\end{theorem}
\begin{proof}
Suppose the claim is false, then there exists an instance $f:\Gamma\times\Gamma\times\Gamma^n\rightarrow\overline{\mathbb R}$ of $\Phi_\Gamma$
with $n+2\ge 2$ variables such that function $g:\Gamma\times\Gamma\rightarrow\overline{\mathbb R}$
defined via $g(x,y)=\min_{z\in\Gamma^{n}} f(x,y,z)$ satisfies $\argmin g=\{(a,b),(b,a)\}$ for some distinct $a,b\in\Gamma$.
Denote $A=\argmin f\subseteq \Gamma\times\Gamma\times\Gamma^n$, $A_{ab}=\{(a,b,z)\:|\:z\in\Gamma^n\}$ and $A_{ba}=\{(b,a,z)\:|\:z\in\Gamma^n\}$.
By construction, $A\subseteq A_{ab}\cup A_{ba}$, $A\cap A_{ab}\ne\varnothing$ and $A\cap A_{ba}\ne\varnothing$.
By Lemma~\ref{lemma:levelset:connectivity}, the subgraph of $(\Gamma\times\Gamma\times\Gamma^n)^{\mathsmaller\sqsubset}$ induced by set $A$ is connected.
Thus, there must exist $p\in A_{ab}$, $q\in A_{ba}$ such that $pq\in (\Gamma\times\Gamma\times\Gamma^n)^{\mathsmaller\sqsubset}$.
By symmetry, we can assume that $p\sqsubset q$.
We have $p=(a,b,x)$ and $q=(b,a,y)$ for some $x,y\in\Gamma^n$.
Condition $p\sqsubset q$ thus implies that $a\sqsubset b$ and $b\sqsubset a$, which is impossible.
\end{proof}

We can finally prove that BLP solves $\Phi_\Gamma$.

\begin{theorem}
Let $\Gamma$ be an extended modular complex.
Language $\Phi_\Gamma$ admits a symmetric fractional polymorphism of every arity $\mbox{$m\ge 2$}$.
Consequently, BLP relaxation solves $\Phi_\Gamma$, and an optimal solution of any $\Phi_\Gamma$-instance can be computed in polynomial time.
\end{theorem}
\begin{proof}
If language $\Phi_\Gamma$ were finite-valued then the claim would immediately follow
from Theorems~\ref{th:GNLAKSFLASF} and~\ref{th:MC}. The main concern will thus be dealing with non finite-valued languages.

Let us define $\Phi^\circ_\Gamma=\{f\in \langle\Phi_\Gamma\rangle\:|\:f\mbox{ is finite-valued}\:\}$.
Since language $\Phi_\Gamma$ does not satisfy (MC), languages $\langle\Phi_\Gamma\rangle$ and $\Phi^\circ_\Gamma$ also do not satisfy (MC).
Clearly, for each $a\in\Gamma$ language $\Phi^\circ_\Gamma$ contains unary function $g_a$ with $\argmin g_a=\{a\}$, namely $g_a(x)=\mu(a,x)$.
By Theorem~\ref{th:MC},  $\Phi^\circ_\Gamma$ admits a symmetric fractional polymorphism of every arity $m\ge 2$.
Let $\omega_m$ be such fractional polymorphism. We claim that $\Phi_\Gamma$ also admits $\omega_m$.
Indeed, consider function $f:\Gamma^n\rightarrow\overline{\mathbb R}$ in $\Phi_\Gamma$. 
We can assume w.l.o.g.\ that $f(x)\ge 0$ for all $x\in \dom f$, and $\mu(x,y)\ge 1$ for distinct $x,y\in \Gamma$.
Define function $\mu_{n}:\Gamma^n\times\Gamma^n\rightarrow\mathbb R$
via $\mu_n(x,y)=\mu_{\Gamma^n}(x,y)$. By Theorem~\ref{lemma:mu-is-Lconvex}, $\mu_n$ is $L$-convex on $\Gamma^n\times\Gamma^n$, and thus $\mu_n\in \Phi_\Gamma$.
For value $C>0$ define function $f_C:\Gamma^n\rightarrow\mathbb R$ via $f_C(x)=\min_{y\in\Gamma^n} (f(y)+C\mu_n(x,y))$,
then $f_C\in\Phi^\circ_\Gamma$ and thus $f_C$ admits $\omega_m$. 
Clearly, for any $C>\max_{x\in\dom f}f(x)$ the following holds: $f_C(x)=f(x)$ if $x\in\dom f$,
and $f_C(x)\ge C$ if $f(x)=\infty$.
By taking the limit $C\rightarrow \infty$ we conclude that $f$ also admits $\omega_m$.
\end{proof}

%%%%%%%%%%%%%%%%%%%%%%%%%%%%%%%%%%%%%%%%%%%%%%%%%%%%%%%%%%%%%%%%%%%%%%%%%%%%%%%%
%%%%%%%%%%%%%%%%%%%%%%%%%%%%%%%%%%%%%%%%%%%%%%%%%%%%%%%%%%%%%%%%%%%%%%%%%%%%%%%%
%%%%%%%%%%%%%%%%%%%%%%%%%%%%%%%%%%%%%%%%%%%%%%%%%%%%%%%%%%%%%%%%%%%%%%%%%%%%%%%%

\subsection{Proof of the hardness direction of Theorem~\ref{th:main}}\label{sec:proofs:NPhardness}

We will use a technique from~\cite{kz13:jacm}.

Consider language $\Phi$ on domain $D$. A pair of elements $(a,b)\in D\times D$ is called {\em conservative}
if there exists a unary function $g_{ab}\in \langle\Phi\rangle$ with $\argmin g_{ab}=\{a,b\}$.
Let $\calS(\Phi)\subseteq D\times D$ be the set of conservative pairs in $\Phi$.
For a tuple $p=(a,b)\in\calS(\Phi)$, we denote $\bar p=(b,a)$; clearly, $\bar p\in\calS(\Phi)$.
Now consider two tuples $p=(a,b)$ and $q=(c,d)$. We say that pair $(p,q)$ is {\em strictly submodular}
if there exists binary cost function $f\in\langle\Phi\rangle$ such that
$$
f(a,c)+f(b,d)<f(a,d)+f(b,c)
$$
%Such $f$ will be called a {\em witness} for $(p,q)$.
Clearly, if $(p,q)$ is strictly submodular, then $(q,p)$ is also strictly submodular, since function $f'$ defined via $f'(x,y)=f(y,x)$ also belongs to $\langle\Phi\rangle$.
We thus say that $\{p,q\}$ is strictly submodular if $(p,q)$ is strictly submodular (or equivalently, if $(q,p)$ is strictly submodular).
Let $\calE(\Phi)\subseteq\binom{\calS(\Phi)}2$ be the set of strictly submodular pairs $\{p,q\}$,
and define undirected graph $\calG(\Phi)=(\calS(\Phi),\calE(\Phi))$.

\begin{theorem}[{\cite[Theorem 3.2(a) and Lemma 5.1(b)]{kz13:jacm}}]\label{th:conservative}
Suppose that language $\Phi$ is finite-valued and that, for each $a\in D$, there exists unary cost function $g_a\in\langle\Phi\rangle$ with $\argmin g_a=\{a\}$.
\begin{itemize}[topsep=1mm]
\setlength\itemsep{0.01mm}
\item[{\rm (a)}] If $\{p,q\},\{q,r\}\in\calE(\Phi)$, then $\{p,r\}\in\calE(\Phi)$.
\item[{\rm (b)}] If $\{p,\bar p\}\in\calE(\Phi)$ for some $p\in\calS(\Phi)$, then $\Phi$ is NP-hard.
\end{itemize}
\end{theorem}

\begin{remark}
In~\cite{kz13:jacm} Kolmogorov and \v{Z}ivn\'y
 formulated Theorem~\ref{th:conservative} only in the case
when $\Phi$ contains all possible $\{0,1\}$-valued unary functions (and thus $\calS(\Phi)=D\times D$).
However, the proofs of the results above use only weaker preconditions stated in Theorem~\ref{th:conservative}.

Note that the graph defined in~\cite{kz13:jacm} had an edge $\{p,q\}$ if and only if our graph $\calG(\Phi)$ has an edge~$\{p,\bar q\}$.
We translated the results from~\cite{kz13:jacm} accordingly.
\end{remark}

We now apply these results to the generalized minimum $0$-extension problem
for metric space $(V,\mu)$ and subset $F\subseteq\binom{V}2$.
Let us define the following language over domain $D=V$:
$$
\Phi \;=\; \{\mu\} \;\cup\; \{\mu_a\::\:a\in V\} \;\cup\; \{\mu_{ab}\::\:\{a,b\}\in F\}
$$
where unary function $\mu_a$, $\mu_{ab}$ are defined via $\mu_a(x)=\mu(x,a)$ and $\mu_{ab}(x)=\min\{\mu(x,a),\mu(x,b)\}$.
Clearly, we have $\Phi\subseteq\langle\ZeroExt{\mu,F}\rangle$.

Let $H=H_\mu=(V,E,w)$ be the graph corresponding to $\mu$. Suppose that $H$ is not $F$-orientable; our goal is to show that $\Phi$ is NP-hard.
We can assume w.l.o.g.\ that $H$ is modular, otherwise $\Phi$ is NP-hard by Theorem~\ref{th:Karzanov:NP-hard}.
Define 
$$
\vec E \,=\, \{(a,b)\::\:\{a,b\}\in E\} \qquad\quad
\vec F \,=\, \{(a,b)\::\:\{a,b\}\in F\}
$$
Let us introduce relations $\parallel$, $\lhd$, $\approx$ for tuples $p=(a,b)$ and $q=(c,d)$ as follows:
\begin{itemize}[topsep=1mm]
\setlength\itemsep{0.01mm}
\item  $p\parallel q$ \;if $p,q\in\vec E$ and $(a,b,d,c)$ is a 4-cycle in $H$;
\item  $p\lhd q$ \:if $p\in\vec E$, $q\in\vec F$, $p\ne q$, and $(c,a,b,d)$ is a shortest subpath in $H$;
\item  $p\approx q$ \,if at least one of the following holds: (i) $p\parallel q$; (ii) $p\lhd q$; (iii) $q\lhd p$.
\end{itemize} \smallskip
Note that relation $\approx$ is symmetric, and accordingly, $(\vec E\cup\vec F,\,\approx)$ is an undirected graph.
\begin{lemma}\label{lemma:asdfgasdg}
Let $p,q \in \calS(\Phi)$.
\begin{itemize}[topsep=1mm]
\setlength\itemsep{0.01mm}
\item[{\rm (a)}] $\vec E\cup\vec F\subseteq \calS(\Phi)$.
\item[{\rm (b)}] If $p\parallel q$, then $\{p,q\}\in\calE(\Phi)$.
\item[{\rm (c)}] If $p\lhd q$, then $\{p,q\}\in\calE(\Phi)$.
\item[{\rm (d)}] If $p\approx q$, then $\{p,q\}\in\calE(\Phi)$.
\end{itemize}
\end{lemma}
\begin{proof}
\myparagraph{(a)} 
First, $\vec F\subseteq \calS(\Phi)$ holds by the definition of $\Phi$. 

As for $\vec E$, for each $\{a,b\}\in E$, we construct a function $g_{ab} \in \langle\Phi\rangle$ by $g_{ab}(x) := \mu_a(x)+\mu_b(x)$. From the minimality of $H$, we have $I(a,b) = \{a,b\}$. We therefore see that $\argmin g_{ab}=\{a,b\}$, and thus, we obtain $(a,b)\in\calS(\Phi)$.

\myparagraph{(b)} Suppose that $p=(a,b)\in\vec E$, $q=(c,d)\in\vec E$, and $(a,b,d,c)$ is a 4-cycle in $H$. The modularity of $H$ implies that  $\{a,d\}\notin E$; otherwise $a,b,d$ would not have a median. 
Similarly, $\{b,c\}\notin E$; otherwise $b,d,c$ would not have a median.
By Theorem~\ref{th:orbits}(a), sequences $(a,b,d)$ and $(b,d,c)$ are shortest subpaths in $H$. This
implies that
$
\mu(a, c)+\mu(b, d) < \mu(a, d)+\mu(b, c)
$. We thereby obtain $\{p,q\}\in\calE(\Phi)$.

\myparagraph{(c)}
Suppose that $p=(a,b)\in\vec E$, $q=(c,d)\in\vec F$, and $(c,a,b,d)$ is a shortest subpath in $H$.
Then $\mu(a, c)+\mu(b, d) < (\mu(a, c)+\mu(a,b)) + (\mu(a,b)+\mu(b, d))=
\mu(b,c) + \mu(a, d)$, and so, we again obtain $\{p,q\}\in\calE(\Phi)$.

\myparagraph{(d)} We decompose $p\approx q$ into cases. If $p\parallel q$, we apply (b). If $p\lhd q$, we apply (c). If $q\lhd p$, we utilize $\{p,q\} = \{q,p\}$ and apply (c) again. We conclude that $\{p,q\}\in\calE(\Phi)$.
\end{proof} \smallskip

\noindent We can finally finish the proof of Theorem~\ref{th:main}.

We claim that there exists an element $p\in\vec E\cup \vec F$ such that $p,\bar p$ are connected by a path in 
$(\vec E\cup\vec F,\,\approx)$. Indeed, if no such $p$ exists, then there exists a mapping $\sigma:\vec E\cup \vec F\rightarrow\{-1,+1\}$
with the following properties: (i) if $q\approx r$, then $\sigma(q)=\sigma(r)$; (ii) $\sigma(q)=-\sigma(\bar q)$ for each $q\in\vec E\cup \vec F$.
Such mapping can be constructed by greedily assigning connected components of graph $(\vec E\cup\vec F,\.\approx)$;
the assumption ensures that no conflicts arise.
Using the mapping $\sigma$, we can define an orientation of $(H,F)$ by orienting $\{a,b\} \in E \cup F$ as $\.a \!\.\rightarrow\.\! b\.$ if $\sigma(a,b)=+1$ and as $\.a \!\.\leftarrow\.\! b\.$ if $\sigma(a,b)=-1$.
Clearly, this orientation is admissible, which contradicts the assumption that $H$ is not $F$-orientable.

From Theorem~\ref{th:conservative}(a) and Lemma~\ref{lemma:asdfgasdg}, we conclude that $\{p,\bar p\}\in\calE(\Phi)$.
Theorem~\ref{th:conservative}(b) now gives that $\Phi$ is NP-hard.
%%%%%%%%%%%%%%%%%%%%%%%%%%%%%%%%%%%%%%%%%%%%%%%%%%%%%%%%%%%%%%%%%%%%%%%%%%%%%%%%
%%%%%%%%%%%%%%%%%%%%%%%%%%%%%%%%%%%%%%%%%%%%%%%%%%%%%%%%%%%%%%%%%%%%%%%%%%%%%%%%
%%%%%%%%%%%%%%%%%%%%%%%%%%%%%%%%%%%%%%%%%%%%%%%%%%%%%%%%%%%%%%%%%%%%%%%%%%%%%%%%

\section{Open questions}\label{sec:open}
In this section we state some open problems that we find interesting.

One of the motivations for this work was trying to find the ``structure''
of tractable finite-valued languages. This leads to the following question.
\begin{itemize}
\item Can it be case that for every tractable finite-valued language $\Phi$
there exists an extended modular complex $\Gamma$ such all functions $f\in\Phi$
are $L$-convex on $\Gamma$?

A negative answer could yield new interesting classes of functions that are not yet known in the area of Discrete Convex Analysis.
\end{itemize}
To simplify the problem, one may a consider a restricted class of languages.

\begin{itemize}

\item Resolve the question above for (a) languages of the form $\ZeroExt{\mu,F}$ where $\mu$ is a metric on $V$ and $F$ is a subset of $2^V$;
(b) languages corresponding to directed metrics, as in~\cite{HiraiMizutani}.

Note that~\cite{HiraiMizutani} characterized tractable directed metrics on a star,
with some complicated fractional polymorphisms. 
As a special case, it would be interesting to verify whether these tractable classes
are captured by the theory developed in this paper. (This question was suggested by an anonymous reviewer).

\end{itemize}
There are also some open problems concerning $L$-convex functions on extended modular complexes.
\begin{itemize}

\item
Is it possible to show that the normal SDA algorithm terminates after a polynomial number of iterations?
\item
As shown in~\cite{Hirai:0ext,Hirai:Lconvexity}, if $f$ an an $L$-convex function on a modular complex $\Gamma$ then $f^\ast$ is $L$-convex on $\Gamma^\ast$.
Is this also true for extended modular complexes?
\item
 Let $\sqsubseteq,\sqsubseteq'$ be admissible relations for modular complex $\Gamma$
such that $\sqsubseteq$ coarsens $\sqsubseteq'$. % e.g.\ $(\sqsubseteq,\sqsubseteq')=(\BPrel,\preceq)$.
Is it true that if function $f$ is $L$-convex on $(\Gamma,\sqsubseteq)$ then it is also 
$L$-convex on $(\Gamma,\sqsubseteq')$? This is known to hold when $\Gamma$ is an oriented path
and $(\sqsubseteq,\sqsubseteq')=(\BPrel,\preceq)$:
any $L^\natural$-convex function is submodular on an integer lattice.
However, we were unable to prove or disprove it in the general case.
\item Theorem~\ref{th:BLPsolvesGamma} shows that $L$-convex functions on extended modular complexes
admit a binary symmetric fractional polymorphism, however the proof of this fact is non-constructive.
Is there a more explicit construction of this polymorphism?
As remarked in Section 6 of~\cite{Hirai:0ext}, answering this question might require a further
thorough investigation of orientable modular graphs.
\end{itemize}

\section*{Acknolwedgements} 
We thank the anonymous reviewers for their careful reading of our manuscript and their many insightful comments and suggestions. 

\bibliographystyle{plain}
\bibliography{VCSP}

\end{document}